\documentclass[a4paper]{article}

\usepackage{a4wide}
\usepackage[UKenglish]{babel}
\usepackage{graphicx,subfigure}
    \graphicspath{{./}{figs/}}
\usepackage[colorlinks=true,linkcolor=blue,citecolor=red]{hyperref}
\usepackage{xcolor}
\usepackage{amssymb,amsthm,empheq,amsmath,amsfonts,bbold,mathrsfs}
\usepackage[inline]{enumitem}
\usepackage{cases}
\usepackage{authblk}

\newtheorem{theorem}{Theorem}[section]
\theoremstyle{remark}\newtheorem{remark}[theorem]{Remark}
\theoremstyle{definition}\newtheorem{definition}[theorem]{Definition}

\newcommand{\abs}[1]{\left\lvert#1\right\rvert}
\newcommand{\ave}[1]{\left\langle#1\right\rangle}
\newcommand{\cB}{\mathcal{B}}
\newcommand{\sB}{\mathscr{B}}
\DeclareMathOperator{\ddiv}{div}
\newcommand{\cH}{\mathcal{H}}
\newcommand{\Lip}{\operatorname{Lip}}
\newcommand{\loc}{\mathrm{loc}}
\newcommand{\norm}[1]{\left\Vert#1\right\Vert}
\newcommand{\cO}{\mathcal{O}}

\newcommand{\cP}{\mathcal{P}}
\newcommand{\pr}[1]{{}^\prime\!#1}
\newcommand{\R}{\mathbb{R}}
\newcommand{\cR}{\mathcal{R}}
\newcommand{\gR}{\mathfrak{R}}
\newcommand{\bV}{\mathbf{V}}
\newcommand{\Var}{\operatorname{Var}}
\newcommand{\bX}{\mathbf{X}}

\allowdisplaybreaks

\begin{document}
\title{Kinetic description and macroscopic limit of swarming dynamics with continuous leader-follower transitions}
\author[1]{Emiliano Cristiani\thanks{\texttt{e.cristiani@iac.cnr.it}}}
\author[2]{Nadia Loy\thanks{\texttt{nadia.loy@polito.it}}}
\author[3]{Marta Menci\thanks{\texttt{m.menci@unicampus.it}}}
\author[2]{Andrea Tosin\thanks{\texttt{andrea.tosin@polito.it}}}
\affil[1]{Istituto per le Applicazioni del Calcolo, Consiglio Nazionale delle Ricerche, Rome, Italy}
\affil[2]{Department of Mathematical Sciences ``G. L. Lagrange'', Politecnico di Torino, Italy}
\affil[3]{Universit\`{a} Campus Bio-Medico di Roma, Rome, Italy}
\renewcommand\Affilfont{\small}
\date{}

\maketitle

\begin{abstract}
In this paper, we derive a kinetic description of swarming particle dynamics in an interacting multi-agent system featuring emerging leaders and followers. Agents are classically characterized by their position and velocity plus a continuous parameter quantifying their degree of leadership. The microscopic processes ruling the change of velocity and degree of leadership are independent, non--conservative and non--local in the physical space, so as to account for long--range interactions. Out of the kinetic description, we obtain then a macroscopic model under a hydrodynamic limit reminiscent of that used to tackle the hydrodynamics of weakly dissipative granular gases, thus relying in particular on a regime of small non--conservative and short--range interactions. Numerical simulations in one- and two-dimensional domains show that the limiting macroscopic model is consistent with the original particle dynamics and furthermore can reproduce classical emerging patterns typically observed in swarms.

\medskip

\noindent{\bf Keywords:} swarm dynamics, transient leadership, Povzner-Boltzmann equation, hydrodynamic limit, non-conservative interactions

\medskip

\noindent{\bf Mathematics Subject Classification:} 35Q20, 35Q70, 92D50.
\end{abstract}

\section{Introduction}
In this paper, we are concerned with kinetic and hydrodynamic modelling of swarming dynamics in self-organising multi-agent systems with group behaviour and transient leadership, biologically inspired by flocks of birds. 

A first introduction to the collective motion of groups of animals can be found in~\cite{sumpter2006,vicsek2012}. Concerning flocking birds, many of the underlying mechanisms have been unveiled by several experimental studies, mostly from a physical perspective, based on tracking and analysing the trajectories of flying birds~\cite{attanasi2014,attanasi2015,ballerini2008b,ballerini2008a,cavagna2010,cavagna2022N,cavagna2013,loffredo2023PhB,procaccini2011} and by other experiments made within similar goals~\cite{ling2019,papadopoulou2022RSOS,pomeroy1992,reynolds2022JRSI}.  These investigations have put in evidence that birds interact within a certain perception range, tend to align with the flight direction of neighbours and rely on a variable-in-time interaction network, as each bird changes its neighbours continuously while moving~\cite{cavagna2013}. In particular, the two already cited papers~\cite{attanasi2015,ling2019} focus on coherent changes in the direction of motion of the group as a whole. They find that collective turns have a localised spatial origin: they are triggered by few members of the flock, that act locally as \textit{leaders}, and happen spontaneously, i.e. without external stimuli (such as e.g., a predator attack). Another important biological mechanism observed in several animal groups, including flocks of birds, is the \textit{switch of leaders}~\cite{butail2019detecting,chen2016switching,mwaffo2018detecting}: since all group members may be leaders temporarily, leaders change continuously, with former leaders becoming followers and then again leaders after a while.  The leader-follower transition is often difficult to understand, because in some flocks of birds, such as e.g., pigeons, it is associated with a \textit{continuous} degree of leadership, which determines a \textit{hierarchy of leaders} based e.g., on the ability of birds to act as turning initiators~\cite{leader1,leader2}.

As far as the mathematical modelling and numerical simulation of swarming dynamics are concerned, we mention first the three seminal papers~\cite{couzin2005,CS2007,vicsek1995}, which have shown how simple algorithms based on elementary interaction rules can generate apparently complex self-organised collective movements. Following these works, many other microscopic models including leaders have been proposed, see e.g.,~\cite{albi2016SIAP,bongini2014NHM,borzi2015M3AS,motsch2011JSP}, sometimes inspired by analogies with fundamental physical principles such as e.g., the inertial spin models~\cite{caglioti2020,cavagna2015,markou2021pp}. Microscopic models allow for exhaustive simulations when the total number of agents, say $N$, is large enough, for then statistical average quantities may be regarded as consistent descriptors of the aggregate behaviour. Nevertheless, one experiences the well-known drawback that the computational cost of a microscopic simulation scales typically as $\cO(N)$ or $\cO(N^2)$, depending on the specific microscopic model, whereas the cost of a macroscopic simulation would be independent of $N$. Moreover, microscopic models are hardly amenable to analytical investigations of the emerging aggregate behaviour of the particle system.

For these reasons, several mesoscopic and macroscopic models have also been proposed and studied~\cite{BolleyM3AS2011,MihaM3AS2020,CarrilloM3AS2011,Carrillo2014,CarrilloJEMS2019,Carrillo2017,CarrilloKRM2009, Carrillo2010,CarrilloM3AS,eftimie2018,fornasier2010PhD}. Among them, we focus in particular on kinetic models, i.e. mesoscopic models which can be formally derived from interacting particle dynamics as evolution equations of the time-dependent probability distribution of the microscopic state of the interacting particles. Particle dynamics may be stated in terms of either ordinary or stochastic differential equations. Kinetic equations are then derived either as mean-field limits of the particle dynamics for $N\to\infty$, see e.g.,~\cite{Peurichard2021}, or as statistical descriptions of the stochastic microscopic interaction processes~\cite{pareschi2013book}.

In this paper, we adopt the point of view of the classical kinetic theory, having its roots in gas dynamics and nowadays fruitfully revisited and applied to the mathematical description of a great variety of interaction phenomena, cf.~\cite{pareschi2013book} and references therein. The leading idea is that agents are identified to a microscopic state, which changes in consequence of interactions and is described statistically by means of a probability distribution defined on an appropriate continuous phase space. Typically, interactions are assumed to be binary processes, i.e. each of them involves two agents at a time. Interactions involving simultaneously more than two agents are neglected as higher order effects. Kinetic models for collective dynamics incorporating leaders and followers have been introduced in~\cite{albi2022AMO,albi2016SIAP,albi2019M3AS,albi2013AML,bernardi2021}. In this case, agents are described by their position, velocity and a discrete label denoting either leadership or followership, which may change in time. Kinetic equations with label switching have been formally analysed in~\cite{loy2021KRM} and applied to the modelling of collective dynamics in the presence of switching leaders in~\cite{albi2023papercugino}. Kinetic models also allow one to derive macroscopic models, which inherit naturally the details of the microscopic particle dynamics. This is typically carried out by introducing the statistical moments of the probability distribution of the microscopic state as average macroscopic quantities. Their evolution equations are then obtained either from a hydrodynamic limit or through asymptotic expansions (Hilbert, Chapman-Enskog), possibly relying on innovative concepts, such as that of \textit{generalised collisional invariants}, when the classical collisional invariants of the kinetic equations are not as many as required to parametrise the local equilibrium distribution, cf.~\cite{degond2007CRM,degond2008M3AS}. These procedures typically lead to closed systems of evolution equations for the average macroscopic quantities in specific regimes of the microscopic parameters~\cite{Peurichard2021,Carrillo2010,fornasier2010PhD}.

In the present work, our starting point is the microscopic model proposed in~\cite{cristiani2021JMB}, which describes the dynamics of turning and flocking birds, like starlings or jackdaws. The model is based on the well-established attraction-alignment-repulsion paradigm~\cite{couzin2002,eftimie2018,sumpter2006,vicsek2012} and is specifically conceived to reproduce spontaneous (i.e., not caused by external stimuli) sudden self-organised changes of direction of motion. Moreover, the model assumes that each bird is a potential turn initiator, meaning that a sort of temporary leadership emerges in a group of indistinguishable birds. Retaining the fundamental ingredients of this model, we focus on an analytical description based on kinetic equations and their macroscopic limits. More specifically, we define a multi-agent system whose agents are characterised by space position and velocity plus a continuous parameter quantifying their \textit{degree of leadership}. We formulate discrete-in-time stochastic particle dynamics inspired by a simplified version of the microscopic model introduced in~\cite{cristiani2021JMB} and featuring non-local binary interactions, so as to take into account that agents may affect each other possibly within an interaction range. The microscopic processes ruling the change of velocity and degree of leadership are independent and both non-conservative, except for the mass. Therefore, to derive formally kinetic and macroscopic descriptions implementing such microscopic dynamics we resort to a Povzner-Boltzmann-type equation~\cite{povzner1962AMSTS} with corrections to the collisional operator suitable to take into account the non-conservativeness of the microscopic interactions when passing to the hydrodynamic limit. For this, we adapt the technique proposed in~\cite{fornasier2010PhD}, where the authors perform the hydrodynamic limit of the Povzner-Boltzmann equation for a system of agents having the velocity as a unique microscopic state. There, the binary interaction rules are non-local and non-conservative as well and, in particular, they are given by the classical conservative hard-spheres collisions plus a dissipative contribution. The key idea is to consider the regime in which the non-conservative contribution is small, in such a way that the leading part of the interaction rules is given by the classical velocity- and energy-preserving hard-spheres mechanism, while the dissipative part is of higher order. This allows the authors to approximate the non-local Maxwellian relative to the leading order terms of the interaction rules with the classical local Maxwellian. In our case, the main difficulty is that this idea has to be evolved to treat multiple microscopic states with different non-mechanical non-conservativeness.

The plan of the paper, after this Introduction, is the following. In Section~\ref{sect:micro.summary} we revise and summarise the main aspects of the microscopic models presented in~\cite{cristiani2021JMB}. Then in Section~\ref{sect:kin} we introduce our discrete-in-time stochastic particle model implementing the interaction dynamics discussed above and derive formally the corresponding Povzner-Boltzmann kinetic equation. In Section~\ref{sect:macro} we derive our macroscopic model of swarming dynamics with transient leadership and continuous leader-follower transitions as a hydrodynamic limit of the Povzner-Boltzmann equation. Next, in Section~\ref{sect:numerics} we discuss some numerical tests, which we use to elucidate the correspondence between the original stochastic particle model and the resulting macroscopic description (in one space dimension) as well as to show that the macroscopic model reproduces successfully (in two space dimensions) some classical swarming patterns, such as splitting and merging of the flock and rings/mills. Finally, in Section~\ref{sect:conclusions} we draw some conclusions and briefly sketch possible research perspectives.

\section{Microscopic outlook}
\label{sect:micro.summary}
We consider a group of $N>1$ agents represented as point masses in $\R^d$, $d\geq 1$. We assume that they are labelled univocally by an index $k=1,\,\dots,\,N$ and we denote by $X_k(t),\,V_k(t)\in\R^d$ the position and velocity of the $k$-th agent at time $t$, respectively. We also distinguish agents in terms of \emph{leaders} and \emph{followers}. This distinction is not permanent, since agents can change status in time. Leaders do not have any particular knowledge or experience compared to other agents, their role is just to initiate turns. To this aim, they try to get far from the group -- the \textit{flock} from now on -- at least until they get back to the status of followers. Roughly speaking, the idea is that any agent can attempt to steer the flock by initiating a turn, but the attempt has a limited duration and if turning is not triggered within a reasonable time the attempt is aborted. Here we consider the leader/follower status as a spectrum, meaning that agents exhibit a possibly non-binary behavior between the statuses of leader and follower. We denote by $\Lambda_k=\Lambda_k(t)\in [0,\,1]$ the \textit{degree of leadership} of agent $k$ at time $t$, being $\Lambda_k=0$ the pure follower behavior and $\Lambda_k=1$ the pure leader behavior.

Concerning interactions, we assume that agents are influenced by flock mates within a certain fixed distance $R>0$ only (metric interaction). It is worth noting that in the microscopic model introduced in~\cite{cristiani2021JMB}, unlike the one studied here, interactions are topological \cite{ballerini2008a} and the distinction between leaders and followers is sharp, in the sense that at any given time an agent is either a leader or a follower. We refer to~\cite{albi2023papercugino} for a kinetic model which retains these features.

In the same spirit as in \cite{cristiani2021JMB}, we consider the following second order dynamics:
\begin{equation}
    \ddot X_k=\mu\mathcal A_k(\bX,\bV,\Lambda_k), \qquad k=1,\,\dots,\,N,
    \label{ODEmicro_1}
\end{equation}
where $\bX:=(X_1,\,\dots,\,X_N)$, $\bV:=(V_1,\,\dots,\,V_N)$, $\mu>0$ is a parameter fixing the time scale and the acceleration $\mathcal A_k$ is given by the sum of three terms accounting respectively for repulsion, alignment, and attraction among the agents:
\begin{equation}
    \mathcal A_k(\bX,\bV,\Lambda_k):=\sum_{X_h\in B_R(X_k)}\left[
        \alpha\frac{X_k-X_h}{\abs{X_k-X_h}^2}+(1-\Lambda_k)\Bigl(
        \beta(V_h-V_k)+
        \gamma(X_h-X_k)\Bigr)\right].
    \label{ODEmicro_2}
\end{equation}
Here, $B_R(X_k)\subset\R^d$ is the ball of radius $R>0$ centered in $X_k$ (metric neighborhood) and $\alpha,\,\beta,\,\gamma>0$ are parameters calibrating the mutual strength of the three types of interaction. We notice that, unlike other approaches in the reference literature, here we choose for simplicity to consider all types of interaction simultaneously active within the interaction neighborhood $B_R(X_k)$. We will see in the numerical tests that this simplification does not preclude the possibility to reproduce well-known self-organized patterns of the flock. We also point out that, while repulsion is common to both leaders and followers, alignment and attraction characterize preferentially the interaction dynamics of the followers, cf. the coefficient $1-\Lambda_k$ in front of the last two terms in~\eqref{ODEmicro_2}.

While the repulsion-alignment-attraction paradigm is nowadays standard and well accepted in the modeling of flock dynamics, the evolution of $\Lambda_k$, driving the leader-follower transition, is much less investigated. In particular, we do not have any biological clue for this term. Therefore, we appeal to the general idea that the $\Lambda_k$'s are independent and identically distributed random variables supported in $[0,\,1]$ and sampled at random instants. This way, agents continuously switch their (non-binary) status between leader, follower and intermediate nuances.

\section{Kinetic approach}\label{sect:kin}
As anticipated in the Introduction, the goal of this paper is to derive a hydrodynamic model of flock dynamics inspired by the ideas outlined in Section~\ref{sect:micro.summary} and incorporating, in particular, the leader-follower transitions. The novelty with respect to the mainstream in the literature is that we do not rely on separate descriptions of leaders and followers with possible mass transfers between the two populations. Instead, we give a unique fluid dynamical representation of the flock, which embeds consistently the gross evolution of the degree of leadership of the agents.

To pursue this goal, we take advantage of the methods of statistical mechanics in general and of the kinetic theory in particular, which provide powerful mathematical tools to upscale consistently simple particle dynamics to complex aggregate macroscopic trends.

\subsection{Stochastic particle model}
\label{sect:stoch_part_mod}
We assume that the flock can be regarded as a large ensemble of indistinguishable agents. A generic representative agent is described at time $t>0$ by the position in space $X_t\in\R^d$, the velocity $V_t\in\R^d$ and degree of leadership $\Lambda_t\in [0,\,1]$. To adopt a statistical point of view, we regard $X_t$, $V_t$, $\Lambda_t$ as random variables with joint distribution
\begin{equation}
    f=f(x,v,\lambda,t):\R^d\times\R^d\times [0,\,1]\times [0,\,+\infty)\to\R_+
    \label{eq:f}
\end{equation}
such that $f(x,v,\lambda,t)\,dx\,dv\,d\lambda$ is the probability that at time $t$ a generic agent of the flock has a position in the infinitesimal volume element $dx$ centered in $x$, a velocity in the infinitesimal volume element $dv$ centered in $v$ and a degree of leadership comprised between $\lambda$ and $\lambda+d\lambda$.

The evolution of the stochastic processes $\{X_t,\,t\geq 0\}$, $\{V_t,\,t\geq 0\}$, $\{\Lambda_t,\,t\geq 0\}$ is motivated by the considerations expressed in Section~\ref{sect:micro.summary}. In particular, referring to a discrete time setting for simplicity but with the idea to pass subsequently to the continuous time limit, we invoke the usual kinematic relationship between $X_t$ and $V_t$ to state that at the new time $t+\Delta{t}$, with $\Delta{t}>0$ small, a generic agent of the flock changes position because of a free drift:
\begin{equation}
    X_{t+\Delta{t}}=X_t+V_t\Delta{t}.
    \label{def:micro_drift}
\end{equation}
On the other hand, the new velocity of the agent is determined by an interaction process inspired by~\eqref{ODEmicro_1}-\eqref{ODEmicro_2} with another agent sampled randomly in the flock, as prescribed by the scheme of binary interactions of the classical kinetic theory. If $(X_t^\ast,\,V_t^\ast,\,\Lambda_t^\ast)$ is the microscopic state of the other agent then we set
\begin{equation}
    V_{t+\Delta{t}}=V_t+\Theta\mathcal{A}(X_t,X_t^\ast,V_t,V_t^\ast,\Lambda_t,\Lambda_t^\ast),
    \label{def:micro_v}
\end{equation}
where
\begin{equation}
    \mathcal{A}(X_t,X_t^\ast,V_t,V_t^\ast,\Lambda_t,\Lambda_t^\ast)=\alpha\dfrac{X_t-X_t^\ast}{\abs{X_t-X_t^\ast}^2}
        +(1-\Lambda_t)\Bigl(\beta(V_t^\ast-V_t)+\gamma(X_t^\ast-X_t)\Bigr),
    \label{def:micro_v_new}
\end{equation}
while $\Theta\in\{0,\,1\}$ is a random variable discriminating whether, in the time step $\Delta{t}$, an interaction between the two sampled agents may actually take place. Specifically, we set
\begin{equation}
    \Theta\sim\operatorname{Bernoulli}(\mu B(\abs{X_t^\ast-X_t})\Delta{t}),
    \label{def:Theta}
\end{equation}
i.e. $\mu B(\abs{X_t^\ast-X_t})\Delta{t}$ is the probability that an interaction between the two sampled agents takes place in the time interval $[t,\,t+\Delta{t}]$. In more detail, $\mu>0$ is the rate of velocity change while $B(\abs{X_t^\ast-X_t})\geq 0$ is the \textit{collision\footnote{Here and henceforth, borrowing the jargon of the classical kinetic theory, by `collision' we mean the same as `interaction'.} kernel}, which fixes the rate at which agents at a given distance $\abs{X_t^\ast-X_t}$ may interact. We assume that $B=B(r)$, $r\in\R_+$, is a non-negative bounded function supported in the interval $[0,\,R]\subset\R_+$ for a given radius $R\in (0,\,+\infty)$. This reproduces the idea of metric interactions within the interaction range $R$. Consequently, the mapping $\R^d\ni x\mapsto B(\abs{x})$ is radially symmetric and supported in the ball $\sB_R(0)\subset\R^d$ of radius $R$ and centred at the origin.

Finally, we postulate the following update process of the degree of leadership:
\begin{equation}
    \Lambda_{t+\Delta t}=\Lambda_t+\Sigma(1-2\Lambda_t),
    \label{def:micro_l}
\end{equation}
where $\Sigma\in\{0,\,1\}$ is a new random variable, independent of $\Theta$, distributed as
\begin{equation}
    \Sigma\sim\operatorname{Bernoulli}(\eta B(\abs{X_t^\ast-X_t})\Delta{t})
    \label{def:Sigma}
\end{equation}
with $\eta>0$ the rate of degree change. Notice that~\eqref{def:Sigma} implies that a change in the degree of leadership of one agent is stimulated by the presence of the other agent in the interaction neighborhood. Nevertheless, according to~\eqref{def:micro_l} the new degree of leadership of either agent does not depend on the current degree of leadership of the other agent. In other words, we assume that agents become leaders or followers spontaneously and not e.g., by imitation or contrast. We also remark that~\eqref{def:micro_l} is such that
$$ \Lambda_{t+\Delta t}=
    \begin{cases}
        \Lambda_t & \text{if } \Sigma=0 \\
        1-\Lambda_t & \text{if } \Sigma=1,
    \end{cases} $$
meaning that any agent undergoes continuously a switch of its degree of leadership between the initial value $\Lambda_0$ and the symmetric one $1-\Lambda_0$ with respect to the mid-value $\frac{1}{2}$.

For the sake of completeness, we point out that~\eqref{def:Theta},~\eqref{def:Sigma} require
$$ \Delta{t}\leq\frac{1}{\max{\{\mu,\,\eta\}}\sup\limits_{x\in\sB_R(0)}B(\abs{x})} $$
as a sufficient consistency condition. However, such a constraint is not a limitation in the continuous time limit $\Delta{t}\to 0^+$.

\begin{remark}
The microscopic model in~\cite{cristiani2021JMB}, here recalled in Section~\ref{sect:micro.summary}, is conceived in terms of \textit{mean-field-like} dynamics among the agents within their interaction neighbourhoods. Therefore, in that setting a single agent is affected simultaneously by all other agents in its neighbourhood, as~\eqref{ODEmicro_2} implies. Conversely, in this paper we adopt a \textit{binary} interaction scheme, in order to recast the modelled swarm dynamics in the frame of the collisional kinetic theory. The latter is indeed based on \textit{pairwise} interactions among the agents, whence collective dynamics emerge over time. Such binary interactions are obtained from the mean-field-like interactions~\eqref{ODEmicro_2} by assuming that each agent interacts, from time to time, with just one other agent randomly selected within its interaction neighborhood. Therefore, a binary interaction coincides structurally with the generic single term of the mean-field-like interactions, cf.~\eqref{ODEmicro_2} and~\eqref{def:micro_v_new}. This approach is reminiscent of the one proposed e.g., in~\cite{carrillo2010SIMA} and is known to be consistent with mean-field-like dynamics on a properly large time scale (precisely, the time scale where the collective behaviour of the swarm emerges).
\end{remark}

\subsection{Boltzmann-type description}
A Boltzmann-type description of the flocking particle dynamics discussed in Section~\ref{sect:stoch_part_mod} consists in a collisional kinetic equation for the distribution function~\eqref{eq:f}, which, being regarded as a probability distribution, is assumed to be non-negative and such that
$$ \int_0^1\int_{\R^d}\int_{\R^d}f(x,v,\lambda,t)\,dx\,dv\,d\lambda=1, \qquad \forall\,t>0. $$

Averaging the discrete time stochastic processes~\eqref{def:micro_drift},~\eqref{def:micro_v},~\eqref{def:micro_l} and passing to the limit $\Delta{t}\to 0^+$, by standard arguments (cf. e.g.,~\cite{pareschi2013book}) we obtain the following Boltzmann-type kinetic equation for $f$:
\begin{equation}
    \partial_tf+v\cdot\nabla_xf=\mu Q^v(f,f)+\eta Q^\lambda(f,f),
    \label{eq:Boltzmann_strong}
\end{equation}
where $Q^v$ and $Q^\lambda$ are the \textit{collisional operators} implementing, at the kinetic level, the independent particle processes~\eqref{def:micro_v},~\eqref{def:micro_l}, respectively. More specifically, these two operators are univocally defined by the following weak formulations:
\begin{align}
    \begin{aligned}[b]
    \ave{Q^v(f,f),\,\varphi} &:= \int_0^1\int_{\R^d}Q^v(f,f)(x,v,\lambda,t)\varphi(v,\lambda)\,dv\,d\lambda \\
    &= \int_{[0,\,1]^2}\int_{\R^{3d}}B(\abs{x_\ast-x})\big(\varphi(v',\lambda)-\varphi(v,\lambda)\big)ff_\ast\,dx_\ast\,dv\,dv_\ast\,d\lambda\,d\lambda_\ast,
    \end{aligned}
    \label{def:Qv_weak}
\end{align}
where, here and henceforth, $f$ and $f_\ast$ are shorthand for $f(x,v,\lambda,t)$ and $f(x_\ast,v_\ast,\lambda_\ast,t)$, respectively, and
\begin{align}
    \begin{aligned}[b]
        \ave{Q^\lambda(f,f),\,\varphi} &:= \int_0^1\int_{\R^d}Q^\lambda(f,f)(x,v,\lambda,t)\varphi(v,\lambda)\,dv\,d\lambda \\
        &= \int_{[0,\,1]^2}\int_{\R^{3d}}B(\abs{x_\ast-x})\big(\varphi(v,\lambda')-\varphi(v,\lambda)\big)ff_\ast\,dx_\ast\,dv\,dv_\ast\,d\lambda\,d\lambda_\ast,
    \end{aligned}
    \label{def:Ql_weak}
\end{align}
where $\varphi:\R^d\times [0,\,1]\to\R$, $\varphi=\varphi(v,\lambda)$, is an arbitrary test function representing any observable quantity whereas
\begin{equation}
    v'=v+\alpha\frac{x-x_\ast}{\abs{x-x_\ast}^2}+(1-\lambda)\bigl(\beta(v_\ast-v)+\gamma(x_\ast-x)\bigr)
\label{eq:micro_v}
\end{equation}
is the post-collisional velocity determined by an interaction of the form~\eqref{def:micro_v}-\eqref{def:micro_v_new} and
\begin{equation}
    \lambda'=1-\lambda
    \label{eq:micro_l}
\end{equation}
is the post-collisional degree of leadership produced by the microscopic rule~\eqref{def:micro_l}. The other agent participating in the binary interaction modifies symmetrically its velocity and degree of leadership according to the rules
\begin{align}
    v_\ast' &= v_\ast+\alpha\frac{x_\ast-x}{\abs{x_\ast-x}^2}+(1-\lambda_\ast)\bigl(\beta(v-v_\ast)+\gamma(x-x_\ast)\bigr), \label{eq:micro_v_symm} \\
    \lambda_\ast' &= 1-\lambda_\ast. \label{eq:micro_l_symm}
\end{align}
In \eqref{def:Qv_weak}-\eqref{def:Ql_weak} the kernel $B$ -classically called the Povzner kernel- takes into account the fact that the particle processes are not localized in the physical space. Hence, the Boltzmann-type kinetic equation \eqref{eq:Boltzmann_strong} may be also called Boltzmann-Povzner equation.
\begin{remark} \label{rem:non-cons_interactions}
In each binary interaction neither the mean velocity nor the mean degree of leadership is conserved in general. Indeed,
\begin{gather*}
    v'+v_\ast'-(v+v_\ast)=(\lambda_\ast-\lambda)\bigl(\beta(v_\ast-v)+\gamma(x_\ast-x)\bigr) \\[2mm]
    \lambda'+\lambda_\ast'-(\lambda+\lambda_\ast)=2\bigl(1-(\lambda+\lambda_\ast)\bigr).
\end{gather*}   
\end{remark}

Multiplying~\eqref{eq:Boltzmann_strong} by $\varphi$ and integrating w.r.t. $v$, $\lambda$ we obtain the weak form of the Boltzmann-type kinetic equation, which reads
\begin{multline}
    \partial_t\int_0^1\int_{\R^d}\varphi(v,\lambda)f(x,v,\lambda,t)\,dv\,d\lambda+\ddiv_x\int_0^1\int_{\R^d}\varphi(v,\lambda)vf(x,v,\lambda,t)\,dv\,d\lambda \\
    =\mu\ave{Q^v(f,f),\,\varphi}+\eta\ave{Q^\lambda(f,f),\,\varphi}.
    \label{eq:Boltz_weak}
\end{multline}

For completeness, we point out that the collisional operators $Q^v(f,f)$, $Q^\lambda(f,f)$ may be written explicitly also in strong form. However, since for the sequel the latter is less useful than~\eqref{def:Qv_weak},~\eqref{def:Ql_weak}, we confine it to Appendix~\ref{app:strong}, cf.~\eqref{def:Qv},~\eqref{def:Ql}.

\subsection{Approximation in the regime of small non-conservative interactions}
\label{sect:small_non-conservative}
In classical kinetic theory, the derivation of macroscopic equations via the hydrodynamic limit relies heavily on the concept of \textit{collisional invariants}, which are observable quantities conserved, on average, by the particle interactions. Collisional invariants are usually sought in the form of powers of the variables characterizing the microscopic state of the particles, which are at the basis of the computation of the statistical moments of the kinetic distribution function.

As pointed out in Remark~\ref{rem:non-cons_interactions}, the interaction rules~\eqref{eq:micro_v}-\eqref{eq:micro_v_symm} and~\eqref{eq:micro_l}-\eqref{eq:micro_l_symm} are pointwise non-conservative, and they are so also on average. Consequently, one cannot extract directly from them conservation properties leading to the identification of collisional invariants. To obtain suitable collisional invariants we adopt therefore the approach, borrowed from the kinetic theory of nearly elastic granular flows~\cite{toscani2004CHAPTER} (see also~\cite{fornasier2010PhD}), which consists in decomposing the microscopic rules in conservative and non-conservative parts.

In more detail, concerning~\eqref{eq:micro_v}-\eqref{eq:micro_v_symm} we consider the decomposition
\begin{align}
    \begin{aligned}[c]
        v' &= \bar{v}'+\frac{\lambda_\ast-\lambda}{2}\beta(v_\ast-v)+\alpha\frac{x-x_\ast}{\abs{x-x_\ast}^2}+(1-\lambda)\gamma(x_\ast-x) \\
        v_\ast' &= \bar{v}_\ast'+\frac{\lambda-\lambda_\ast}{2}\beta(v-v_\ast)+\alpha\frac{x_\ast-x}{\abs{x-x_\ast}^2}+(1-\lambda_\ast)\gamma(x-x_\ast)
    \end{aligned}
    \label{eq:micro_v_noncons}
\end{align}
where
\begin{equation}
    \bar{v}'=v+\left(1-\frac{\lambda+\lambda_\ast}{2}\right)\beta(v_\ast-v), \qquad 
        \bar{v}_\ast'=v_\ast+\left(1-\frac{\lambda+\lambda_\ast}{2}\right)\beta(v-v_\ast)
    \label{eq:micro_v_cons}
\end{equation} 
is pointwise conservative for the mean velocity, indeed
$$ \bar{v}'+\bar{v}_\ast'-(v+v_\ast)=0. $$

\begin{remark}
In~\eqref{eq:micro_v}-\eqref{eq:micro_v_symm}, also the repulsion term preserves the mean velocity pointwise. Nevertheless, here we choose not to include it in~\eqref{eq:micro_v_cons} so as to focus on a conservative part of the binary interaction rule involving a pure exchange of momentum between the interacting agents. Notwithstanding the partial arbitrariness of this choice, we will see later that it has the merit of isolating a single interaction mechanism (the alignment) in the conservative part, that will determine the local equilibrium, while yielding, in the hydrodynamic limit, a macroscopic model in which all the three interaction mechanisms (repulsion, alignment, attraction) play simultaneously a role.
\end{remark}

Next, we consider the regime in which the non-conservative contribution~\eqref{eq:micro_v_noncons} to the microscopic dynamics is small, i.e.
\begin{equation}
    \alpha=\alpha_0\delta, \qquad \beta=\beta_0\delta, \qquad \gamma=\gamma_0\delta
    \label{eq:q_inv_2}
\end{equation}
with $0<\delta\ll 1$. Because of \eqref{eq:q_inv_2}, we can approximate the test function $\varphi(v',\lambda)$ appearing in~\eqref{def:Qv_weak} in a neighborhood of $\bar{v}'$ as
\begin{equation}
    \varphi(v',\lambda)\approx\varphi(\bar{v}',\lambda)+\delta\left(\beta_0\frac{\lambda_\ast-\lambda}{2}(v_\ast-v)+\alpha_0\frac{x-x_\ast}{\abs{x-x_\ast}^2}
        +(1-\lambda)\gamma_0(x_\ast-x)\right)\cdot\nabla_v\varphi(\bar{v}',\lambda),
    \label{eq:Tay1}
\end{equation}
where we have assumed that $\varphi$ is sufficiently smooth and we have neglected higher order terms in the Taylor expansion. Plugging~\eqref{eq:Tay1} into~\eqref{def:Qv_weak} we obtain
$$ \ave{Q^v(f,f),\,\varphi}=\ave{\bar{Q}^v(f,f),\,\varphi}+\delta\ave{\cR^v(f,f),\,\varphi} $$
where
\begin{equation}
    \ave{\bar{Q}^v(f,f),\,\varphi}=
        \int_{[0,\,1]^2}\int_{\R^{3d}}B(\abs{x_\ast-x})\big(\varphi(\bar{v}',\lambda)-\varphi(v,\lambda)\big)ff_\ast\,dx_\ast\,dv\,dv_\ast\,d\lambda\,d\lambda_\ast
    \label{eq:barQv_weak}
\end{equation}
and
\begin{multline}
    \ave{\cR^v(f,f),\,\varphi}=\int_{[0,\,1]^2}\int_{\R^{3d}}B(\abs{x_\ast-x})\Biggl(\beta_0\frac{\lambda_\ast-\lambda}{2}(v_\ast-v)
        +\alpha_0\frac{x-x_\ast}{\abs{x-x_\ast}^2} \\
    +(1-\lambda)\gamma_0(x_\ast-x)\Biggr)\cdot\nabla_v\varphi(\bar{v}',\lambda)ff_\ast\,dx_\ast\,dv\,dv_\ast\,d\lambda\,d\lambda_\ast.
    \label{eq:Rv_weak}
\end{multline}

As far as rule~\eqref{eq:micro_l}-\eqref{eq:micro_l_symm} is concerned, in order to put in evidence a conservative contribution we generalize it by reformulating it as
\begin{equation}
    \lambda'=\bar{\lambda}'+\delta\bigl(1-2\lambda+\nu(\lambda-\lambda_\ast)\bigr), \qquad
        \lambda_\ast'=\bar{\lambda}_\ast'+\delta\bigl(1-2\lambda_\ast+\nu(\lambda_\ast-\lambda)\bigr)
    \label{eq:micro_l_corr}
\end{equation}
where
\begin{equation}
    \bar{\lambda}'=\lambda +\nu(\lambda_\ast-\lambda), \qquad
        \bar{\lambda}_\ast'=\lambda_\ast +\nu(\lambda-\lambda_\ast)
    \label{eq:micro_lambda_cons}
\end{equation}
satisfy
$$ \bar{\lambda}'+\bar{\lambda}_\ast'-(\lambda+\lambda_\ast)=0. $$
In~\eqref{eq:micro_l_corr},~\eqref{eq:micro_lambda_cons}, $\nu\in (0,\,1)$ is a new parameter that we will discuss briefly in a moment. Moreover, we notice that for $\delta=1$~\eqref{eq:micro_l_corr} reduces to~\eqref{eq:micro_l}.

\begin{remark} \label{rem:opposition.neutral}
For $\delta<1$, rule~\eqref{eq:micro_l_corr} features a conservative and non-conservative part. The former, i.e.~\eqref{eq:micro_lambda_cons}, corresponds to \textit{imitation} dynamics leading an agent to relax its degree of leadership towards that of the other agent at rate proportional to $\nu$. Conversely, the latter is built as an interplay between \textit{opposition} dynamics, expressed by $\delta\nu(\lambda-\lambda_\ast)$ and $\delta\nu(\lambda_\ast-\lambda)$, which lead an agent to differentiate its degree of leadership from that of the other agent at rate proportional to $\delta\nu$, and \textit{neutral} dynamics, expressed by $\delta(1-2\lambda)$ and $\delta(1-2\lambda_\ast)$, which induce an agent to relax spontaneously its degree of leadership towards the mid-value $\frac{1}{2}$ at rate $\delta$.
\end{remark}

Again, $\delta$ being small, we can approximate the test function $\varphi(v,\lambda')$ appearing in~\eqref{def:Ql_weak} in a neighborhood of $\bar{\lambda}'$ as
\begin{equation}
    \varphi(v,\lambda')\approx\varphi(v,\bar{\lambda}')+\delta\bigl(1-2\lambda+\nu(\lambda-\lambda_\ast)\bigr)\partial_\lambda\varphi(v,\bar{\lambda}'),
    \label{eq:Tay3}
\end{equation}
where we neglect higher order terms of the Taylor expansion. Plugging~\eqref{eq:Tay3} into~\eqref{def:Ql_weak} we obtain
$$ \ave{Q^\lambda(f,f),\,\varphi}=\ave{\bar{Q}^\lambda(f,f),\,\varphi}+\delta\ave{\cR^\lambda(f,f),\,\varphi} $$
where
\begin{equation}
    \ave{\bar{Q}^\lambda(f,f),\,\varphi}=\int_{[0,\,1]^2}\int_{\R^{3d}}B(\abs{x_\ast-x})\big(\varphi(v,\bar{\lambda}')-\varphi(v,\lambda)\big)
        ff_\ast\,dx_\ast\,dv\,dv_\ast\,d\lambda\,d\lambda_\ast
    \label{eq:barQl_weak}
\end{equation}
and
\begin{equation}
    \ave{\cR^\lambda(f,f),\,\varphi}= \\
    \int_{[0,\,1]^2}\int_{\R^{3d}}B(\abs{x_\ast-x})\bigl(1-2\lambda+\nu(\lambda-\lambda_\ast)\bigr)\partial_\lambda\varphi(v,\bar{\lambda}')
        ff_\ast\,dx_\ast\,dv\,dv_\ast\,d\lambda\,d\lambda_\ast.
    \label{eq:Rl_weak}
\end{equation}

On the whole, in the regime of small non-conservative interactions the Boltzmann-type equation~\eqref{eq:Boltz_weak} is approximated as
\begin{multline}
    \partial_t\int_0^1\int_{\R^d}\varphi(v,\lambda)f(x,v,\lambda,t)\,dv\,d\lambda+\ddiv_x\int_0^1\int_{\R^d}\varphi(v,\lambda)vf(x,v,\lambda,t)\,dv\,d\lambda \\
    =\mu\ave{\bar{Q}^v(f,f),\,\varphi}+\mu\delta\ave{\cR^v(f,f),\,\varphi}+\eta\ave{\bar{Q}^\lambda(f,f),\,\varphi}+\eta\delta\ave{\cR^\lambda(f,f),\,\varphi}.
    \label{eq:Boltz_weak_approx}
\end{multline}

In Appendix~\ref{app:strong} we report, for completeness, the strong forms of the operators $\bar{Q}^v(f,f)$, $\cR^v(f,f)$, $\bar{Q}^\lambda(f,f)$ and $\cR^\lambda(f,f)$, cf.~\eqref{eq:barQv}--\eqref{eq:Rl}.

\section{Macroscopic limit}\label{sect:macro}
In order to recover a macroscopic description of particle dynamics, the \textit{statistical moments} of the kinetic distribution function $f$ with respect to the microscopic states of the particles are typically introduced. In the present case, the microscopic state of the agents is defined by the variables $v\in\R^d$ and $\lambda\in [0,\,1]$, therefore we introduce:
\begin{enumerate}[label=(\roman*)]
\item the \textit{density} of the agents in the position $x\in\R^d$ at time $t>0$:
$$ \rho(x,t):=\int_0^1\int_{\R^d}f(x,v,\lambda,t)\,dv\,d\lambda; $$
\item the \textit{mean velocity} of the agents in the position $x\in\R^d$ at time $t>0$:
$$ u(x,t):=\frac{1}{\rho(x,t)}\int_0^1\int_{\R^d}vf(x,v,\lambda,t)\,dv\,d\lambda; $$
\item the \textit{mean degree of leadership} of the swarm in the position $x\in\R^d$ at time $t>0$:
$$ l(x,t):=\frac{1}{\rho(x,t)}\int_0^1\int_{\R^d}\lambda f(x,v,\lambda,t)\,dv\,d\lambda. $$
\end{enumerate}
Choosing $\varphi(v,\lambda)=1$ in~\eqref{eq:Boltz_weak} we obtain the continuity equation
\begin{equation}
    \partial_t\rho+\ddiv_x(\rho u)=0,
    \label{eq:continuity}
\end{equation}
which links the evolution of the density $\rho$ to the transport operated by the mean velocity $u$.
Next, choosing $\varphi(v,\lambda)=v$ we get the balance of linear momentum
$$ \partial_t(\rho u)+\ddiv_x\bigl(\rho(u\otimes u+\mathbb{D})\bigr)=\mu\ave{Q^v(f,f),v} $$
upon observing from~\eqref{def:Ql_weak} that $\ave{Q^\lambda(f,f),\,v}=0$. Here,
$$ \mathbb{D}(x,t):=\frac{1}{\rho(x,t)}\int_0^1\int_{\R^d}(v-u(x,t))\otimes (v-u(x,t))f(x,v,\lambda,t)\,dv\,d\lambda $$
is the variance-covariance matrix of the velocity. Finally, choosing $\varphi(v,\lambda)=\lambda$ we discover
\begin{equation}
    \partial_t(\rho l)+\ddiv_x\int_0^1\int_{\R^d}\lambda vf(x,v,\lambda,t)\,dv\,d\lambda=\eta\ave{Q^\lambda(f,f),\,\lambda}
    \label{eq:mean_lambda}
\end{equation}
upon observing from~\eqref{def:Qv_weak} that $\ave{Q^v(f,f),\,\lambda}=0$. Here, the second term on the left-hand side is the product moment of $f$ with respect to $v$ and $\lambda$.

Clearly, equations~\eqref{eq:continuity}--\eqref{eq:mean_lambda} are not closed at the macroscopic scale because they still require the knowledge of the kinetic distribution function $f$. Therefore, they do not constitute a self-consistent macroscopic model in the hydrodynamic parameters $\rho$, $u$, $l$. To circumvent this difficulty of the theory, we take advantage of the regime of small non-conservativeness of the interactions discussed in Section~\ref{sect:small_non-conservative} and formalised by the approximated Boltzmann-type equation~\eqref{eq:Boltz_weak_approx}. 

\subsection{The hydrodynamic limit}
More precisely, we introduce a small parameter
$$ 0<\varepsilon\ll 1, $$
conceptually analogous to the \textit{Knudsen number} of the classical kinetic theory, and we consider the scaling
\begin{equation}
    \delta=\varepsilon, \qquad \mu=\eta=\frac{1}{\varepsilon},
    \label{eq:scaling}
\end{equation}
in such a way that the smallness of the non-conservative content of the microscopic interactions is balanced by the increased interaction rates. Notice that $\delta$, and therefore $\varepsilon$ within the above scaling, scales the interaction coefficients $\alpha,\,\beta,\,\gamma$ via~\eqref{eq:q_inv_2}. Hence,~\eqref{eq:Boltz_weak_approx} becomes
\begin{multline}
    \partial_t\int_0^1\int_{\R^d}\varphi(v,\lambda)f^\varepsilon(x,v,\lambda,t)\,dv\,d\lambda+\ddiv_x\int_0^1\int_{\R^d}\varphi(v,\lambda)vf^\varepsilon(x,v,\lambda,t)\,dv\,d\lambda \\
    =\frac{1}{\varepsilon}\ave{\bar{Q}^{v,\varepsilon}(f^\varepsilon,f^\varepsilon)+\bar{Q}^\lambda(f^\varepsilon,f^\varepsilon),\,\varphi}
        +\ave{\cR^{v,\varepsilon}(f^\varepsilon,f^\varepsilon)+\cR^\lambda(f^\varepsilon,f^\varepsilon),\,\varphi},
    \label{eq:Boltz_weak_scaled}
\end{multline}
where we use the notation $f^\varepsilon$ because the solution to~\eqref{eq:Boltz_weak_approx} depends now on the scaling parameter $\varepsilon$. We also use the notations $\bar{Q}^{v,\varepsilon}$, $\cR^{v,\varepsilon}$ to stress that these two operators are parameterized by $\varepsilon$ also independently of $f^\varepsilon$, as it is clear from~\eqref{eq:barQv_weak},~\eqref{eq:Rv_weak}. Notice indeed that $\bar{v}'$ depends on $\varepsilon$ through the scaled coefficient $\beta$, cf.~\eqref{eq:micro_v_cons}. Conversely, from~\eqref{eq:barQl_weak},~\eqref{eq:Rl_weak} we see that $\bar{Q}^\lambda$, $\cR^\lambda$ are not parameterized by $\varepsilon$, essentially because $\bar{\lambda}'$ does not depend on $\varepsilon$, cf.~\eqref{eq:micro_lambda_cons}.

Let
$$ n^\varepsilon(x,\lambda,t):=\frac{1}{\rho^\varepsilon(x,t)}\int_{\R^d}f^\varepsilon(x,v,\lambda,t)\,dv $$
be the marginal distribution of the degree of leadership $\lambda$ in the point $x$ at time $t$. Choosing $\varphi(v,\lambda)=\psi(\lambda)$ in~\eqref{eq:Boltz_weak_scaled}, i.e. an observable quantity independent of $v$, we get
\begin{multline}
    \partial_t\!\left(\rho^\varepsilon(x,t)\int_0^1\psi(\lambda)n^\varepsilon(x,\lambda,t)\,d\lambda\right)+
            \ddiv_x\int_0^1\int_{\R^d}\psi(\lambda)vf^\varepsilon(x,v,\lambda,t)\,dv\,d\lambda \\
        =\frac{1}{\varepsilon}\ave{\bar{Q}^\lambda(f^\varepsilon,f^\varepsilon),\,\psi}+\ave{\cR^\lambda(f^\varepsilon,f^\varepsilon),\,\psi},
    \label{eq:n.eps}
\end{multline}
after observing that
$$ \ave{\bar{Q}^{v,\varepsilon}(f^\varepsilon,f^\varepsilon),\,\psi}=\ave{\cR^{v,\varepsilon}(f^\varepsilon,f^\varepsilon),\,\psi}=0,
    \qquad \forall\,\psi=\psi(\lambda), $$
cf.~\eqref{eq:barQv_weak},~\eqref{eq:Rv_weak}. A direct computation using~\eqref{eq:barQl_weak},~\eqref{eq:Rl_weak} shows that the terms on the right-hand side of~\eqref{eq:n.eps} can be written explicitly using $n^\varepsilon$ as
$$ \ave{\bar{Q}^\lambda(f^\varepsilon,f^\varepsilon),\,\psi}=\rho^\varepsilon(x,t)\displaystyle{\int_{[0,\,1]^2}\int_{\R^d}}B(\abs{x_\ast-x})\rho^\varepsilon(x_\ast,t)
        \left(\psi(\bar{\lambda}')-\psi(\lambda)\right)n^\varepsilon n_\ast^\varepsilon\,dx_\ast\,d\lambda\,d\lambda_\ast $$
and
\begin{multline*}
   \ave{\cR^\lambda(f^\varepsilon,f^\varepsilon),\,\psi}=\rho^\varepsilon(x,t)\int_{[0,\,1]^2}\int_{\R^d}B(\abs{x_\ast-x})\rho^\varepsilon(x_\ast,t)
        (1-2\lambda+\nu(\lambda-\lambda_\ast)) \\
    \times\psi'(\bar{\lambda}')n^\varepsilon n^\varepsilon_\ast\,dx_\ast\,d\lambda\,d\lambda_\ast.
\end{multline*}
Let now
$$ p^\varepsilon(x,v,t):=\frac{1}{\rho^\varepsilon(x,t)}\int_0^1f^\varepsilon(x,v,\lambda,t)\,d\lambda $$
be the marginal distribution of the velocity $v$ in the point $x$ at time $t$. Setting $\varphi(v,\lambda)=\phi(v)$ in~\eqref{eq:Boltz_weak_scaled}, i.e. an observable quantity independent of $\lambda$, we obtain
\begin{multline}
        \partial_t\!\left(\rho^\varepsilon(x,t)\int_{\R^d}\phi(v)p^\varepsilon(x,v,t)\,dv\right)+
            \ddiv_x\!\left(\rho^\varepsilon(x,t)\int_{\R^d}\phi(v)vp^\varepsilon(x,v,t)\,dv\right) \\
        =\frac{1}{\varepsilon}\ave{\bar{Q}^{v,\varepsilon}(f^\varepsilon,f^\varepsilon),\,\phi}+\ave{\cR^{v,\varepsilon}(f^\varepsilon,f^\varepsilon),\,\phi},
    \label{eq:p.eps}
\end{multline}
upon noticing from~\eqref{eq:barQl_weak},~\eqref{eq:Rl_weak} that
$$ \ave{\bar{Q}^\lambda(f^\varepsilon,f^\varepsilon),\,\phi}=\ave{\cR^\lambda(f^\varepsilon,f^\varepsilon),\,\phi}=0, \qquad \forall\,\phi=\phi(v). $$

Next, we give the following
\begin{definition} \label{def:coll_inv}
We call \emph{collisional invariant} of the kinetic equation~\eqref{eq:n.eps} any observable quantity $\psi=\psi(\lambda)$ such that
$$ \ave{\bar{Q}^\lambda(g,g),\,\psi}=0 $$
for every distribution function $g$.

Likewise, we call \emph{collisional invariant} of the kinetic equation~\eqref{eq:p.eps} any observable quantity $\phi=\phi(v)$ such that
$$ \ave{\bar{Q}^{v,\varepsilon}(g,g),\,\phi}=0 $$
for every distribution function $g$ and every $\varepsilon>0$.
\end{definition}

Plugged into~\eqref{eq:n.eps},~\eqref{eq:p.eps} collisional invariants produce averaged equations free from the singular coefficient $\frac{1}{\varepsilon}$, hence suitable to pass to the hydrodynamic limit $\varepsilon\to 0^+$ and recover a universal macroscopic description of the system.

It is easy to check that both $\psi\equiv 1$ and $\phi\equiv 1$ are collisional invariants in the sense of Definition~\ref{def:coll_inv}. Nevertheless, owing to the non-locality in space embodied in the collision kernel $B$ contained in $\bar{Q}^{v,\varepsilon}$ and $\bar{Q}^\lambda$, determining collisional invariants in general is not trivial. To circumvent this difficulty of the theory, taking inspiration from~\cite{lacho1990ARMA} we resort to a local-in-space approximation of the collisional operators $\bar{Q}^{v,\varepsilon}$, $\bar{Q}^\lambda$ under the assumption that the interaction range $R$ is sufficiently small. Specifically, we consider
\begin{align*}
    & \ave{\bar{Q}^{v,\varepsilon}_\loc(f^\varepsilon,f^\varepsilon),\,\phi}:=\cB_0\int_{[0,\,1]^2}\int_{\R^{2d}}\bigl(\phi(\bar{v}')-\phi(v)\bigr)
        f^\varepsilon(x,v,\lambda,t)f^\varepsilon(x,v_\ast,\lambda_\ast,t)\,dv\,dv_\ast\,d\lambda\,d\lambda_\ast, \\
    & \ave{\bar{Q}^\lambda_\loc(f^\varepsilon,f^\varepsilon),\,\psi}:=\cB_0{(\rho^\varepsilon(x,t))}^2\int_{[0,\,1]^2}\big(\psi(\bar{\lambda}')-\psi(\lambda)\big)
        n^\varepsilon(x,\lambda,t)n^\varepsilon(x,\lambda_\ast,t)\,d\lambda\,d\lambda_\ast,
\end{align*}
where
$$ \cB_0:=\int_{\sB_R(0)}B(\abs{x})\,dx, $$
and we reformulate~\eqref{eq:n.eps},~\eqref{eq:p.eps} accordingly:
\begin{multline}
    \partial_t\!\left(\rho^\varepsilon(x,t)\int_0^1\psi(\lambda)n^\varepsilon(x,\lambda,t)\,d\lambda\right)+
            \ddiv_x\int_0^1\int_{\R^d}\psi(\lambda)vf^\varepsilon(x,v,\lambda,t)\,dv\,d\lambda \\
        =\frac{1}{\varepsilon}\ave{\bar{Q}^\lambda_\loc(f^\varepsilon,f^\varepsilon),\,\psi}+\ave{\cR^\lambda(f^\varepsilon,f^\varepsilon),\,\psi},
    \label{eq:n.eps_loc}
\end{multline}
\begin{multline}
        \partial_t\!\left(\rho^\varepsilon(x,t)\int_{\R^d}\phi(v)p^\varepsilon(x,v,t)\,dv\right)+
            \ddiv_x\!\left(\rho^\varepsilon(x,t)\int_{\R^d}\phi(v)vp^\varepsilon(x,v,t)\,dv\right) \\
        =\frac{1}{\varepsilon}\ave{\bar{Q}^{v,\varepsilon}_\loc(f^\varepsilon,f^\varepsilon),\,\phi}+\ave{\cR^{v,\varepsilon}(f^\varepsilon,f^\varepsilon),\,\phi}.
    \label{eq:p.eps_loc}
\end{multline}
Correspondingly, we rephrase Definition~\ref{def:coll_inv} as
\begin{definition} \label{def:local_coll_inv}
We call \emph{local collisional invariants} any observable quantities $\psi=\psi(\lambda)$, $\phi=\phi(v)$, respectively, such that
$$ \ave{\bar{Q}^\lambda_\loc(g,g),\,\psi}=0, \qquad \ave{\bar{Q}^{v,\varepsilon}_\loc(g,g),\,\phi}=0 $$
for every distribution function $g$ and every $\varepsilon>0$.
\end{definition}

Recalling~\eqref{eq:micro_v_cons},~\eqref{eq:micro_lambda_cons}, we see that $\psi(\lambda)=\lambda$ and $\phi(v)=v$ are local collisional invariants in the sense of Definition~\ref{def:local_coll_inv}. Plugging them, together with $\psi=\phi\equiv 1$, into~\eqref{eq:n.eps_loc},~\eqref{eq:p.eps_loc} we obtain:
\begin{equation}
    \resizebox{.93\textwidth}{!}{$
    \begin{cases}
        \partial_t\rho^\varepsilon(x,t)+\ddiv_x\!\left(\rho^\varepsilon(x,t)\displaystyle{\int_{\R^d}}vp^\varepsilon(x,v,t)\,dv\right)=0, \\[3mm]
        \partial_t\!\left(\rho^\varepsilon(x,t)\displaystyle{\int_{\R^d}}vp^\varepsilon(x,v,t)\,dv\right)+
            \ddiv_x\!\left(\rho^\varepsilon(x,t)\displaystyle{\int_{\R^d}}v\otimes vp^\varepsilon(x,v,t)\,dv\right)=\ave{\cR^{v,\varepsilon}(f^\varepsilon,f^\varepsilon),\,v}, \\[3mm]
        \partial_t\!\left(\rho^\varepsilon(x,t)\displaystyle{\int_0^1}\lambda n^\varepsilon(x,\lambda,t)\,d\lambda\right)+
            \ddiv_x\displaystyle{\int_0^1\int_{\R^d}}\lambda vf^\varepsilon(x,v,\lambda,t)\,dv\,d\lambda=\ave{\cR^\lambda(f^\varepsilon,f^\varepsilon),\,\lambda}
    \end{cases}
    $}
    \label{eq:hydro.eps}
\end{equation}
for every $\varepsilon>0$. In particular, we have taken into account that
$$ \ave{\cR^{v,\varepsilon}(f^\varepsilon,f^\varepsilon),\,1}=\ave{\cR^\lambda(f^\varepsilon,f^\varepsilon),\,1}=0, \quad
    \ave{\cR^\lambda(f^\varepsilon,f^\varepsilon),\,v}=\ave{\cR^{v,\varepsilon}(f^\varepsilon,f^\varepsilon),\,\lambda}=0, \qquad
        \forall\,\varepsilon>0. $$

We consider now~\eqref{eq:n.eps_loc},~\eqref{eq:p.eps_loc} in the \textit{hydrodynamic limit} $\varepsilon\to 0^+$, so as to recover the local equilibrium distribution $f^0$ which, together with its (normalized) marginals $n^0$, $p^0$, allows one to close system~\eqref{eq:hydro.eps} within the same limit. Assuming formally that all terms of~\eqref{eq:n.eps_loc},~\eqref{eq:p.eps_loc} remain bounded in the limit, the usual way to look for physically meaningful local equilibrium distributions is to see them as equilibrium solutions of the pure interaction dynamics:
\begin{equation}
    \partial_tf^\varepsilon=\frac{1}{\varepsilon}\left(\bar{Q}^{v,\varepsilon}_\loc(f^\varepsilon,f^\varepsilon)+\bar{Q}^\lambda_\loc(f^\varepsilon,f^\varepsilon)\right)
    \label{eq:pure_int}
\end{equation}
for $\varepsilon\to 0^+$. Testing~\eqref{eq:pure_int} against an observable quantity $\psi=\psi(\lambda)$ yields
$$ \partial_t\bigl(\rho^\varepsilon\ave{n^\varepsilon,\,\psi}\bigr)=\frac{1}{\varepsilon}\ave{\bar{Q}^\lambda_\loc(f^\varepsilon,f^\varepsilon),\,\psi}, $$
which, in the limit $\varepsilon\to 0^+$ and owing to the assumption of boundedness recalled above, produces $\ave{\bar{Q}^\lambda_\loc(f^0,f^0),\,\psi}=0$, whence
$$ \int_{[0,\,1]^2}\big(\psi(\bar{\lambda}')-\psi(\lambda)\big)n^0(x,\lambda,t)n^0(x,\lambda_\ast,t)\,d\lambda\,d\lambda_\ast=0,
    \qquad \forall\,\psi=\psi(\lambda). $$
Since $\psi(\lambda)=1$ and $\psi(\lambda)=\lambda$ are local collisional invariants, they satisfy trivially this relationship and imply simply that $n^0$ has unit mass for all $x$, $t$ and is parameterized by the local mean degree of leadership $l(x,t)$. With $\psi(\lambda)={(\lambda-l(x,t))}^2$ we discover instead
$$ 2\nu(1-\nu)\Var[\lambda]=0, $$
where
$$ \Var[\lambda](x,t):=\int_0^1\lambda^2n^0(x,\lambda,t)\,d\lambda-l^2(x,t) $$
is the local variance of the distribution of $\lambda$. Since $0<\nu<1$ this implies $\Var[\lambda]=0$, thus the $\lambda$-marginal $n^0$ is a Dirac delta centered at $\lambda=l(x,t)$:
\begin{equation}\label{eq:n0}
    n^0(x,\lambda,t)=\delta(\lambda-l(x,t)).
\end{equation}
This also implies that $f^0$ has the form
\begin{equation}
    f^0(x,v,\lambda,t)=\rho(x,t)p^0(x,v,t)\delta(\lambda-l(x,t)).
    \label{eq:f0}
\end{equation}

Testing now~\eqref{eq:pure_int} against an observable quantity $\phi=\phi(v)$ we obtain
\begin{equation}
    \partial_t\bigl(\rho^\varepsilon\ave{p^\varepsilon,\,\phi}\bigr)=\frac{1}{\varepsilon}\ave{\bar{Q}^{v,\varepsilon}_\loc(f^\varepsilon,f^\varepsilon),\,\phi}.
    \label{eq:pure_int.v}
\end{equation}
In this case the limit $\varepsilon\to 0^+$ is more delicate, because the collisional operator $\bar{Q}^{v,\varepsilon}_\loc(f^\varepsilon,f^\varepsilon)$ depends on $\varepsilon$ also besides $f^\varepsilon$, cf.~\eqref{eq:barQv_weak}. In particular, by dominated convergence we observe that $\ave{\bar{Q}^{v,\varepsilon}_\loc(g,g),\,\phi}\to 0$ when $\varepsilon\to 0^+$ for every distribution function $g$ and every (smooth) observable $\phi$. Therefore, $\frac{1}{\varepsilon}\ave{\bar{Q}^{v,\varepsilon}_\loc(f^\varepsilon,f^\varepsilon),\,\phi}$ is in principle indefinite in the limit. However, if $\phi$ is smooth, say e.g., $\phi\in C^{1,1}(\R^d)$, we may expand
$$ \phi(\bar{v}')=\phi(v)+\nabla{\phi(\tilde{v})}\cdot\varepsilon\beta_0\left(1-\frac{\lambda+\lambda_\ast}{2}\right)(v_\ast-v) $$
with $\tilde{v}:=\theta\bar{v}'+(1-\theta)v$ for some $\theta\in [0,\,1]$. Consequently
$$ \frac{1}{\varepsilon}\ave{\bar{Q}^{v,\varepsilon}_\loc(f^\varepsilon,f^\varepsilon),\,\phi}=\ave{\cP(f^\varepsilon,f^\varepsilon),\,\phi}+\gR^\varepsilon(\phi,f^\varepsilon), $$
where
\begin{equation}
    \ave{\cP(f^\varepsilon,f^\varepsilon),\,\phi}:=\beta_0\cB_0\int_{[0,\,1]^2}\int_{\R^{2d}}\nabla{\phi(v)}\cdot\left(1-\frac{\lambda+\lambda_\ast}{2}\right)(v_\ast-v)
        f^\varepsilon f^\varepsilon_\ast\,dv\,dv_\ast\,d\lambda\,d\lambda_\ast,
    \label{eq:P.Fokker-Planck}
\end{equation}
and here $f_\ast^\varepsilon$ is shorthand for $f_\ast^\varepsilon(x,v_\ast,\lambda_\ast,t)$ owing to the local-in-space approximation, whereas the remainder is
$$ \gR^\varepsilon(\phi,f^\varepsilon):=\beta_0\cB_0\int_{[0,\,1]^2}\int_{\R^{2d}}\Bigl(\nabla{\phi(\tilde{v})}-\nabla{\phi(v)}\Bigr)\cdot\left(1-\frac{\lambda+\lambda_\ast}{2}\right)(v_\ast-v)
    f^\varepsilon f^\varepsilon_\ast\,dv\,dv_\ast\,d\lambda\,d\lambda_\ast. $$
In particular,
\begin{align}
    \begin{aligned}[b]
        \abs{\gR^\varepsilon(\phi,f^\varepsilon)} &\leq \beta_0\cB_0\Lip(\nabla{\phi})\theta\int_{[0,\,1]^2}\int_{\R^{2d}}
            \abs{\bar{v}'-v}\left(1-\frac{\lambda+\lambda_\ast}{2}\right)\abs{v_\ast-v}f^\varepsilon f_\ast^\varepsilon\,dv\,dv_\ast\,d\lambda\,d\lambda_\ast \\
        &= \varepsilon\beta_0^2\cB_0\Lip(\nabla{\phi})\theta\int_{[0,\,1]^2}\int_{\R^{2d}}
            \left(1-\frac{\lambda+\lambda_\ast}{2}\right)^2{\abs{v_\ast-v}}^2f^\varepsilon f_\ast^\varepsilon\,dv\,dv_\ast\,d\lambda\,d\lambda_\ast \\
        &\leq \varepsilon\beta_0^2\cB_0\Lip(\nabla{\phi})\theta(\rho^\varepsilon)^2\int_{\R^{2d}}{\abs{v_\ast-v}}^2p^\varepsilon p_\ast^\varepsilon\,dv\,dv_\ast \\
        &= 2\varepsilon\beta_0^2\cB_0\Lip(\nabla{\phi})\theta(\rho^\varepsilon)^2\Var^\varepsilon[v]
    \end{aligned}
    \label{eq:remainder_v}
\end{align}
with $\Lip(\nabla\phi)>0$ the Lipschitz constant of $\nabla{\phi}$ and
$$ \Var^\varepsilon[v](x,t):=\int_{\R^d}\abs{v}^2p^\varepsilon(x,v,t)\,dv-\abs{u(x,t)}^2 $$
the local variance of the $\varepsilon$-distribution of $v$. To ascertain the trend of $\gR^\varepsilon$ in the limit $\varepsilon\to 0^+$ we investigate that of $\Var^\varepsilon[v]$. Letting $\phi(v)=\abs{v-u^\varepsilon}^2$ in~\eqref{eq:pure_int.v} and observing that in the pure local interaction dynamics $\rho^\varepsilon$ and $u^\varepsilon$ are conserved (put $\phi(v)=1,\,v$ in~\eqref{eq:pure_int.v} to see that the right-hand side vanishes) we get
$$ \rho^\varepsilon\partial_t\ave{p^\varepsilon,\,\abs{v-u^\varepsilon}^2}=\frac{1}{\varepsilon}\cB_0\int_{[0,\,1]^2}\int_{\R^{2d}}\left(\abs{\bar{v}'-u^\varepsilon}^2-\abs{v-u^\varepsilon}^2\right)
    f^\varepsilon f_\ast^\varepsilon\,dv\,dv_\ast\,d\lambda\,d\lambda_\ast, $$
whence, after some manipulations, we obtain
$$ \partial_t\Var^\varepsilon[v]\leq\beta_0\cB_0\rho^\varepsilon\left(\varepsilon\beta_0\Var^\varepsilon[v]
    +2\int_{\R^{2d}}\abs{v\cdot(v_\ast-v)}p^\varepsilon p_\ast^\varepsilon\,dv\,dv_\ast\right). $$
But
\begin{align*}
    \int_{\R^{2d}}\abs{v\cdot(v_\ast-v)}p^\varepsilon p_\ast^\varepsilon\,dv\,dv_\ast &\leq
        \int_{\R^{2d}}\abs{v\cdot v_\ast}p^\varepsilon p_\ast^\varepsilon\,dv\,dv_\ast+\int_{\R^d}\abs{v}^2p^\varepsilon\,dv \\
    &\leq\left(\int_{\R^d}\abs{v}p^\varepsilon\,dv\right)^2+\int_{\R^d}\abs{v}^2p^\varepsilon\,dv \\
    &\leq 2\int_{\R^d}\abs{v}^2p^\varepsilon\,dv=2\left(\Var^\varepsilon[v]+(u^\varepsilon)^2\right)
\end{align*}
and finally
$$ \partial_t\Var^\varepsilon[v]\leq\beta_0\cB_0\rho^\varepsilon\Bigl((\varepsilon\beta_0+4)\Var^\varepsilon[v]+2(u^\varepsilon)^2\Bigr). $$
If we take $\varepsilon$ so small that $\varepsilon\beta_0\leq 1$ we have
$$ \partial_t\Var^\varepsilon[v]\leq\beta_0\cB_0\rho^\varepsilon\Bigl(5\Var^\varepsilon[v]+2(u^\varepsilon)^2\Bigr), $$
which implies that $\Var^\varepsilon[v]$ is bounded for $\varepsilon\to 0^+$ at any time under the assumption that so are the hydrodynamic parameters $\rho^\varepsilon$, $u^\varepsilon$.

Back to~\eqref{eq:remainder_v}, we conclude $\gR(\phi,f^\varepsilon)\to 0$ when $\varepsilon\to 0^+$. Passing to the limit in~\eqref{eq:pure_int.v} this yields then
\begin{align*}
    \rho\partial_t\ave{p^0,\,\phi} &= \beta_0\cB_0\int_{[0,\,1]^2}\int_{\R^{2d}}\nabla{\phi(v)}\cdot\left(1-\frac{\lambda+\lambda_\ast}{2}\right)(v_\ast-v)
        f^0f_\ast^0\,dv\,dv_\ast\,d\lambda\,d\lambda_\ast \\
    &=\beta_0\cB_0(1-l)\rho^2\int_{\R^d}\nabla{\phi}(v)\cdot(u-v)p^0\,dv,
\end{align*}
where we have used the form~\eqref{eq:f0} of $f^0$. Restricting $\phi$ to compactly supported functions we see that this is a weak form of
$$ \partial_tp^0=-\beta_0\cB_0(1-l)\rho\ddiv_v{\!\left[(u-v)p^0\right]}, $$
whose equilibrium solution satisfies $(u-v)p^0=0$. In distributional sense we deduce therefore
\begin{equation}\label{eq:p0}
    p^0(x,v,t)=\delta(v-u(x,t))
\end{equation}
and finally from~\eqref{eq:f0}
\begin{equation}
    f^0(x,v,\lambda,t)=\rho(x,t)\delta(v-u(x,t))\otimes\delta(\lambda-l(x,t)).
    \label{eq:f0.complete}
\end{equation}
Plugging~\eqref{eq:n0},~\eqref{eq:p0},~\eqref{eq:f0.complete} into~\eqref{eq:hydro.eps}, after passing there formally to the limit $\varepsilon\to 0^+$, we get the following set of hydrodynamic equations:
\begin{subnumcases}{\label{eq:macro}}
    \partial_t\rho+\ddiv_x{(\rho u)}=0 \label{eq:macro_a} \\
    \partial_t(\rho u)+\ddiv_x{(\rho u\otimes u)}=
        \rho(x,t)\displaystyle\int_{\R^d}B(\abs{x_\ast-x})\rho(x_\ast,t) \nonumber \\
        \qquad\qquad\qquad\qquad\qquad\qquad\qquad\qquad
            \times\biggl(\frac{l(x_\ast,t)-l(x,t)}{2}\beta_0(u(x_\ast,t)-u(x,t)) \nonumber \\
        \qquad\qquad\qquad\qquad\qquad\qquad\qquad\qquad
            \displaystyle+\alpha_0\frac{x-x_\ast}{\abs{x-x_\ast}^2}+\gamma_0\bigl(1-l(x,t)\bigr)(x_\ast-x)\biggr)\,dx_\ast \label{eq:macro_b} \\
    \partial_t(\rho l)+\ddiv_x{(\rho lu)}=\rho(x,t)\displaystyle\int_{\R^d}B(\abs{x_\ast-x})\rho(x_\ast,t) \nonumber \\
        \qquad\qquad\qquad\qquad\qquad\qquad\qquad
            \times\Bigl(1-2l(x,t)+\nu\bigl(l(x,t)-l(x_\ast,t)\bigr)\Bigr)\,dx_\ast, \label{eq:macro_c}
\end{subnumcases}
which can be rewritten in non-conservative form as
\begin{subnumcases}{\label{eq:macro.noncons}}
    \partial_t\rho+\ddiv_x{(\rho u)}=0 \label{eq:macro.noncons_a} \\
    \partial_tu+u\cdot\nabla{u}=
        \displaystyle\int_{\R^d}B(\abs{x_\ast-x})\rho(x_\ast,t)\biggl(\frac{l(x_\ast,t)-l(x,t)}{2}\beta_0(u(x_\ast,t)-u(x,t)) \nonumber \\
            \qquad\qquad\qquad\qquad \displaystyle+\alpha_0\frac{x-x_\ast}{\abs{x-x_\ast}^2}+\gamma_0\bigl(1-l(x,t)\bigr)(x_\ast-x)\biggr)dx_\ast \label{eq:macro.noncons_b} \\
    \partial_tl+u\cdot\nabla{l}=\displaystyle\int_{\R^d}B(\abs{x_\ast-x})\rho(x_\ast,t)\Bigl(1-2l(x,t)+\nu\bigl(l(x,t)-l(x_\ast,t)\bigr)\Bigr)dx_\ast, \label{eq:macro.noncons_c}
\end{subnumcases}
For consistency with the derivation discussed so far, we require that the non-locality in space on the right-hand side of both~\eqref{eq:macro} and~\eqref{eq:macro.noncons} be sufficiently small.

\begin{remark}
A comment about the simultaneous presence in~\eqref{eq:macro}, and likewise in~\eqref{eq:macro.noncons}, of local approximations and non-local terms is in order. As it is evident from~\eqref{eq:n.eps_loc}--\eqref{eq:hydro.eps}, system~\eqref{eq:macro} originates from a local approximation of the collisional operators $\bar{Q}^\lambda$, $\bar{Q}^{v,\varepsilon}$ while the operators $\cR^\lambda$, $\cR^v$ are left non-local. Such a different treatment of these operators has the following twofold reason. On one hand, since the original microscopic dynamics are non-local in space, we aimed to derive a hydrodynamic description which could keep track consistently of such a non-locality. On the other hand, the non-locality in the kinetic description makes it difficult to obtain explicit analytical expressions of the equilibrium distributions, which are needed to pass to the hydrodynamic limit. In order to balance between these two competing needs, we chose to localise only the collisional operators $\bar{Q}^\lambda$, $\bar{Q}^{v,\varepsilon}$, which enter the determination of the equilibrium distributions, leaving the remainders $\cR^\lambda$, $\cR^v$ non-local, since owing to the scaling~\eqref{eq:scaling} they do not affect the search for equilibrium distributions. Clearly, smallness of such a non-locality has to be required for consistency between the local closure of the fluxes and the non-local terms in~\eqref{eq:macro}.
\end{remark}

Model~\eqref{eq:macro.noncons} is a closed system of coupled equations for the hydrodynamic parameters $\rho$, $u$, $l$. The first equation is the continuity equation expressing the conservation of mass of the agents. The second equation is a Burger-like momentum equation, whose right-hand side accounts for the non-conservative contribution of alignment dynamics plus repulsion and attraction, cf.~\eqref{eq:micro_v_noncons}. The third equation is an auxiliary transport equation, which describes the aggregate evolution of leader-follower transitions in terms of neutral and opposition dynamics (cf.~\eqref{eq:micro_l_corr} and Remark~\ref{rem:opposition.neutral}).

\subsection{On the local approximations of~\texorpdfstring{$\boldsymbol{\bar{Q}^{v,\varepsilon}}$}{},~\texorpdfstring{$\boldsymbol{\bar{Q}^\lambda}$}{} in the hydrodynamic limit}
The hydrodynamic system~\eqref{eq:macro} has been obtained for $\varepsilon\to 0^+$ after approximating the collisional operators $\bar{Q}^{v,\varepsilon}$, $\bar{Q}^\lambda$ with their local-in-space versions $\bar{Q}^{v,\varepsilon}_\loc$, $\bar{Q}^\lambda_\loc$, respectively. Furthermore, in the hydrodynamic limit the operator $\frac{1}{\varepsilon}\bar{Q}^{v,\varepsilon}_\loc$ has been formally replaced by the operator $\cP$, cf.~\eqref{eq:P.Fokker-Planck}, for the determination of the local marginal equilibrium distribution $p^0$.

In this section, we investigate the order of the introduced approximations to ascertain the consistency of the macroscopic description~\eqref{eq:macro} with respect to the original stochastic particle description~\eqref{def:micro_drift}-\eqref{def:micro_v}.

Our main result is the following:
\begin{theorem} \label{theo:local_approx_Q}
Let $f=f(x,v,\lambda,t)$ be a distribution function such that:
\begin{itemize}
\item $f(\cdot,v,\lambda,t)\in C^{1,1}(\R^d)$ for every $v\in\R^d$, $\lambda\in [0,\,1]$, $t\geq 0$;
\item $f(x,\cdot,\lambda,t)$ is $x$-uniformly compactly supported for every $\lambda\in [0,\,1]$, $t\geq 0$.
\end{itemize}
Then:
$$ \lim_{\varepsilon\to 0^+}\abs{\ave{\frac{1}{\varepsilon}\bar{Q}^{v,\varepsilon}(f,f)-\cP(f,f),\,\phi}}\leq\cO(R^{1+d}), \qquad \forall\,\phi=\phi(v)\in C^{0,1}(\R^d)\cap C^{1,1}(\R^d) $$
and likewise
$$ \abs{\ave{\bar{Q}^\lambda(f,f)-\bar{Q}^\lambda_\loc(f,f),\,\psi}}\leq\cO(R^{1+d}), \qquad \forall\,\psi=\psi(\lambda)\in L^\infty(0,\,1). $$
\end{theorem}
\begin{proof}
\begin{enumerate}[label=\arabic*.]
\item We begin with the approximation of $\frac{1}{\varepsilon}\bar{Q}^{v,\varepsilon}$. We have:
\begin{align}
    \begin{aligned}[b]
        &\abs{\ave{\frac{1}{\varepsilon}\bar{Q}^{v,\varepsilon}(f,f)-\cP(f,f),\,\phi}} \\
        &\quad =\left\vert\int_{[0,\,1]^2}\int_{\R^{3d}}B(\abs{x_\ast-x})\frac{\phi(\bar{v}')-\phi(v)}{\varepsilon}ff_\ast\,dx_\ast\,\dots\,d\lambda_\ast\right. \\
        &\quad \phantom{=}\left.-\beta_0\cB_0\int_{[0,\,1]^2}\int_{\R^{2d}}\nabla{\phi}(v)\left(1-\frac{\lambda+\lambda_\ast}{2}\right)\cdot(v_\ast-v)ff_\ast\,dv\,\dots\,d\lambda_\ast\right\vert,
    \end{aligned}
    \label{eq:proof.1}
\end{align}
where, here and henceforth, for the sake of simplicity in the non-local terms $f_\ast$ and $dx_\ast\,\dots\,d\lambda_\ast$ are shorthand for $f(x_\ast,v_\ast,\lambda_\ast,t)$ and $dx_\ast\,dv\,dv_\ast\,d\lambda\,d\lambda_\ast$, respectively, whereas in the local terms $f_\ast$ and $dv\,\dots\,d\lambda_\ast$ are shorthand for $f(x,v_\ast,\lambda_\ast,t)$ and $dv\,dv_\ast\,d\lambda\,d\lambda_\ast$, respectively.

Since $\phi$ is smooth, we write
$$ \phi(\bar{v}')-\phi(v)=\nabla{\phi}(\tilde{v})\cdot(\bar{v}'-v) $$
with $\tilde{v}:=\theta\bar{v}'+(1-\theta)v$ for some $\theta\in [0,\,1]$ and we continue the previous calculation as
\begin{align}
    \resizebox{.87\textwidth}{!}{$\displaystyle
    \begin{aligned}[b]
        \eqref{eq:proof.1} &= \left\vert\beta_0\int_{[0,\,1]^2}\int_{\R^{3d}}B(\abs{x_\ast-x})\nabla{\phi}(v)\left(1-\frac{\lambda+\lambda_\ast}{2}\right)\cdot(v_\ast-v)
            ff_\ast\,dx_\ast\,\dots\,d\lambda_\ast\right. \\
        &\phantom{=} +\beta_0\int_{[0,\,1]^2}\int_{\R^{3d}}B(\abs{x_\ast-x})\left(\nabla{\phi}(\tilde{v})-\nabla{\phi}(v)\right)\left(1-\frac{\lambda+\lambda_\ast}{2}\right)\cdot(v_\ast-v)
            ff_\ast\,dx_\ast\,\dots\,d\lambda_\ast \\
        &\phantom{=} \left.-\beta_0\cB_0\int_{[0,\,1]^2}\int_{\R^{2d}}\nabla{\phi}(v)\left(1-\frac{\lambda+\lambda_\ast}{2}\right)\cdot(v_\ast-v)ff_\ast\,dv\,\dots\,d\lambda_\ast\right\vert.
    \end{aligned}
    $}
    \label{eq:proof.2}
\end{align}
Invoking now the assumed smoothness of $f$ we write
\begin{equation}
    f(x_\ast,v_\ast,\lambda_\ast,t)=f(x,v_\ast,\lambda_\ast,t)+\nabla_x{f}(\tilde{x},v_\ast,\lambda_\ast,t)\cdot (x_\ast-x),
    \label{eq:f.expansion_x}
\end{equation}
where $\tilde{x}:=\vartheta x_\ast+(1-\vartheta)x$ for some $\vartheta\in [0,\,1]$, whence
\begin{align}
    \resizebox{0.93\linewidth}{!}{$\displaystyle
    \begin{aligned}[b]
        \eqref{eq:proof.2} &= \left\vert\beta_0\int_{[0,\,1]^2}\int_{\R^{3d}}B(\abs{x_\ast-x})\nabla{\phi}(v)\left(1-\frac{\lambda+\lambda_\ast}{2}\right)\cdot(v_\ast-v)(x_\ast-x)\cdot f\nabla_x\tilde{f}_\ast\,
            dx_\ast\,\dots\,d\lambda_\ast\right. \\
        &\phantom{=} +\beta_0\cB_0\int_{[0,\,1]^2}\int_{\R^{2d}}\left(\nabla{\phi}(\tilde{v})-\nabla{\phi}(v)\right)\left(1-\frac{\lambda+\lambda_\ast}{2}\right)\cdot(v_\ast-v)
            ff_\ast\,dv_\ast\,\dots,\,d\lambda_\ast \\
        &\phantom{=} +\left.\beta_0\int_{[0,\,1]^2}\int_{\R^{3d}}B(\abs{x_\ast-x})\left(\nabla{\phi}(\tilde{v})-\nabla{\phi}(v)\right)\left(1-\frac{\lambda+\lambda_\ast}{2}\right)\cdot(v_\ast-v)
            (x_\ast-x)\cdot f\nabla_x\tilde{f}_\ast\,dx_\ast\,\dots,\,d\lambda_\ast\right\vert.
    \end{aligned}
    $}
    \label{eq:proof.3}
\end{align}
Here, $\nabla_x\tilde{f}_\ast$ stands for $\nabla_xf(\tilde{x},v_\ast,\lambda_\ast,t)$ for brevity. Since $\phi$, $\nabla{\phi}$ are Lipschitz continuous, we may estimate $\abs{\nabla{\phi}(v)}\leq\Lip(\phi)$ and
\begin{align*}
    \abs{\nabla{\phi}(\tilde{v})-\nabla{\phi}(v)} &\leq \Lip(\nabla{\phi})\abs{\tilde{v}-v} \\
    &= \theta\Lip(\nabla{\phi})\abs{\bar{v}'-v}=\theta\beta_0\varepsilon\Lip(\nabla{\phi})\left(1-\frac{\lambda+\lambda_\ast}{2}\right)\abs{v_\ast-v} \\
    &\leq \beta_0\varepsilon\Lip(\nabla{\phi})\abs{v_\ast-v},
\end{align*}
where we have used~\eqref{eq:micro_v_cons} and we have taken into account that $0\leq\theta,\,1-\frac{\lambda+\lambda_\ast}{2}\leq 1$.

Similarly, writing $\nabla_xf(\tilde{x},v_\ast,\lambda_\ast,t)=\nabla_xf(x,v_\ast,\lambda_\ast,t)+\left(\nabla_xf(\tilde{x},v_\ast,\lambda_\ast,t)-\nabla_xf(x,v_\ast,\lambda_\ast,t)\right)$ we estimate
\begin{align*}
    \abs{\nabla_xf(\tilde{x},v_\ast,\lambda_\ast,t)-\nabla_xf(x,v_\ast,\lambda_\ast,t)} &\leq \Lip_x(\nabla_xf_\ast)\abs{\tilde{x}-x} \\
    &= \Lip_x(\nabla_xf_\ast)\vartheta\abs{x_\ast-x} \\
    &\leq \Lip_x(\nabla_xf_\ast)\abs{x_\ast-x},
\end{align*}
where $\Lip_x(\nabla_xf_\ast)$ denotes the Lipschitz constant of $\nabla_xf$ with respect to $x$, which depends in general on the other variables $v_\ast$, $\lambda_\ast$, $t$.

Employing these estimates in~\eqref{eq:proof.3} we discover:
\begin{align*}
    \eqref{eq:proof.3} &\leq \beta_0\cB_1\Lip(\phi)\int_{[0,\,1]^2}\int_{\R^{2d}}\abs{v_\ast-v}f\abs{\nabla_xf_\ast}\,dv\,\dots\,d\lambda_\ast \\
    &\phantom{\leq} +\beta_0\cB_2\Lip(\phi)\int_{[0,\,1]^2}\int_{\R^{2d}}\abs{v_\ast-v}f\Lip_x(\nabla{f}_\ast)\,dv\,\dots\,d\lambda_\ast \\
    &\phantom{\leq} +\varepsilon\beta_0^2\cB_0\Lip(\nabla{\phi})\int_{[0,\,1]^2}\int_{\R^{2d}}\abs{v_\ast-v}^2ff_\ast\,dv\,\dots\,d\lambda_\ast \\
    &\phantom{\leq} +\varepsilon\beta_0^2\cB_1\Lip(\nabla{\phi})\int_{[0,\,1]^2}\int_{\R^{2d}}\abs{v_\ast-v}f\abs{\nabla_xf_\ast}\,dv\,\dots\,d\lambda_\ast \\
    &\phantom{\leq} +\varepsilon\beta_0^2\cB_2\Lip(\nabla{\phi})\int_{[0,\,1]^2}\int_{\R^{2d}}\abs{v_\ast-v}^2f\Lip_x(\nabla_xf_\ast)\,dv\,\dots\,d\lambda_\ast,
\end{align*}
where
$$ \cB_k:=\int_{\sB_R(0)}{\abs{x}}^kB(\abs{x})\,dx, \qquad k\in\mathbb{N}. $$
By switching to polar coordinates in $\R^d$ and invoking the boundedness of the collision kernel $B$ we estimate, in particular,
\begin{align*}
    \cB_k &= \int_{\partial\sB_1(0)}\int_0^Rr^kB(r)\,r^{d-1}dr\,d\cH^{d-1} \\
    &\leq \norm{B}_\infty\cH^{d-1}(\partial\sB_1(0))\int_0^Rr^{k+d-1}\,dr \\
    &= \norm{B}_\infty\cH^{d-1}(\partial\sB_1(0))\frac{R^{k+d}}{k+d}\lesssim R^{k+d},
\end{align*}
$\cH^{d-1}$ being the $(d-1)$-dimensional Hausdorff measure in $\R^d$.

Due to the assumption of $x$-uniform compactness in $v$ of the support of $f$, all integral terms in the estimate above are finite (notice, in particular, that $\abs{\nabla_xf_\ast}$ and $\Lip_x(\nabla_xf_\ast)$ vanish uniformly in $x$ for $\abs{v}$ large enough). Therefore\footnote{We use the notation $a\lesssim b$ to mean that, given $a,\,b\geq 0$, there exists a constant $C>0$, whose specific value is unimportant, such that $a\leq Cb$.}:
$$ \lim_{\varepsilon\to 0^+}\abs{\ave{\frac{1}{\varepsilon}\bar{Q}^{v,\varepsilon}(f,f)-\cP(f,f),\,\phi}}\lesssim\cB_1+\cB_2\lesssim R^{1+d}+R^{2+d}
    \sim\cO(R^{1+d}) $$
for $R$ small.

\item As far as the approximation of the operator $\bar{Q}^\lambda$ is concerned, by relying again on expansion~\eqref{eq:f.expansion_x} we have:
\begin{align*}
    & \abs{\ave{\bar{Q}^\lambda(f,f)-\bar{Q}^\lambda_\loc(f,f),\,\psi}} \\
    &\quad =\abs{\int_{[0,\,1]^2}\int_{\R^{3d}}B(\abs{x_\ast-x})\bigl(\psi(\bar{\lambda}')-\psi(\lambda)\bigr)
        (x_\ast-x)\cdot f\nabla_x\tilde{f}\,dx_\ast\,\dots\,d\lambda_\ast} \\
    &\quad \leq\cB_1\int_{[0,\,1]^2}\int_{\R^{2d}}\abs{\psi(\bar{\lambda}')-\psi(\lambda)}f\abs{\nabla{f_\ast}}\,dv\,\dots\,d\lambda_\ast \\
    &\qquad +\cB_2\int_{[0,\,1]^2}\int_{\R^{2d}}\abs{\psi(\bar{\lambda}')-\psi(\lambda)}f\Lip_x(\nabla{f_\ast})\,dv\,\dots\,d\lambda_\ast \\
    &\quad \leq\norm{\psi}_\infty\left(\cB_1\int_{[0,\,1]^2}\int_{\R^{2d}}f\abs{\nabla{f_\ast}}\,dv\,\dots\,d\lambda_\ast
        +\cB_2\int_{[0,\,1]^2}\int_{\R^{2d}}f\Lip(\nabla{f_\ast})\,dv\,\dots\,d\lambda_\ast\right) \\
    &\quad \lesssim \cB_1+\cB_2\sim\cO(R^{1+d})
\end{align*}
for $R$ small. \qedhere
\end{enumerate}
\end{proof}

\begin{remark}\label{remark_R}
Theorem~\ref{theo:local_approx_Q} ensures that the considered approximations of $\bar{Q}^{v,\varepsilon}$ and $\bar{Q}^\lambda$ are consistent with each other, although they have been obtained in partly different ways due to the structural dependence of $\bar{Q}^{v,\varepsilon}$, unlike $\bar{Q}^\lambda$, on $\varepsilon$. Because of this, and notwithstanding the same order of approximation in $R$ of the operators, we expect different accuracy among the macroscopic quantities $\rho$, $u$, $l$ provided by the macroscopic model \eqref{eq:macro} compared to the same quantities computed as statistical moments of the kinetic distribution function $f$. The reason is that the macroscopic model is obtained in the hydrodynamic limit by closing system~\eqref{eq:hydro.eps} with $f^0$~\eqref{eq:f0.complete}, which is determined approximately in the local regime of small $R$. First, we observe that $\rho$ evolves due to the free particle drift only, hence the approximation in $R$ of $f^0$ does not affect it. Second, $l$ depends mainly on the marginal distribution $n^0$, which is found in local regime of small $R$. Finally, and differently, the fact that $\bar{Q}^{v,\varepsilon}\to 0$ as $\varepsilon\to 0^+$ implies primarily that $p^0$ is the initial condition $p(x,v,0)$ of the pure interaction dynamics. Next, for $R$ small and in the high frequency regime ($1/\varepsilon$ with $\varepsilon\to 0^+$), $p^0$ is approximated by $\delta(v-u)$. This approximation, determined by the microscopic dynamics, is also consistent with the classical choice of the monokinetic closure in the derivation of the Euler equations from the non-collisional Vlasov equation. As a consequence, if one prescribes actually $p(x,v,0)=\delta(v-u(x,0))$ then the evolution of $u$ yielded by \eqref{eq:macro_b} is not affected by $R$ unlike $u$ computed as the mean of $f$ solving~\eqref{eq:Boltz_weak}.   
\end{remark}

\section{Numerical tests}\label{sect:numerics}
In this section we present some numerical tests in 1D and 2D. Tests in 1D are mainly dedicated to the numerical confirmation of the expected match between the micro- and the macro-scale. Conversely, tests in 2D are devoted to show the ability of the model to reproduce some patterns of interest, like merge and split of flocks, cf.\ \cite{cristiani2021JMB}. 

For the reader's convenience, we start detailing the numerical schemes we used for simulations. 

\subsection{Numerical approximation}\label{scheme}
In this section we present the numerical schemes used to discretize the three equations for $\rho$, $u$ and 
%$\bar\lambda$ 
$l$ in \eqref{eq:macro.noncons} in 2D. The 1D case can be easily derived by the 2D case. 
We face a system coupling a conservation law, a Burgers' equation with source term, and an advection equation with variable velocity and source term.
While these equations are standard in the literature, and many numerical schemes are available for all of them, the combination of the three numerical schemes is not trivial since stability issues can arise. In addition, the effect of boundary conditions is not at all negligible and must be carefully taken into account.

\subsubsection{Numerical grid}
Let us denote by $(x^1,x^2)$ the two components of the space vector $x$.
 The computational domain $\Omega\times[0,T]$, for some $\Omega\subset \R^2$ and final time $T>0$, is divided in cells of side $\Delta x^1 \times \Delta x^2 \times \Delta t$, where $\Delta x^1$, $\Delta x^2$ are the space steps in the two dimensions and $\Delta t$ is the time step, respectively. 
Let us assume that $T$ is a multiple of $\Delta t$ to avoid rounding, and define $n_T:=\frac{T}{\Delta t}\in\mathbb N$. 
Similarly, we choose $\Omega$ as a rectangular domain of size $n^1\Delta x^1\times n^2\Delta x^2$, with $n^1, n^2\in\mathbb N$.
The generic space cell and space-time cell are defined as 

$$
C_{i,j}:=
\left[x^1_{i}-\frac{\Delta x^1}{2}, x^1_{i}+\frac{\Delta x^1}{2}\right) \times 
\left[x^2_{j}-\frac{\Delta x^2}{2}, x^2_{j}+\frac{\Delta x^2}{2}\right), 
\qquad
C^n_{i,j}:=C_{i,j} \times \left[t^n, t^{n+1}\right)
$$ 

where, as usual,
$$
\left\{x^1_0,\ldots,x^1_{n^1}\right\}, \qquad
\left\{x^2_0,\ldots,x^2_{n^2}\right\}, \qquad
\{t^0,\ldots,t^{n_T}\}$$
are the grid points.

Finally, we denote by $\rho^n_{i,j}$ the approximation of the \emph{integral mean} over a cell $C_{i,j}$ at time $t^n$,
$$
\rho^n_{i,j}\approx \frac{1}{\Delta x^1\Delta x^2}\int_{C_{i,j}}\rho(x,t^n) dx,
$$
while we denote by $u^n_{i,j}$ and $l^n_{i,j}$ the \emph{pointwise} approximation of $u$ and $l$ at grid nodes
$$
u^n_{i,j}\approx u(x^1_{i},x^2_{j},t^n), \qquad l^n_{i,j}\approx l(x^1_{i},x^2_{j},t^n).
$$

\subsubsection{Equation for~\texorpdfstring{$\boldsymbol{\rho}$}{}}
For the conservation law we use the push-forward scheme proposed in \cite{piccoliARMA2011} and then used in, e.g., \cite{cristiani2023CMS, cristiani2011MMS}. 
It is a natural generalization of the 1D upwind scheme and, although it exhibits a slightly diffusive behavior, it is stable and conservative even with sign-changing velocity fields (under the CFL condition).
The push-forward scheme is explicit in time and truly two-dimensional. 
Also, it follows the physics of the underlying problem: at each time step, some mass leaves the cell $C_{i,j}$, while other mass coming from neighboring cells enters $C_{i,j}$. A CFL condition of the form
$$ \Delta{t}\max_{i,j,n}\abs{u_{i,j}^n}\leq\min\{\Delta{x^1},\,\Delta{x^2}\} $$
is imposed to avoid that mass covers a distance larger than a cell in one time step.

The balance of mass is given by \cite[Sect.\ 5.5.2]{cristiani2014book}
\begin{equation}\label{schema_push_forward}
\rho^{n+1}_{i,j} = \frac{1}{\Delta x^1 \Delta x^2}
\sum\limits_{r,s}
\rho^{n}_{i,j}\ \mathcal L\big(C_{i,j}\cap\gamma^n(C_{r,s})\big), \qquad n=0,\ldots,n_T-1,
\end{equation}
where $\mathcal L(\mathcal K)$ denotes the Lebesgue measure of a subset $\mathcal K\subseteq\R^2$, 
and 
$$
\gamma^n(\zeta):=\zeta+u(\zeta, n\Delta t)\Delta t,\qquad \zeta\in\R^2
$$
is the discrete-in-time map which pushes the mass forward in space. 

The scheme \eqref{schema_push_forward} can be written in a more computer-friendly form as follows: 
let us denote by $\delta_{i,j}$ the standard Kronecker delta and by $(\phantom{x})^\pm$ the positive/negative part. Let us also define $X:=\Delta t \ u$ and denote by $X^1, X^2$ are the two components of the vector $X$.
Then, the scheme \eqref{schema_push_forward} can be written as
$$ \rho^{n+1}_{i,j} =\frac{1}{\Delta x^1\Delta x^2}\sum_{\substack{\abs{r-i}\leq 1 \\ \abs{s-j}\leq 1}} \rho^{n}_{r,s} \ \Gamma_{r,s}^{1,n} \ \Gamma_{r,s}^{2,n},
    \qquad n=0,\ldots,n_T-1, $$
where
\begin{eqnarray*}
\Gamma_{r,s}^{1,n} := 
\left(X_{r,s}^{1,n}\right)^+ \delta_{r,i-1} + 
\left(X_{r,s}^{1,n}\right)^- \delta_{r,i+1} + 
\left(\Delta x^1-|X_{r,s}^{1,n}|\right)\delta_{r,i}, \\ [2mm]
\Gamma_{r,s}^{2,n} := 
\left(X_{r,s}^{2,n}\right)^+ \delta_{s,j-1} + 
\left(X_{r,s}^{2,n}\right)^- \delta_{s,j+1} + 
\left(\Delta x^2-|X_{r,s}^{2,n}|\right)\delta_{s,j}.
\end{eqnarray*}
\subsubsection{Equation for~\texorpdfstring{$\boldsymbol{u}$}{}}
Let us denote by $G_u[\rho,u,l]$ the right-hand side of equation \eqref{eq:macro.noncons_b}.
The vector equation \eqref{eq:macro.noncons_b} is explicitly written as
$$
\left\{
\begin{array}{l}
     \partial_t u^1 + u^1 \partial_{x^1}u^1 + u^2 \partial_{x^2}u^1 = G_u^1[\rho,u,l],  \\ [2mm]
     \partial_t u^2 + u^1 \partial_{x^1}u^2 + u^2 \partial_{x^2}u^2 = G_u^2[\rho,u,l].
\end{array}
\right.
$$
For numerical purposes, it is useful to add a small diffusion term as follows, this will result in an increased numerical stability. Doing this we get
$$
\left\{
\begin{array}{l}
     \partial_t u^1 + u^1 \partial_{x^1}u^1 + u^2 \partial_{x^2}u^1 = G_u^1[\rho,u,l] + D  \triangle u^1,  \\ [2mm]
     \partial_t u^2 + u^1 \partial_{x^1}u^2 + u^2 \partial_{x^2}u^2 = G_u^2[\rho,u,l] + D \triangle u^2.
\end{array}
\right.
$$
for some $D>0$ small.
Denoting by $\widetilde G_u$ the discretized right-hand side obtained by a first-order quadrature formula, and using finite differences to approximate derivatives, we get the following explicit numerical scheme:
$$
\left\{
\begin{array}{l}
     u^{1,n+1}_{i,j} = u^{1,n}_{i,j} 
     - \frac{\Delta t}{2\Delta x^1}\ u^{1,n}_{i,j}\ (u^{1,n}_{i+1,j}-u^{1,n}_{i-1,j})
     - \frac{\Delta t}{2\Delta x^2}\ u^{2,n}_{i,j}\ (u^{1,n}_{i,j+1}-u^{1,n}_{i,j-1})
     + 
     \\ [2mm] \hspace{2cm}
     \Delta t \ \widetilde G^1_u[\rho^n,u^n,l^n](x_i,y_j,t^n) + 
    % \\ [2mm] \hspace{2cm}
     D \frac{\Delta t}{\Delta x^1 \Delta x^2}(u^{1,n}_{i-1,j}-2u^{1,n}_{i,j}+u^{1,n}_{i+1,j})+
     \\ [2mm] \hspace{2cm}
      D \frac{\Delta t}{\Delta x^1 \Delta x^2}(u^{1,n}_{i,j-1}-2u^{1,n}_{i,j}+u^{1,n}_{i,j+1}),  \\ [4mm]
     u^{2,n+1}_{i,j} = u^{2,n}_{i,j} 
     - \frac{\Delta t}{2\Delta x^1}\ u^{1,n}_{i,j}\ (u^{2,n}_{i+1,j}-u^{2,n}_{i-1,j})
     - \frac{\Delta t}{2\Delta x^2}\ u^{2,n}_{i,j}\ (u^{2,n}_{i,j+1}-u^{2,n}_{i,j-1})+
     \\ [2mm] \hspace{2cm}
      \Delta t \ \widetilde G^2_u[\rho^n,u^n,l^n](x_i,y_j,t^n) + 
     %\\ [2mm] \hspace{2cm}
     D \frac{\Delta t}{\Delta x^1 \Delta x^2}(u^{2,n}_{i-1,j}-2u^{2,n}_{i,j}+u^{2,n}_{i+1,j})+
     \\ [2mm] \hspace{2cm}
      D \frac{\Delta t}{\Delta x^1 \Delta x^2}(u^{2,n}_{i,j-1}-2u^{2,n}_{i,j}+u^{2,n}_{i,j+1}).
\end{array}
\right.
$$

In order to further increase stability, we have actually used the \emph{implicit} version of the previous scheme, obtained considering the time step $n+1$ instead on $n$ in the two right-hand sides, thus obtaining two coupled linear systems of size $n^1 n^2 \times n^1 n^2$ with unknowns $\{u^{1,n+1}\}_{i,j}$ and $\{u^{2,n+1}\}_{i,j}$, respectively.

\subsubsection{Equation for~\texorpdfstring{$\boldsymbol{l}$}{}}
For the third equation we employ an explicit first-order semi-Lagrangian scheme \cite[Sect.\ 5.1.3]{falcone2013book}. 
Let us denote by $G_{l}[\rho,l]$ the right-hand side of equation \eqref{eq:macro.noncons_c}, and by $\widetilde G_{l}$ its discretization obtained by means of a first-order quadrature formula. 
Then the scheme reads as
$$
l^{n+1}_{i,j}=I[l^{n}]\big((x^1_i,x^2_j)-\Delta t\ u^n_{i,j}\big)+\Delta t \ \widetilde G_{l}[\rho^n,l^n](x^1_{i},x^2_{j},t^n)
$$
where, for any $\zeta\in\R^2$, $I[l^{n}](\zeta)$ is the approximated value of $l(\zeta, t^n)$ obtained by interpolation using available values of $l^n$ at grid nodes.
For our purposes, we considered a bilinear interpolation which uses the values of the four closest vertices to the point $\zeta$.
More precisely, assuming that the point $\zeta$ is enclosed by the grid nodes $(x_{i},y_{j}), (x_{i+1},y_{j}), (x_{i},y_{j+1}), (x_{i+1},y_{j+1})$, for any function $\omega=\omega(x,y)$ and point $\zeta\in\R^2$, we have
$$
I_{\text{bil}}[\omega](\zeta)=
\pi_1(\zeta) \ \omega(x^1_{i},x^2_{j})+
\pi_2(\zeta) \ \omega(x^1_{i+1},x^2_{j})+
\pi_3(\zeta) \ \omega(x^1_{i+1},x^2_{j+1})+
\pi_4(\zeta) \ \omega(x^1_{i},x^2_{j+1})
$$
with 
$$
\begin{array}{ll}
   \pi_1(\zeta):=\frac{(x^1_{i+1}-\zeta^1)(x^2_{j+1}-\zeta^2)}{\Delta x^1 \Delta x^2},
   &
   \pi_2(\zeta):=\frac{(\zeta^1-x^1_{i})(x^2_{j+1}-\zeta^2)}{\Delta x^1 \Delta x^2}, \\ [3mm]
   \pi_3(\zeta):=\frac{(\zeta^1-x^1_{i})(\zeta^2-x^2_j)}{\Delta x^1 \Delta x^2}, 
   &
   \pi_4(\zeta):=\frac{(x^1_{i+1}-\zeta^1)(\zeta^2-x^2_j)}{\Delta x^1 \Delta x^2}.
\end{array}
$$

\subsection{Comparison of the microscopic and macroscopic model in 1D}
In this section we compare the results of the numerical integration of the microscopic stochastic process and of the macroscopic model in 1D. In particular, we integrate the microscopic model \eqref{def:micro_drift}-\eqref{def:Sigma} with a direct Monte Carlo algorithm as done, e.g., in \cite{loy2021KRM}, with $N=10^5$ particles. The macroscopic model is integrated according to the scheme illustrated in Section \ref{scheme}. 
Model parameters are summarized in Table \ref{tab:allparameters}.
\begin{table}[!t]
\centering
 \caption{Choice of parameters}
    \begin{tabular}{|l|c|c|c|c|c|c|c|c|c|}
        \hline
	Test & $\alpha$ & $\beta$ & $\gamma$ & $\eta$ & $\mu$ & $\nu$ & $R$ \\
	\hline\hline
  	1D & 0.01 & 0.5 & 1 & 1 & 1 & 0.8 & varied \\
        \hline
        2Da & 0.0225 & 0.5 & 0.5 & 0.05 & 0.5 & 0.8 & 0.3 \\
        \hline	
        2Db & 0.01 & 0.1 & 1.3 & 0.3 & 1.5 & 0.2 & 0.25 \\
        \hline
        2Dc & 0.01 & 1 & 0.4 & 2 & 3 & 1 & 0.4 \\
	\hline
    \end{tabular}
\label{tab:allparameters}
\end{table}

Simulations at both scales are performed over the non-dimensional spatial domain $I= \left[0, 1/R\right]$, with final time $T=5$.
We approximate the solution of \eqref{eq:macro.noncons} starting from the following initial conditions: 
\begin{equation*}
\begin{aligned}  
\rho(x,0)&= \frac{1}{\sqrt{2\pi \sigma^2}} \exp\left({-\frac{(x-x_0)^2}{2\sigma^2}}\right),\\
u(x,0)&=0,\\
%u^1(t^0,x)&=0\\
%u^2(t^0,x)&=0\\
l (x,0)&= 0 ,
\nonumber
\end{aligned}
\end{equation*}
for $x \in I $. Here $x_0=\frac{1}{2R}$ and $\sigma=\frac{0.1}{R}$. Null Dirichlet boundary conditions are considered. We have chosen $B(x_\ast-x)=\chi_{||x_\ast-x||\le R}$.

\begin{table}[!t]
\centering
\caption{$L^2$-distance of the difference between microscopic and macroscopic $\rho$, $\rho u$ and $\rho l $.}
\label{tab:L2_norm}
\subtable[$\rho$]{
    \begin{tabular}{|l||c|c|}
    \hline
    & $R=0.02$ & $R=0.01$ \\
    \hline\hline
    $\varepsilon=10^{-3}$ & 0.3894 & 0.2112 \\
    \hline
    $\varepsilon=10^{-4}$ & 0.1232 & 0.1167 \\
    \hline
    \end{tabular}
}
\subtable[$\rho u$]{
    \begin{tabular}{|c|c|}
    \hline
    $R=0.02$ & $R=0.01$ \\
    \hline\hline
    0.0084 & 0.0082 \\
    \hline
    0.00114 & 0.00110 \\
    \hline
    \end{tabular}
}
\subtable[$\rho l$]{
    \begin{tabular}{|c|c|}
    \hline
    $R=0.02$ & $R=0.01$ \\
    \hline\hline
    0.0908 & 0.0680 \\
    \hline
    0.0742 & 0.0675 \\
    \hline
    \end{tabular}}
\end{table}

In order to quantify the agreement between the numerical results, we have computed the $L^2$-distance between microscopic and macroscopic density $\rho$ and first order moments of $f$ in $\lambda$ and $v$, i.e.\ $\rho u$ and $\rho l$, at final time, for some values of $R$ and $\varepsilon$; see Table \ref{tab:L2_norm}. As expected, for fixed $R$ the discrepancy decreases as $\varepsilon$ decreases, and, for fixed $\varepsilon$ the discrepancy decreases as $R$ decreases.

In Figure~\ref{fig1} we show the profiles of the solution $(\rho,\,\rho u,\,\rho l)$ at the final computational time in the cases considered in Table \ref{tab:L2_norm}.
The numerical results highlight the fact that a small $\varepsilon$ is particularly effective in matching microscopic and macroscopic simulations.

\begin{figure}[!ht]
\centering
\subfigure[$R=0.02, \varepsilon=10^{-3}$]{\includegraphics[width=0.49\textwidth]{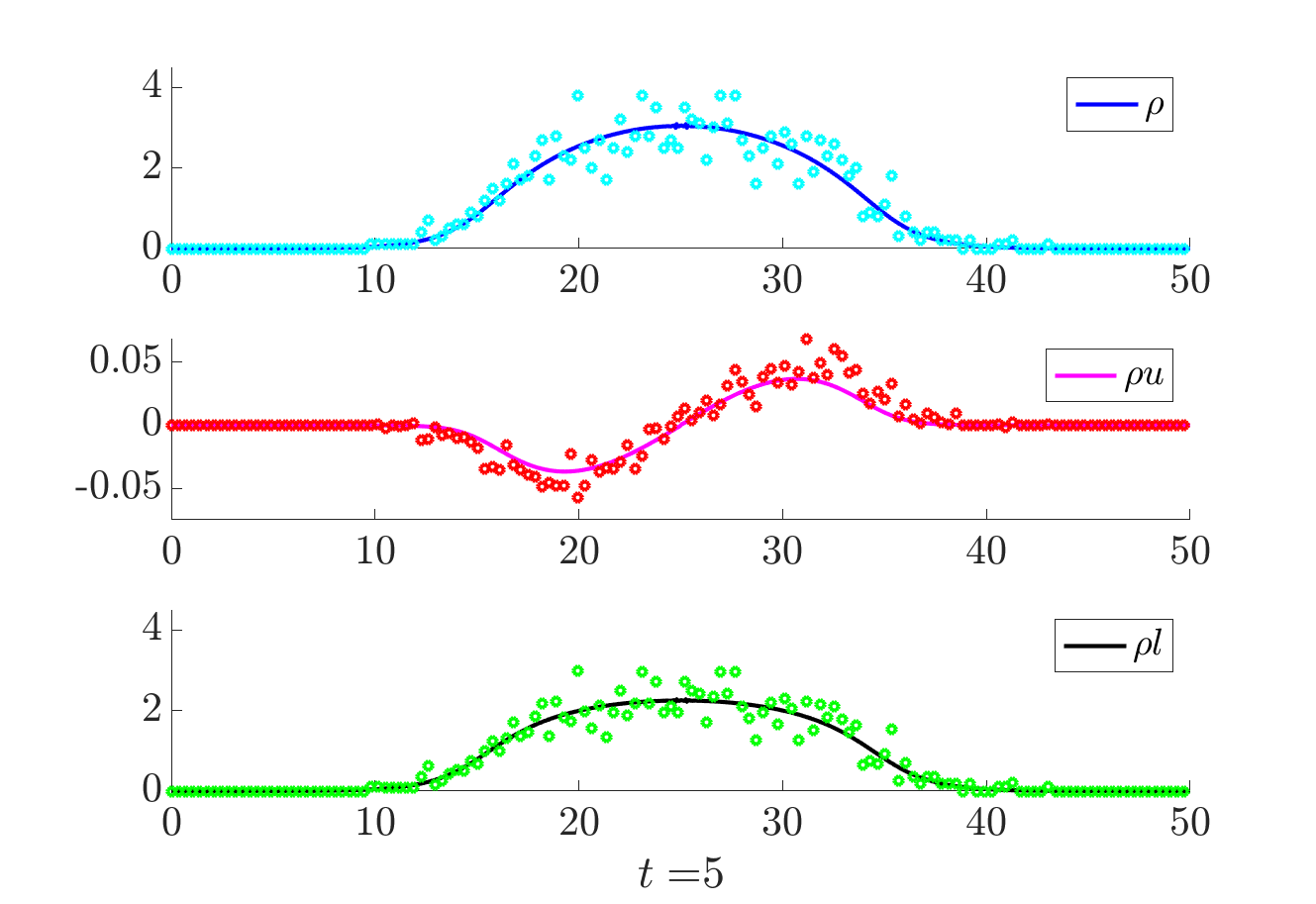}}
\subfigure[$R=0.01, \varepsilon=10^{-3}$]{\includegraphics[width=0.49\textwidth]{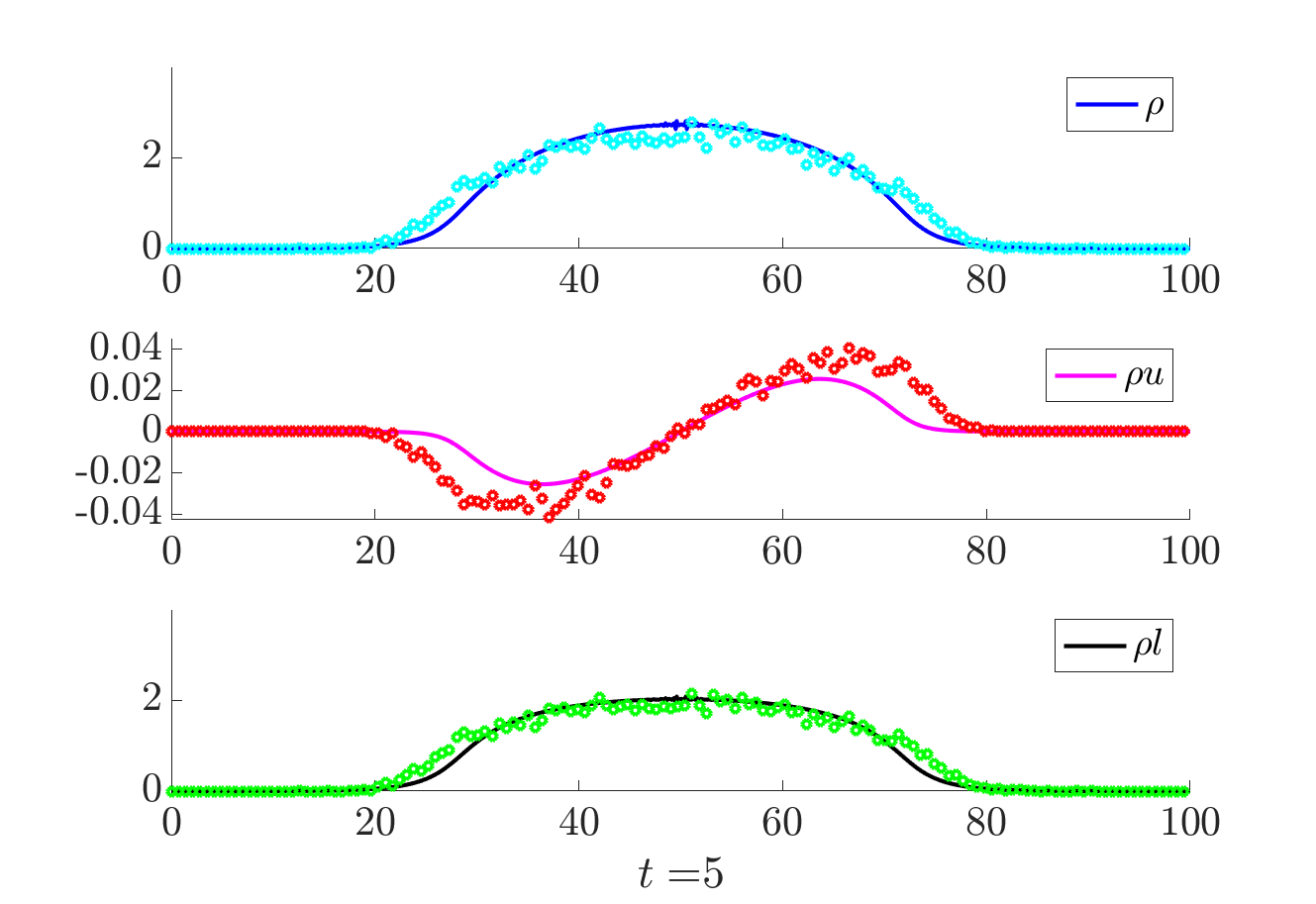}}\\
\subfigure[$R=0.02, \varepsilon=10^{-4}$]{\includegraphics[width=0.49\textwidth]{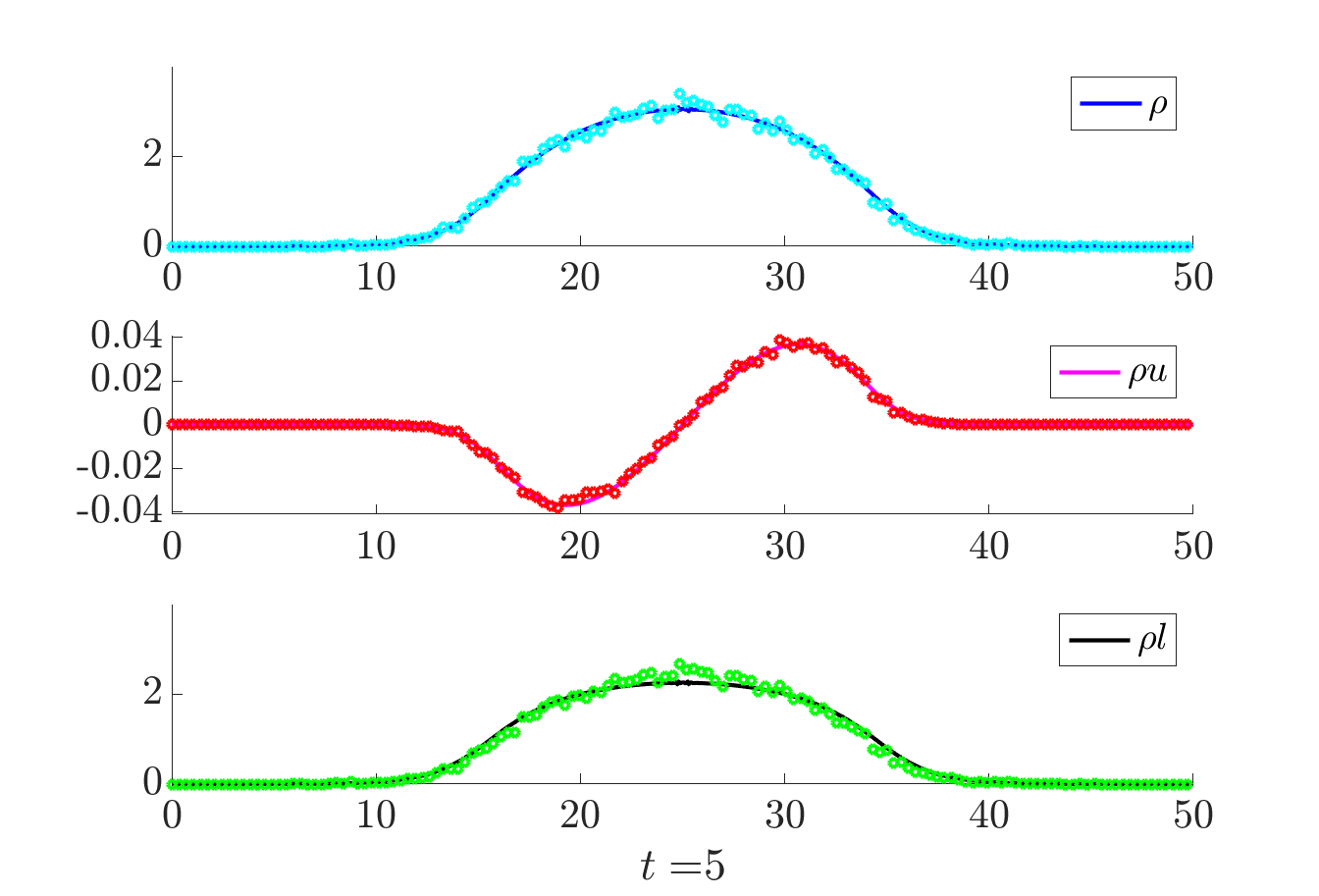}}
\subfigure[$R=0.01, \varepsilon=10^{-4}$]{\includegraphics[width=0.49\textwidth]{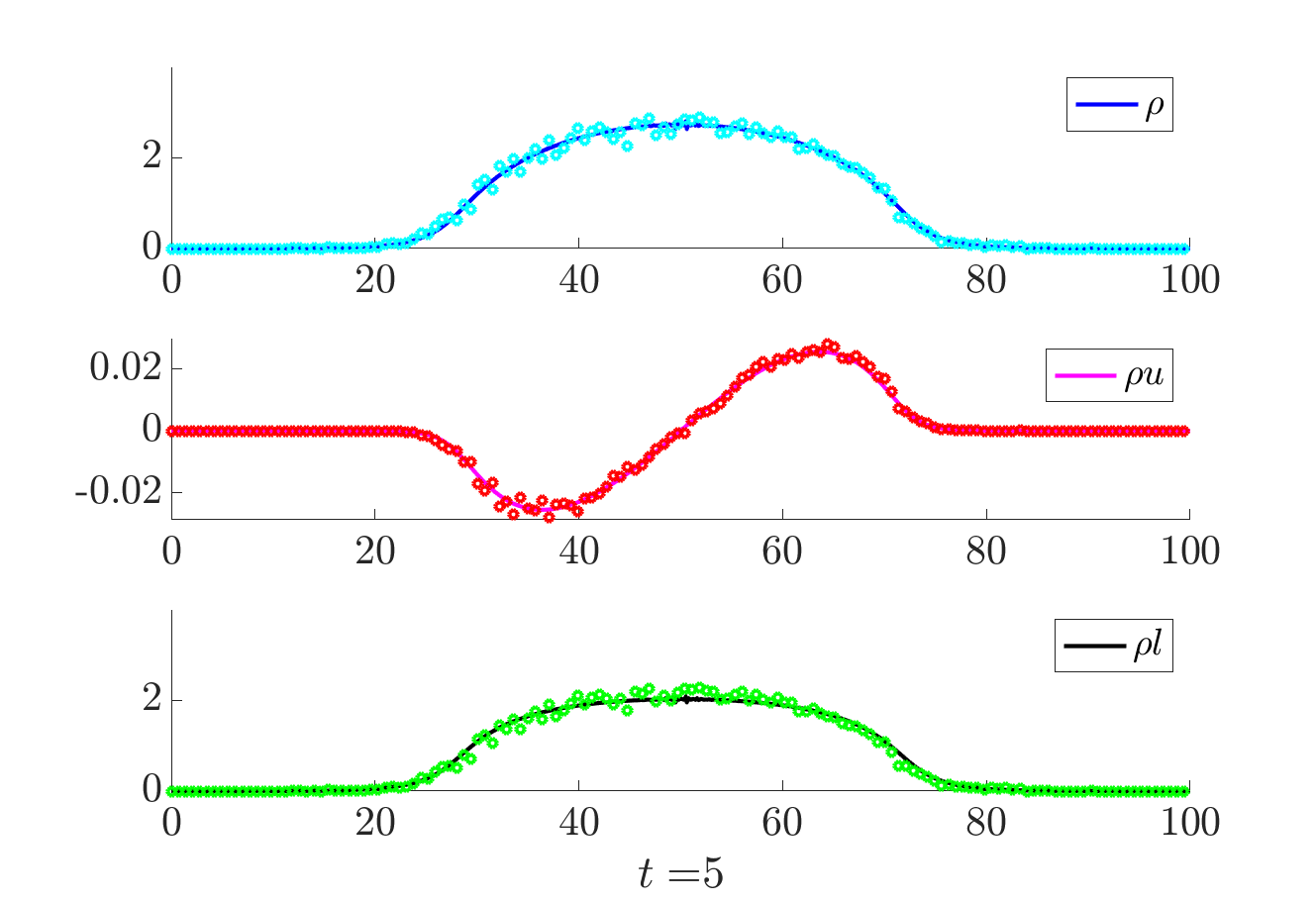}}
\caption{Test 1D. Comparison between the numerical results of the macroscopic model \eqref{eq:macro.noncons} with null Dirichlet boundary conditions (continuous lines) and of the microscopic stochastic process \eqref{def:micro_drift}-\eqref{def:Sigma} (`o' markers). A small numerical instability can be observed in the approximation of $\rho$ when the velocity changes sign and this is due to the upwind scheme used for the continuity equation.
}
\label{fig1}
\end{figure}

\subsection{Numerical tests in 2D}
In this section we present the results of numerical simulations of the macroscopic model \eqref{eq:macro.noncons} in a two-dimensional domain. The values of the model parameters are reported in Table \ref{tab:allparameters}.

\subsubsection{Test 2Da: Turning and split}
In this test we consider a single flock initially in a steady state, which starts moving and elongates. To this end we consider a square domain $\Omega= \left[0, 2\right] \times \left[0, 2\right]$, a final time $T=300$, and we set initial data as follows:
\begin{equation*}
\begin{aligned}  
\rho(x,0)&=  \exp\left({-\frac{(x^1-C^1_{0})^2}{2\sigma_0^2}-\frac{(x^2-C^2_{0})^2}{2\sigma_0^2}}\right),\\
u(x,0)&=0,\\
%u^1(t^0,x)&=0\\
%u^2(t^0,x)&=0\\
l (x,0)&= 0.9 \exp\left({\frac{-(x^1-C^1_{1})^2}{2\sigma_1^2}-\frac{(x^2-C^2_{1})^2}{2\sigma_1^2}}\right)+ 0.8 \exp\left({\frac{-(x^1-C^1_{2})^2}{2\sigma_2^2}-\frac{(x^2-C^2_{2})^2}{2\sigma_2^2}}\right) ,
\nonumber
\end{aligned}
\end{equation*}
for any $x=(x^1,x^2) \in \Omega $. Here 
$C_{0}=(C^1_{0},C^2_{0})=(1,1)$, 
$C_{1}=(C^1_{1},C^2_{1})=(0.8,0.8)$,
$C_{2}=(C^1_{2},C^2_{2})=(1.3,1.3)$,
$\sigma_0=\sqrt{0.03}$,
$\sigma_1=\sigma_2=\sqrt{0.02}$.
Figure \ref{fig:test2_ci} shows the initial conditions for $\rho$ and $l$.
Note that at initial time the flock has null velocity, this means that the initial condition for $l$ will play a crucial role in the dynamics.
\begin{figure}[!t]
\centering
\subfigure[]{\includegraphics[scale=0.3]{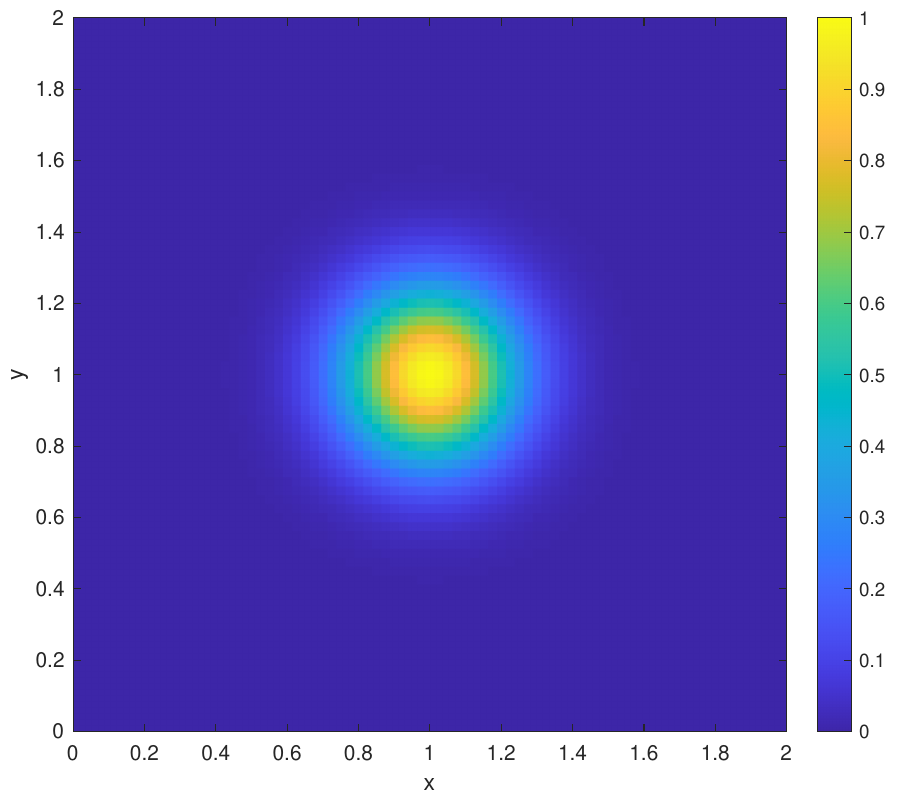}}
\subfigure[]{\includegraphics[scale=0.3]{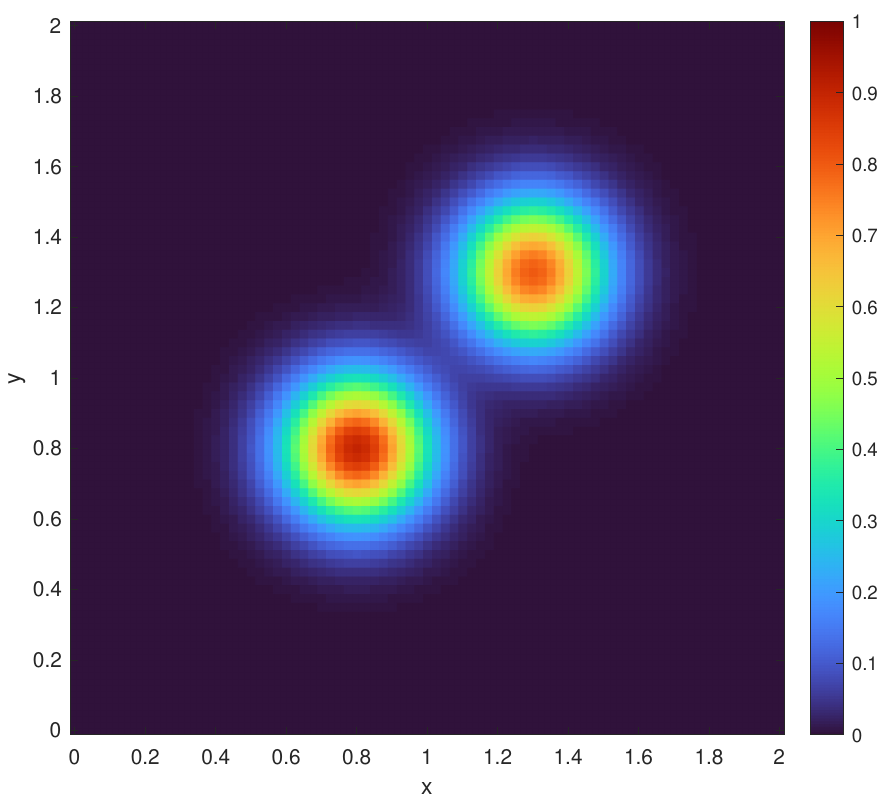}}
\caption{Test 2Da: Initial conditions for (a) $\rho$ and (b) $l$.}
\label{fig:test2_ci}
\end{figure}

Boundary conditions are of Dirichlet type, with $\rho=u=l=0$.
For the numerical approximation we set $\Delta x^1=\Delta x^2=0.025$ and $\Delta t= 0.1$.

Figure \ref{fig:test2} shows three snapshots of the evolution. At the beginning, the group starts moving in south-west direction, influenced by a higher degree of leadership in that region.
After a while, we observe a turning behavior, with some agents which start moving in the opposite direction. As a result, the flock elongates and the interior density decreases.
\begin{figure}[!t]
\centering
\subfigure[]{\includegraphics[scale=0.3]{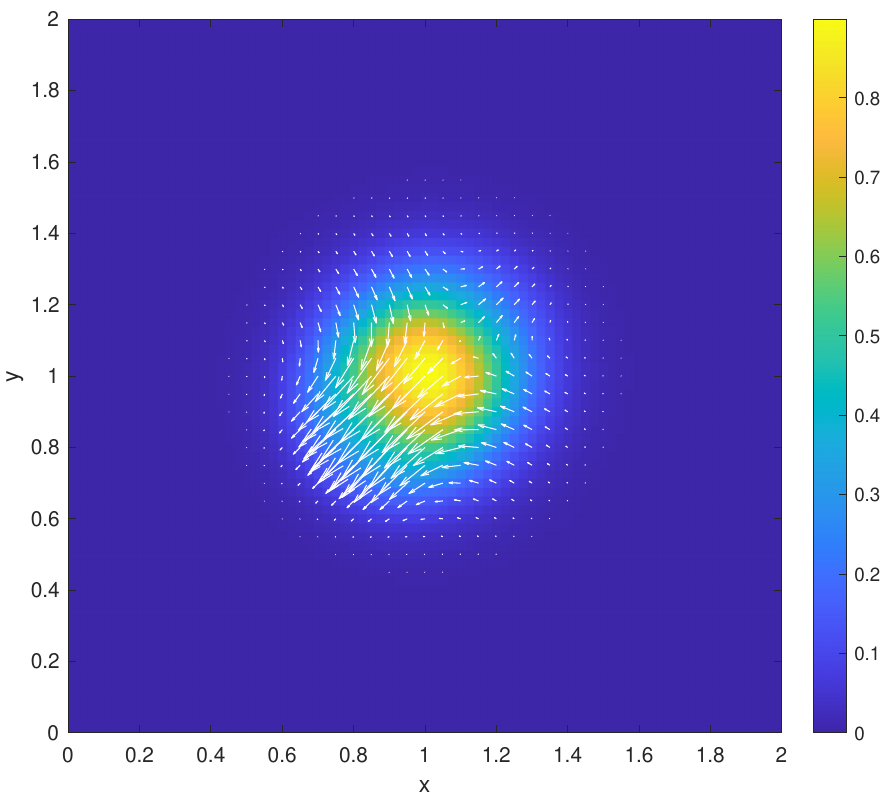}}
\subfigure[]{\includegraphics[scale=0.3]{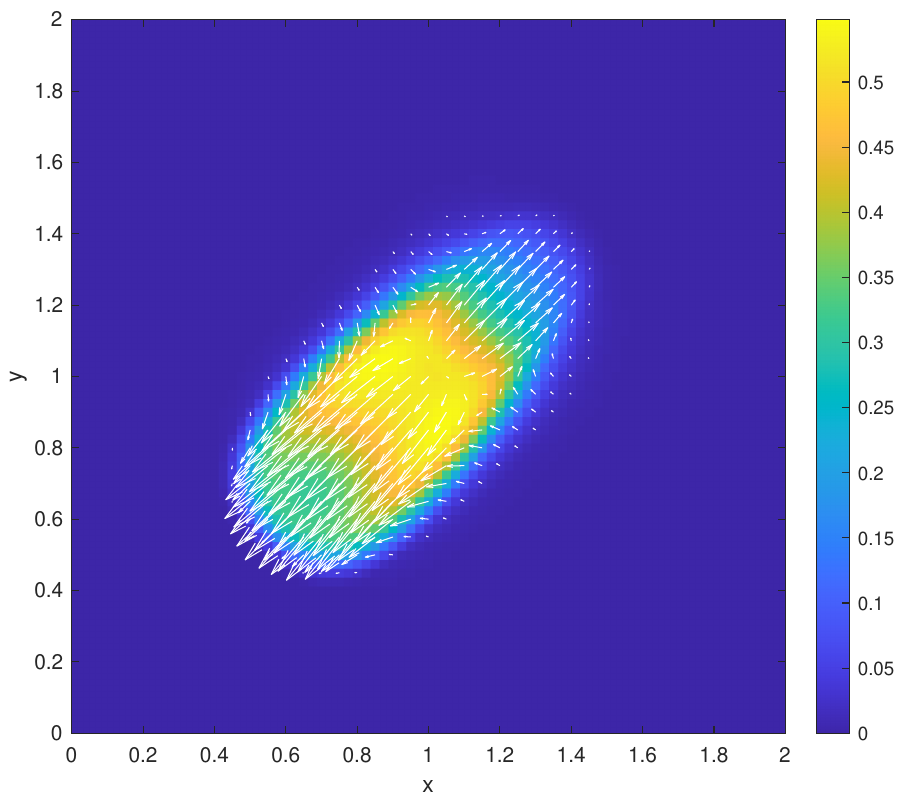}}
\subfigure[]{\includegraphics[scale=0.3]{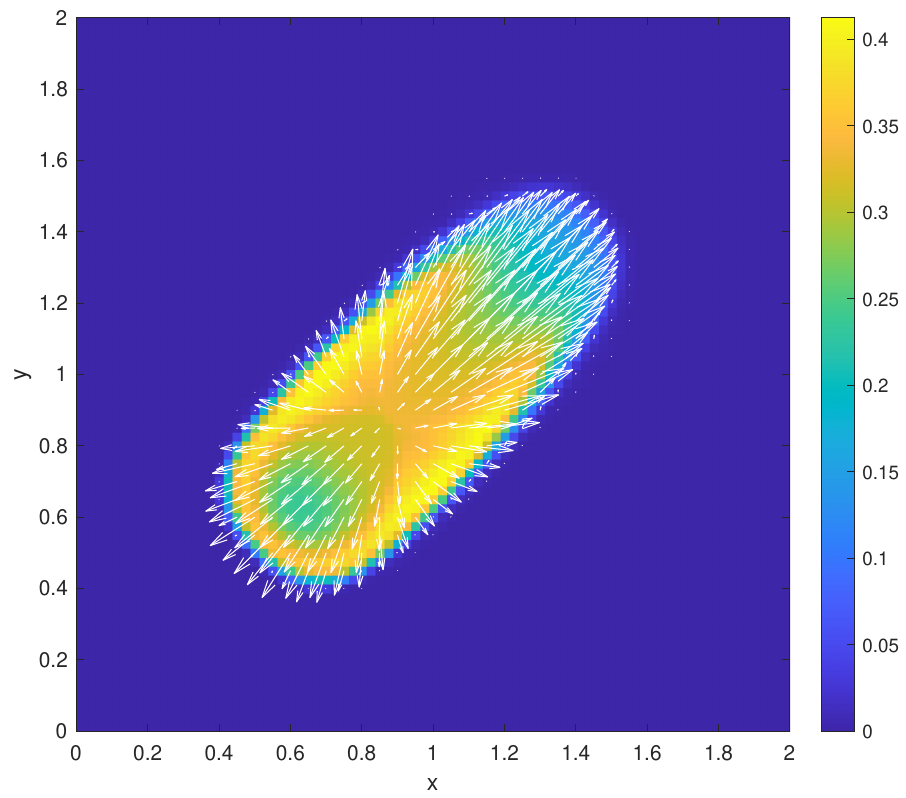}}
\subfigure[]{\includegraphics[scale=0.3]{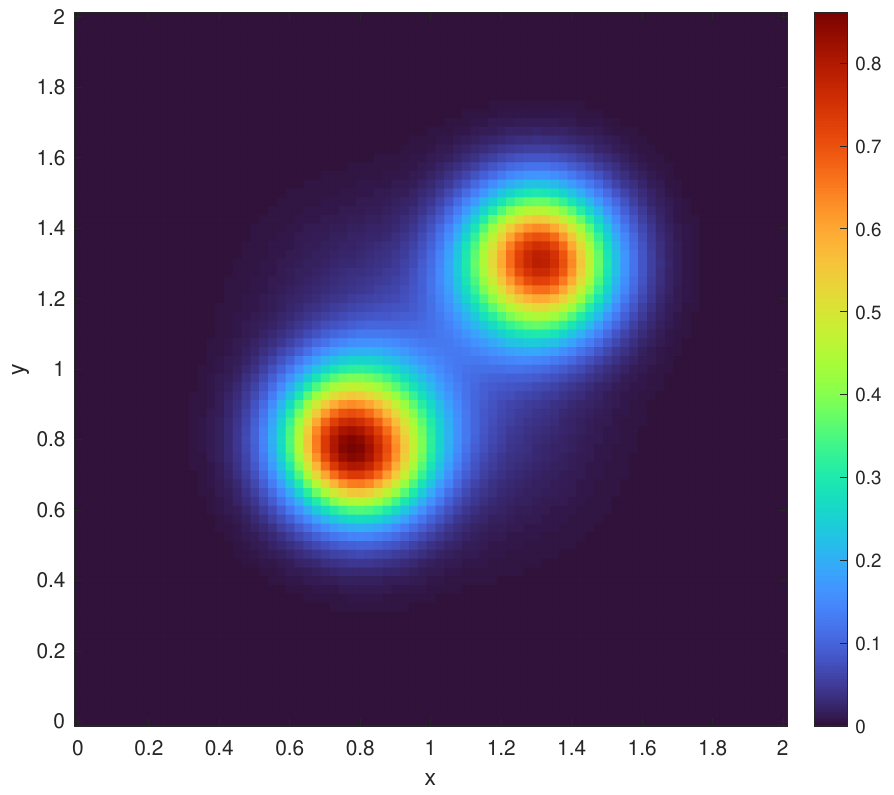}}
\subfigure[]{\includegraphics[scale=0.3]{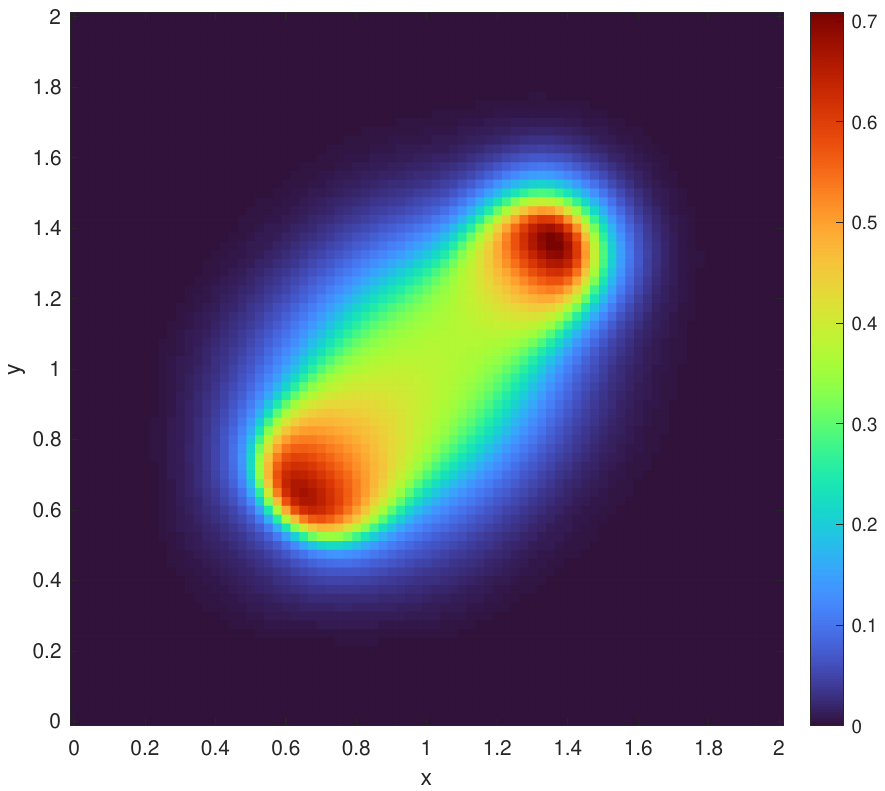}}
\subfigure[]{\includegraphics[scale=0.3]{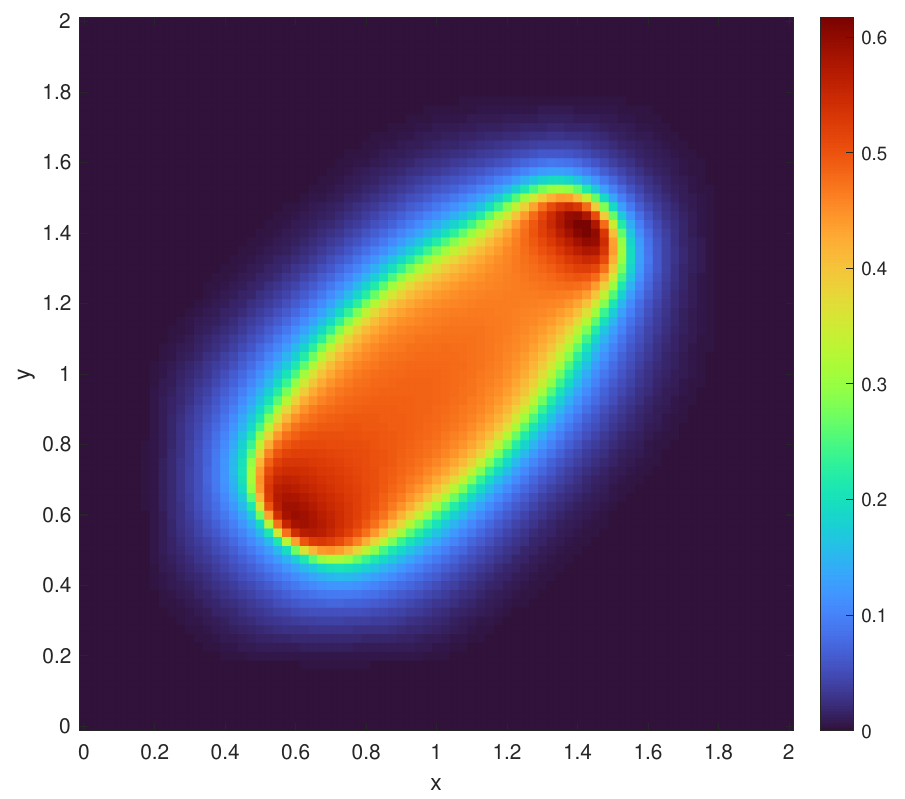}}
\caption{Test 2Da: evolution of $\rho$ (first line) and $l$ (second line) at time a),d) $t=10$, b),e) $t=90$, c),f) $t=210$ . White arrows describe the velocity field. }
\label{fig:test2}
\end{figure}

\subsubsection{Test 2Db: Merge}
In this test we simulate a scenario where two all-follower groups are initially distinct, then they merge interacting with each other. 
To this end, we consider a square domain $\Omega=\left[0, 1\right] \times \left[0, 1\right]$, a final time $T=350$, and we set initial data as follows:
\begin{equation*}
\begin{aligned}  
\rho(x,0)&= \exp\left({-\frac{(x^1-C^1_{3})^2}{2\sigma_3^2}-\frac{(x^2-C^2_{3})^2}{2\sigma_3^2}}\right)+ \exp\left(-{\frac{(x^1-C^1_{4})^2}{2\sigma_4^2}-\frac{(x^2-C^2_{4})^2}{2\sigma_4^2}}\right), \\
u(x,0)&=0,\\
%u^1(t^0,x)&=0\\
%u^2(t^0,x)&=0\\
l (x,0)&=0,
\end{aligned}
\end{equation*}
for any $x=(x^1,x^2) \in \Omega$. Here 
$C_{3}=(C^1_{3},C^2_{3})=(0.4,0.7)$,
$C_{4}=(C^1_{4},C^2_{4})=(0.6,0.3)$,
$\sigma_3=\sigma_4=\sqrt{0.004}$.
Initial density configuration is shown in Figure \ref{fig:test3_test4_ci}.
\begin{figure}[!t]
\centering
\includegraphics[scale=0.3]{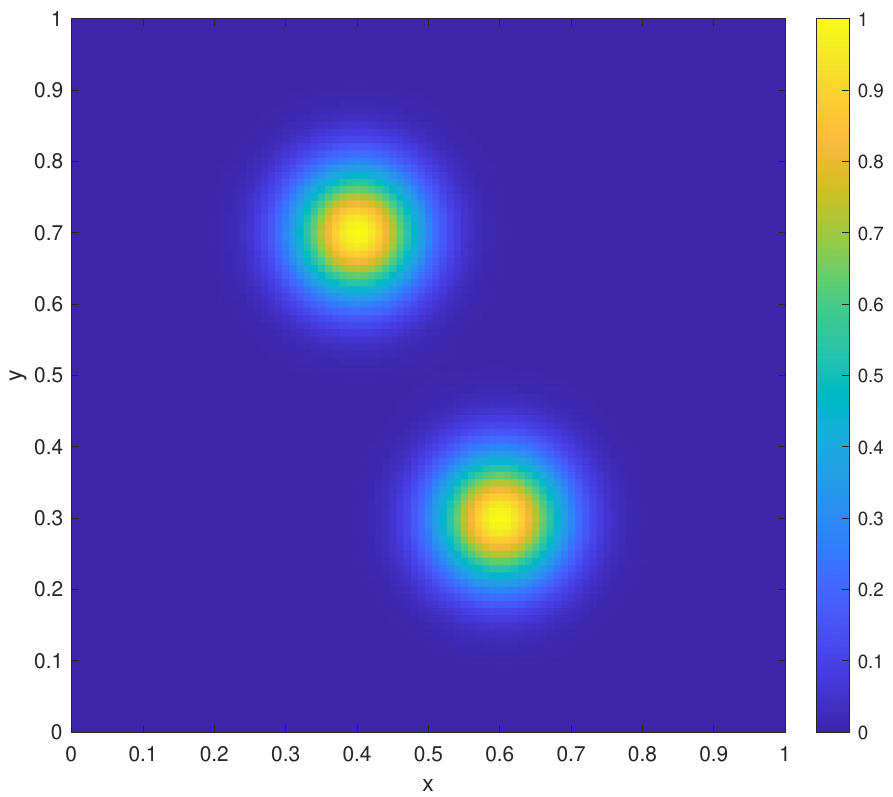}
\caption{Test 2Db: initial condition for $\rho$.}
\label{fig:test3_test4_ci}
\end{figure}

Boundary conditions are of Dirichlet type, with $\rho=u=l=0$. The discretisation steps are $\Delta x^1=\Delta x^2=0.01$ and $\Delta t= 0.2$.

Figure \ref{fig:test3} shows three snapshots of the numerical simulation. 
At the beginning, the attraction force prevails and let the two groups merge.
Meanwhile, the degree of leadership, initially equal to 0, increases and leaders appear. 
Once the two flocks merged, newly formed leaders establish a common direction of motion, heading the flock in the south-west direction.  
\begin{figure}[!t]
\centering
\subfigure[]{\includegraphics[scale=0.3]{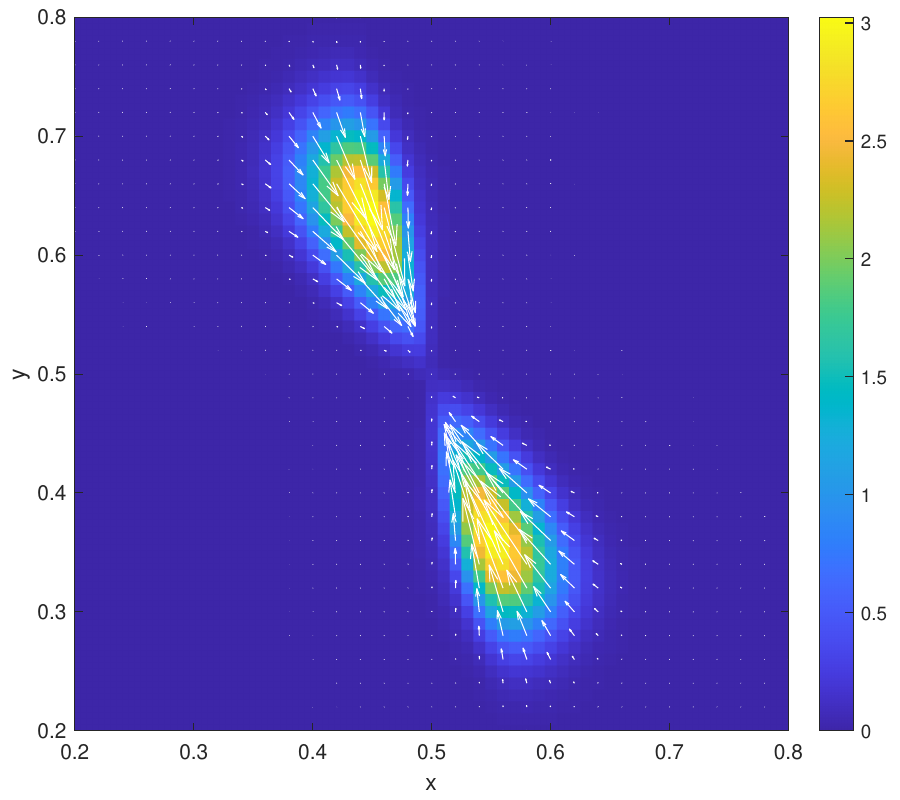}}
\subfigure[]{\includegraphics[scale=0.3]{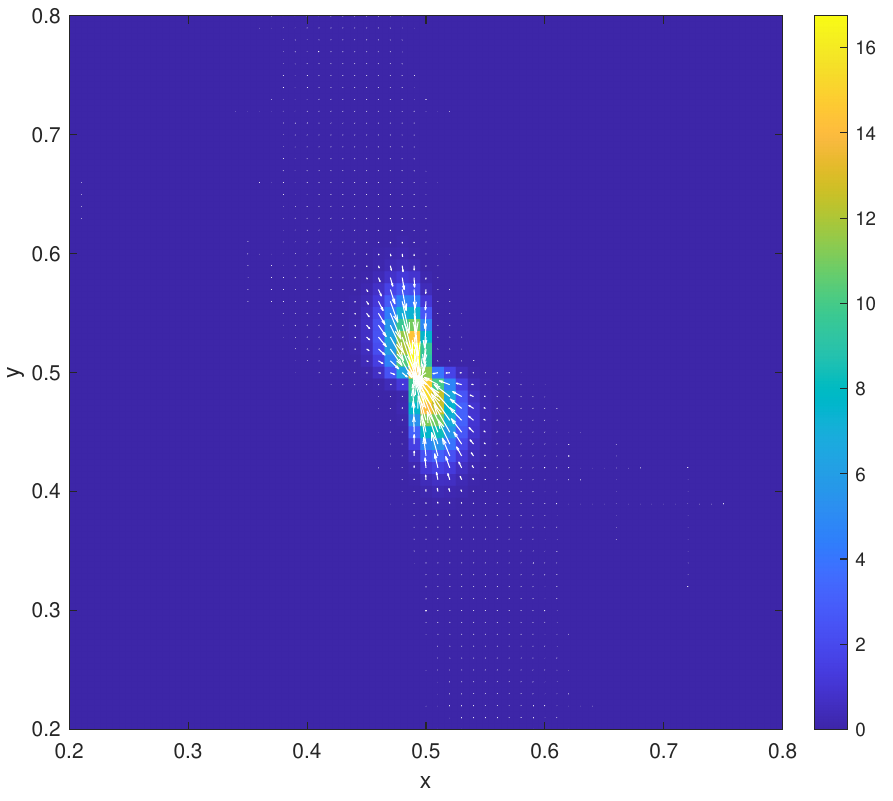}}
\subfigure[]{\includegraphics[scale=0.3]{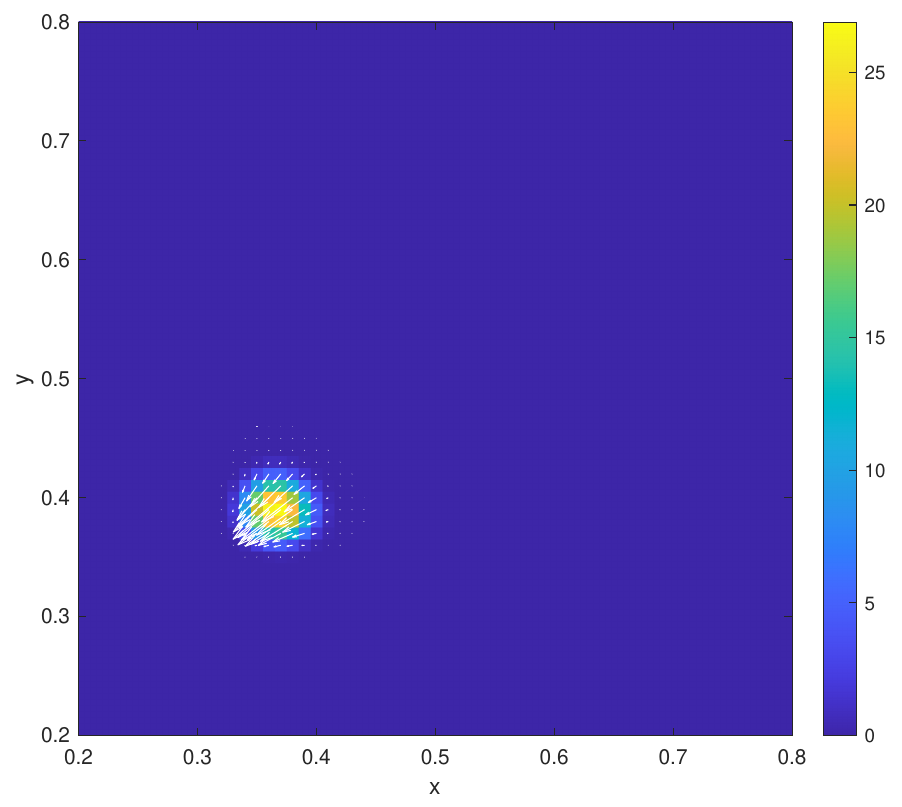}}\\
\subfigure[]{\includegraphics[scale=0.3]{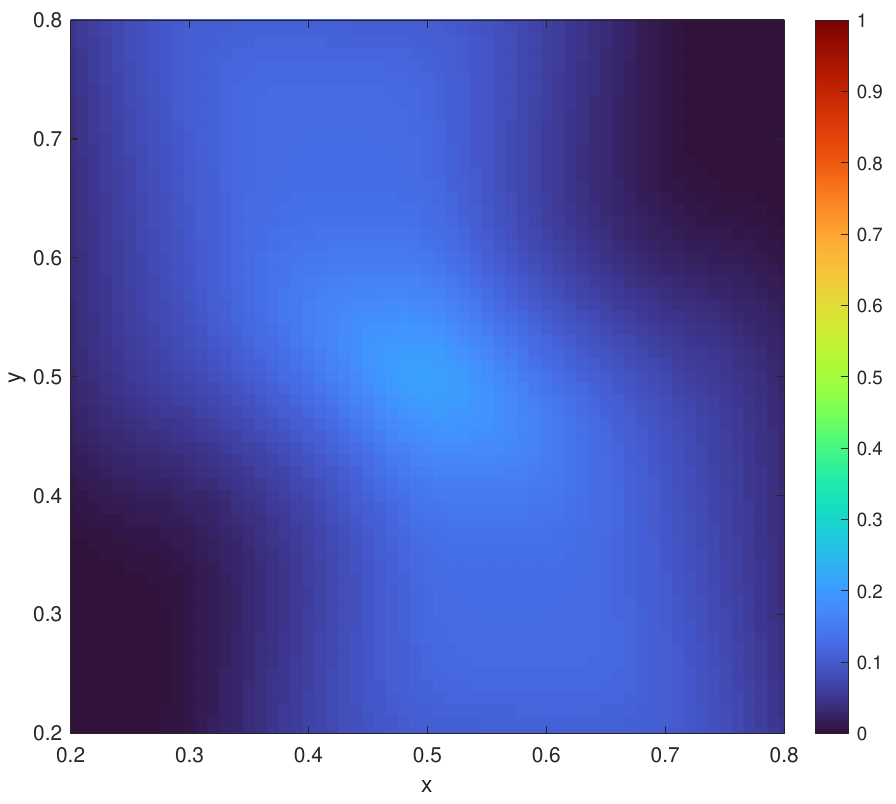}}
\subfigure[]{\includegraphics[scale=0.3]{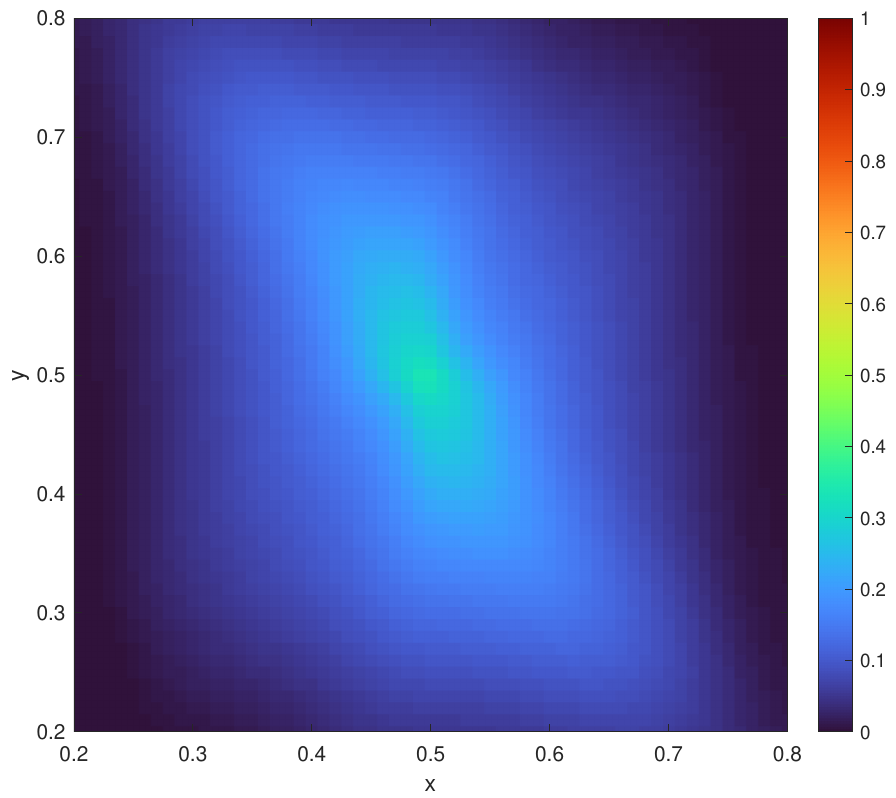}}
\subfigure[]{\includegraphics[scale=0.3]{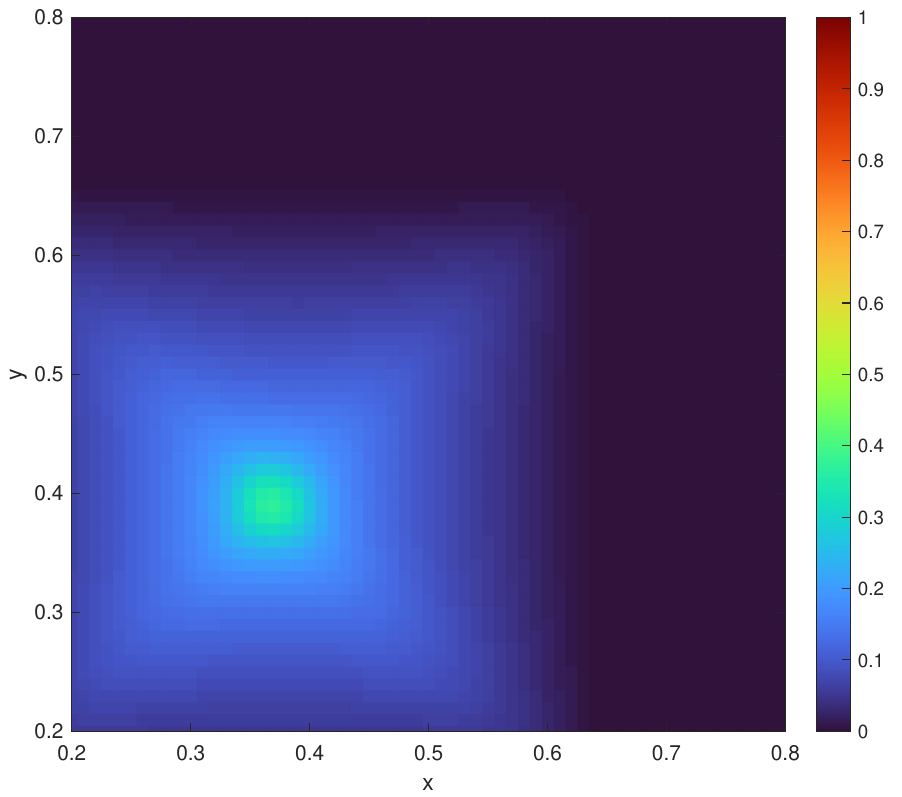}}
\caption{Test 2Db: evolution of $\rho$ (first line) and $l$ (second line) at time a),d) $t=20$, b),e) $t=40$, c),f) $t=180$ . White arrows describe the velocity field. }
\label{fig:test3}
\end{figure}

Figure \ref{fig:test3_cb0} shows the same evolution with null Dirichlet boundary conditions for $\rho$ and $u$, but assuming $l \equiv 1$ at the border. It is interesting to compare the solutions of Figures \ref{fig:test3} and \ref{fig:test3_cb0} in order to quantify the impact of boundary conditions.
It can be seen that the flock merge and find a common solution in both case, but the size of the final flock is quite different, proving that a nonnegligible effect of boundary conditions exists.
\begin{figure}[!t]
\centering
\subfigure[]{\includegraphics[scale=0.3]{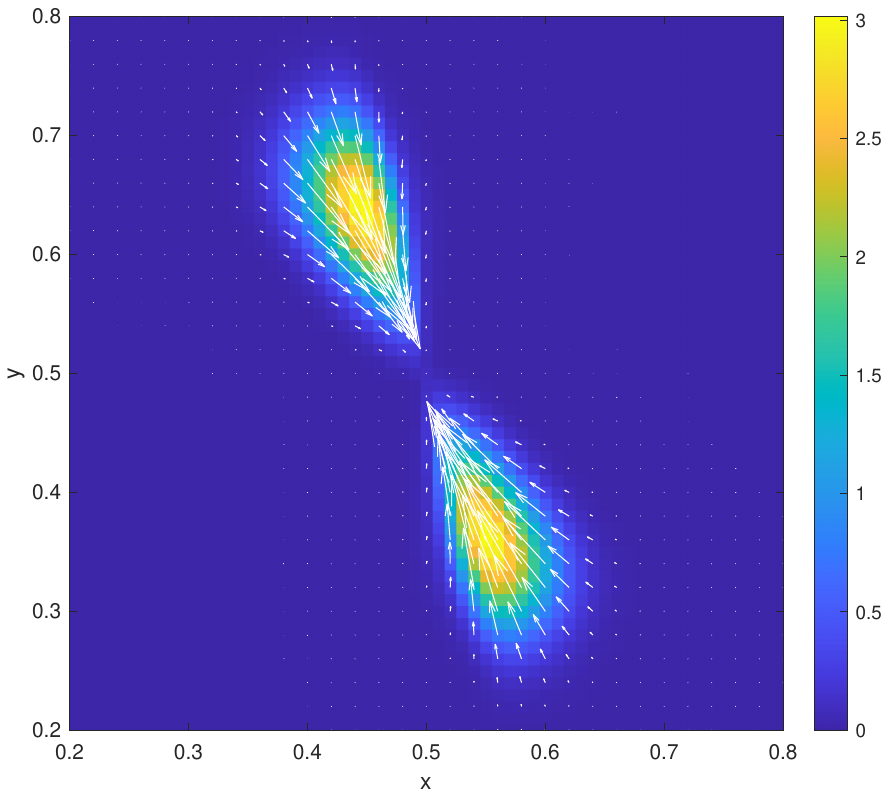}}
\subfigure[]{\includegraphics[scale=0.3]{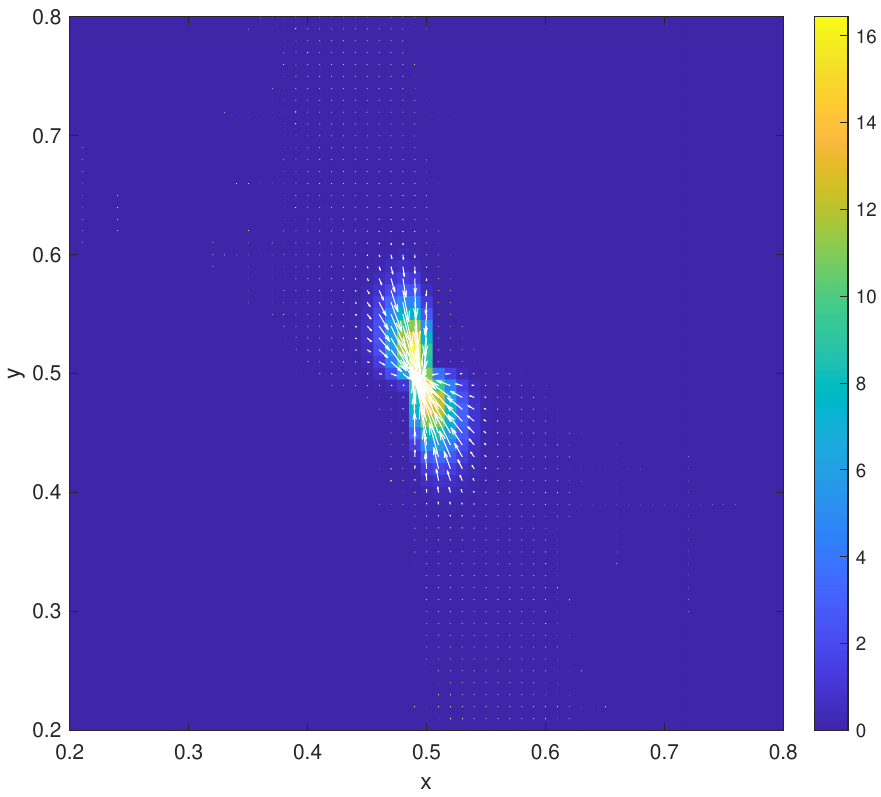}}
\subfigure[]{\includegraphics[scale=0.3]{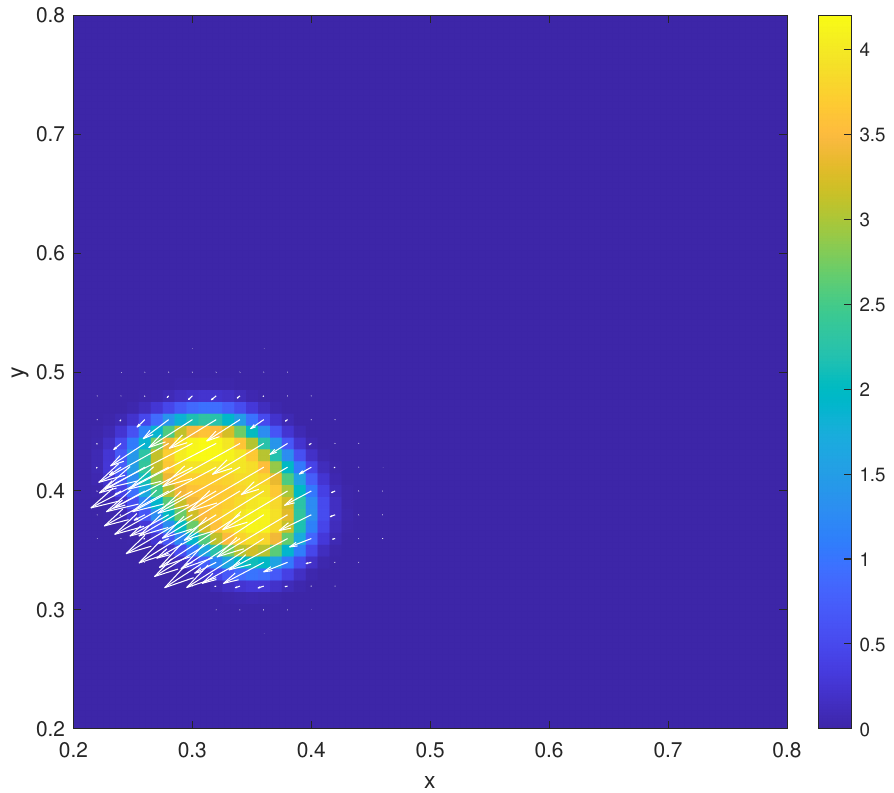}}\\
\subfigure[]{\includegraphics[scale=0.3]{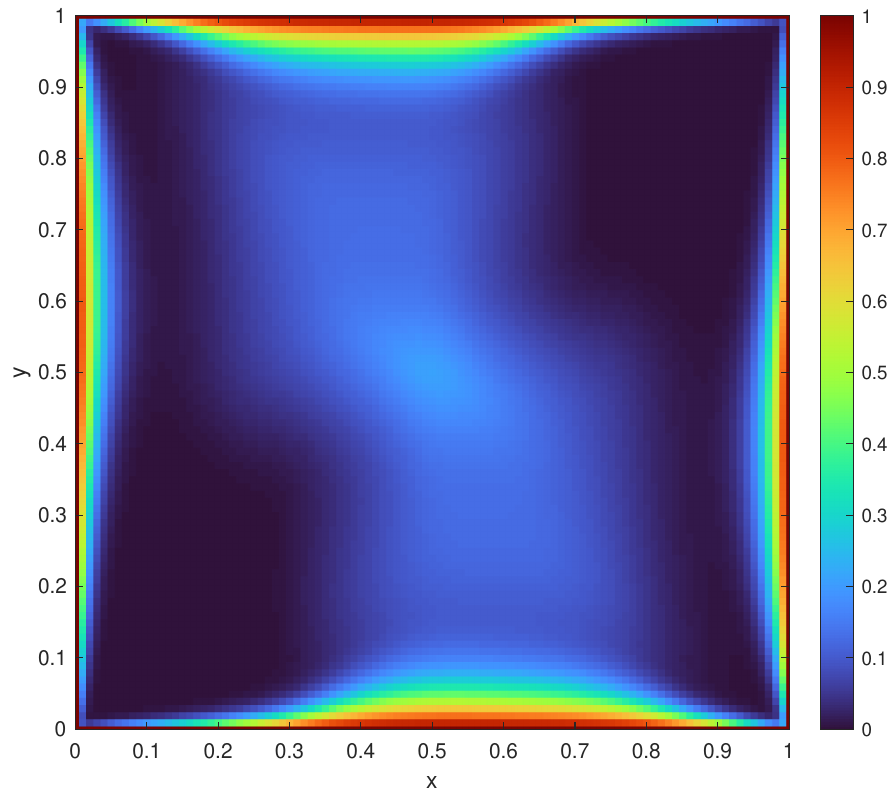}}
\subfigure[]{\includegraphics[scale=0.3]{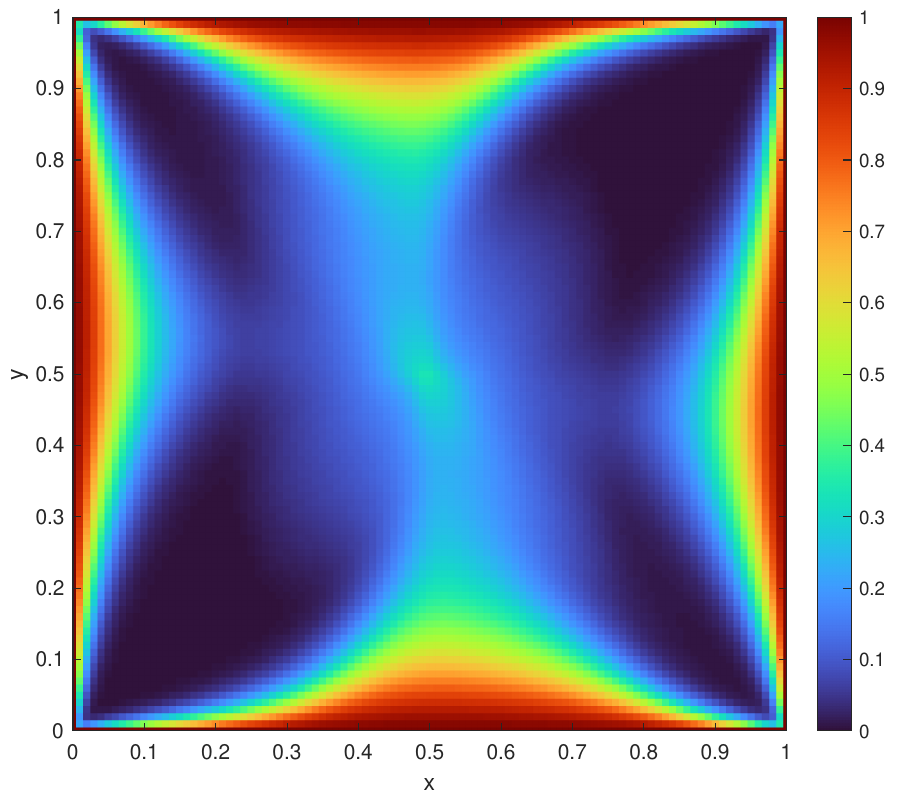}}
\subfigure[]{\includegraphics[scale=0.3]{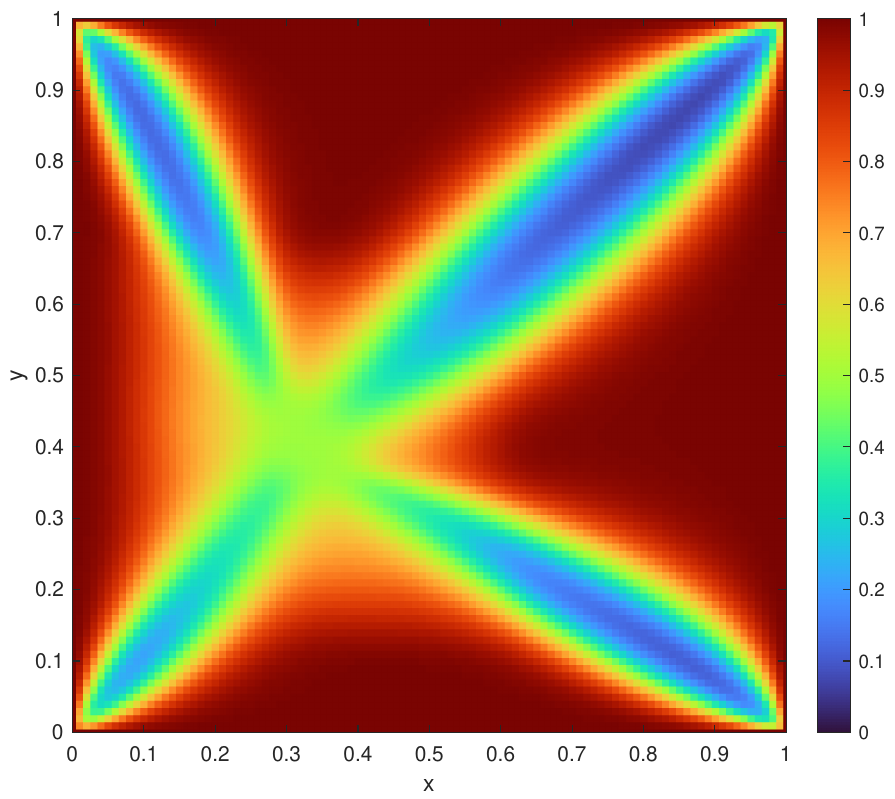}}
\caption{Test 2Db$^\prime$ (different boundary conditions for $l$): evolution of $\rho$ (first line) and $l$ (second line) at time a),d) $t=20$, b),e) $t=40$, c),f) $t=180$ . White arrows describe the velocity field. }
\label{fig:test3_cb0}
\end{figure}

\subsubsection{Test 2Dc: Ring formation}
In this test we consider the same initial and boundary conditions of Test 2Db (Figure \ref{fig:test3_test4_ci}), but we increase all the model parameters but $\alpha$ which remains equal, and $\gamma$ which is instead decreased. Doing this, we expect an accelerated dynamics of $l$ and a less impact of the attraction force.
Result is shown in Figure \ref{fig:test4}. 
We observe that the two flocks still merge but remain less compressed than before. Moreover $l$ evolves faster toward to a mixed follower-leader status, which also spreads on all the domain due to the greater value of $R$.

More interesting, a ring formation appears, characterized by a high density of agents concentrated at the boundary of the flock. This configuration is well known and typically appears in all-to-all repulsion-attraction models \cite{albi2014SIAP, kolokolnikov2011PRE}.

\begin{figure}[!t]
\centering
\subfigure[]{\includegraphics[scale=0.3]{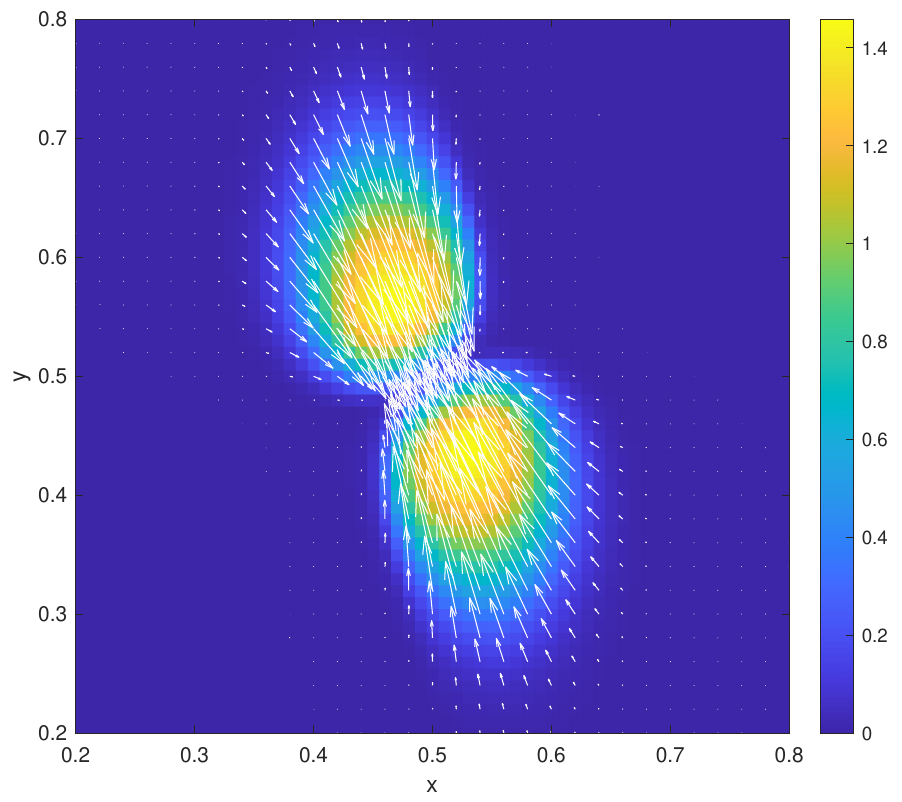}}
\subfigure[]{\includegraphics[scale=0.3]{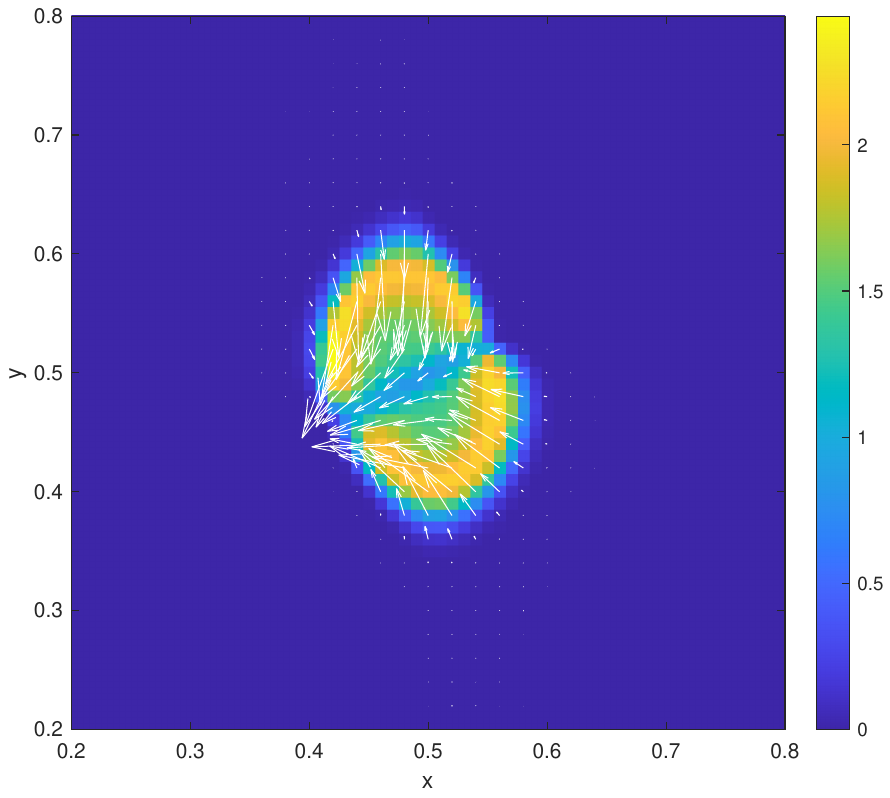}}
\subfigure[]{\includegraphics[scale=0.3]{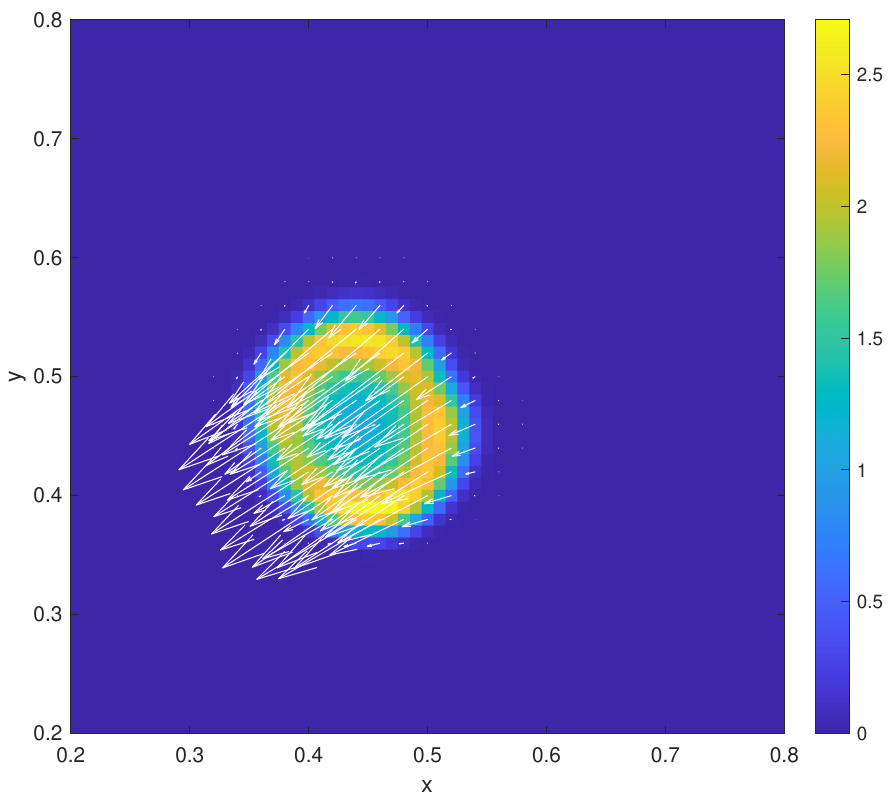}}\\
\subfigure[]{\includegraphics[scale=0.3]{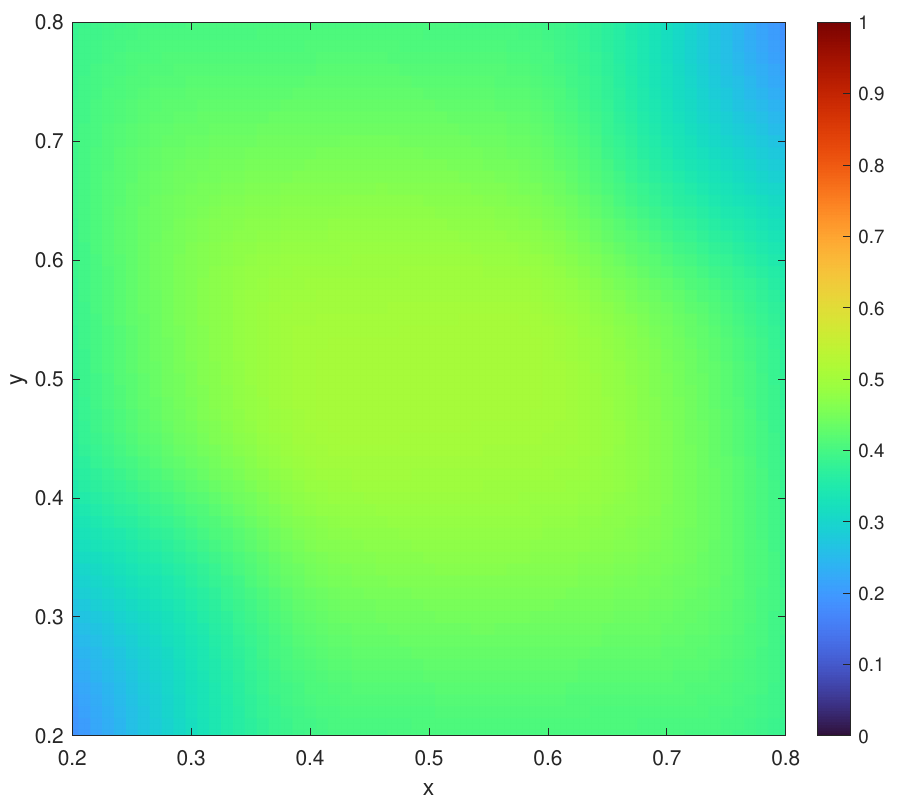}}
\subfigure[]{\includegraphics[scale=0.3]{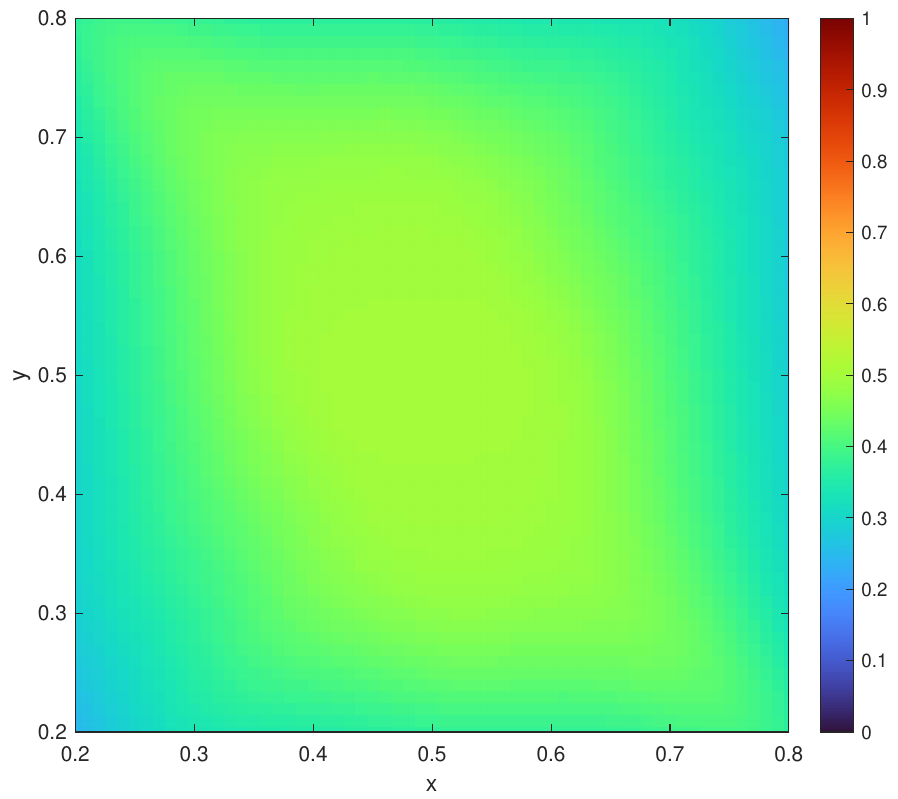}}
\subfigure[]{\includegraphics[scale=0.3]{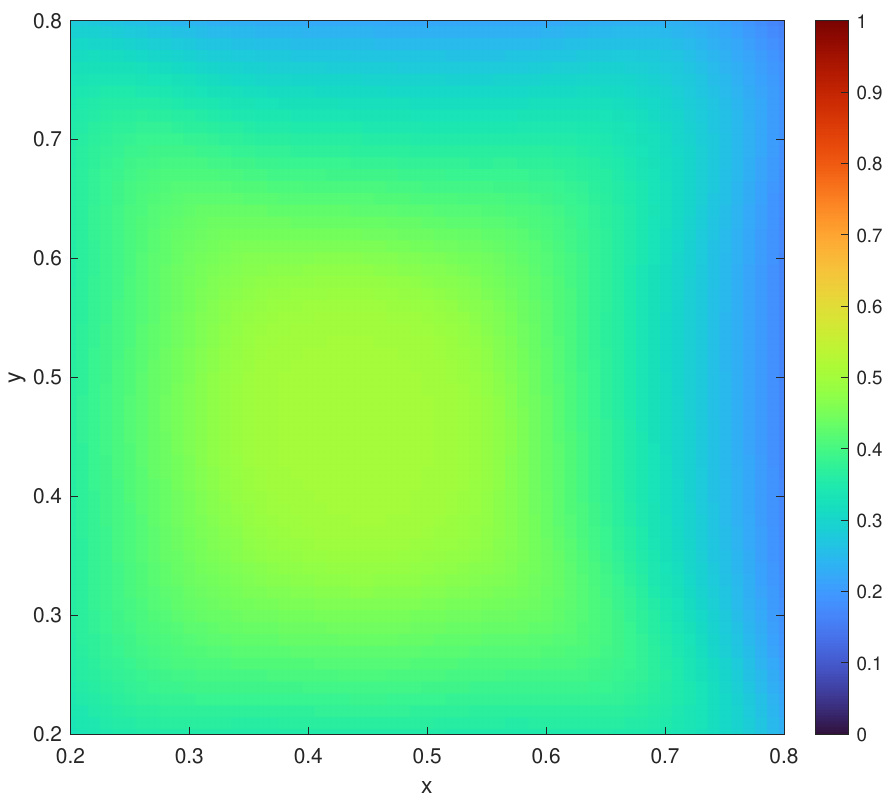}}
\subfigure[]{\includegraphics[scale=0.28]{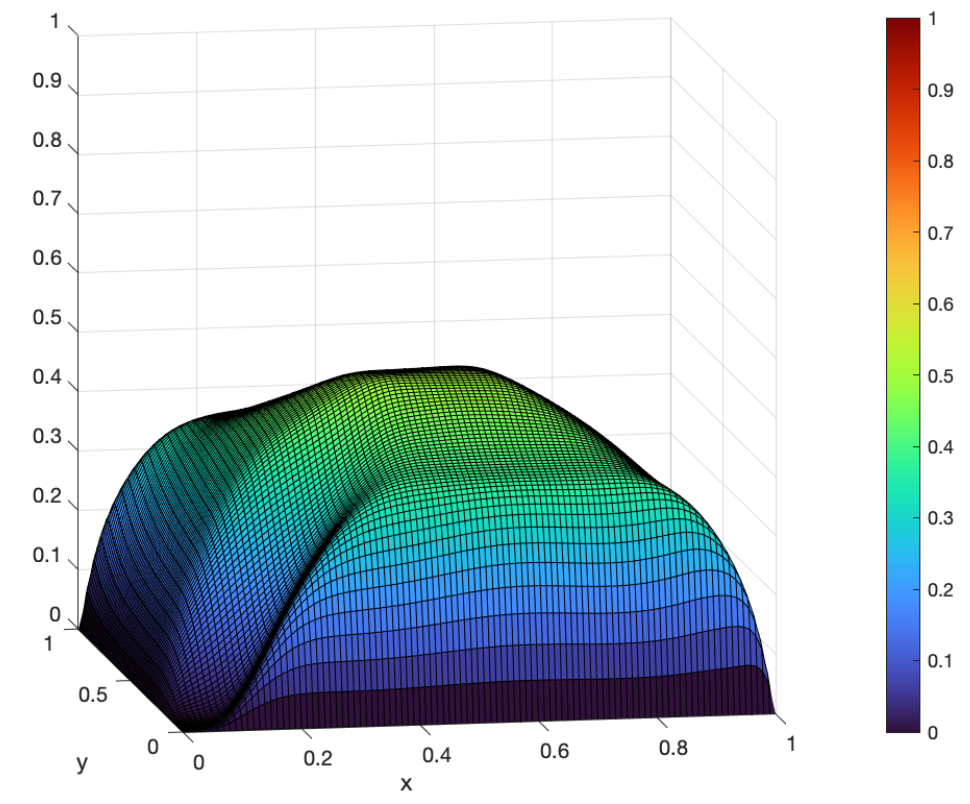}}
\subfigure[]{\includegraphics[scale=0.28]{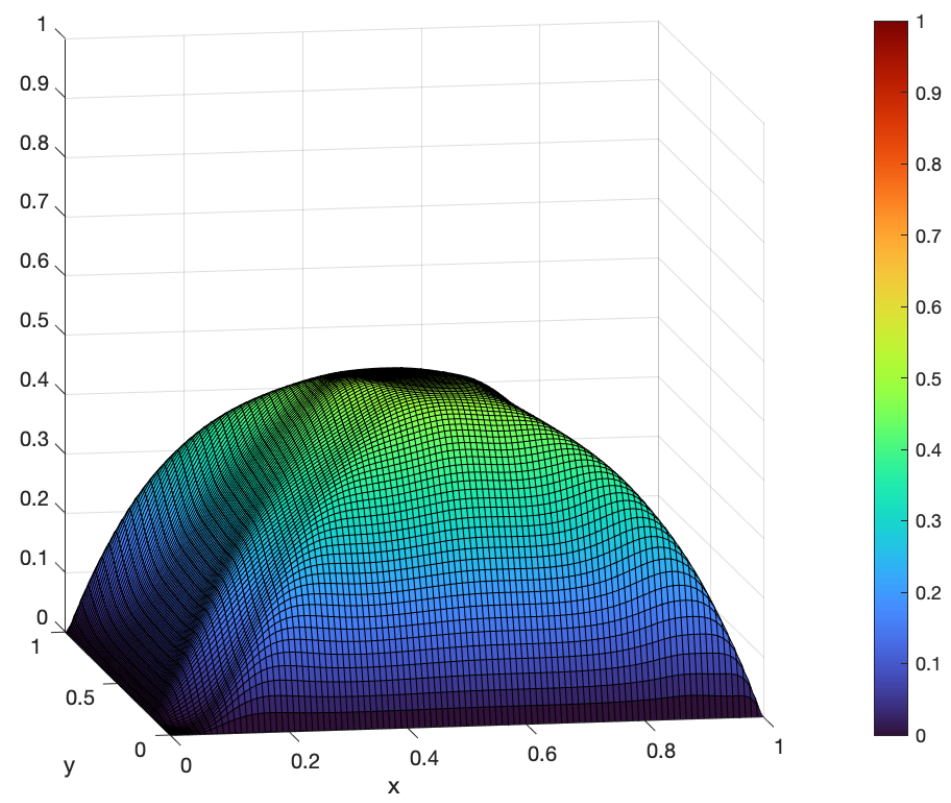}}
\subfigure[]{\includegraphics[scale=0.28]{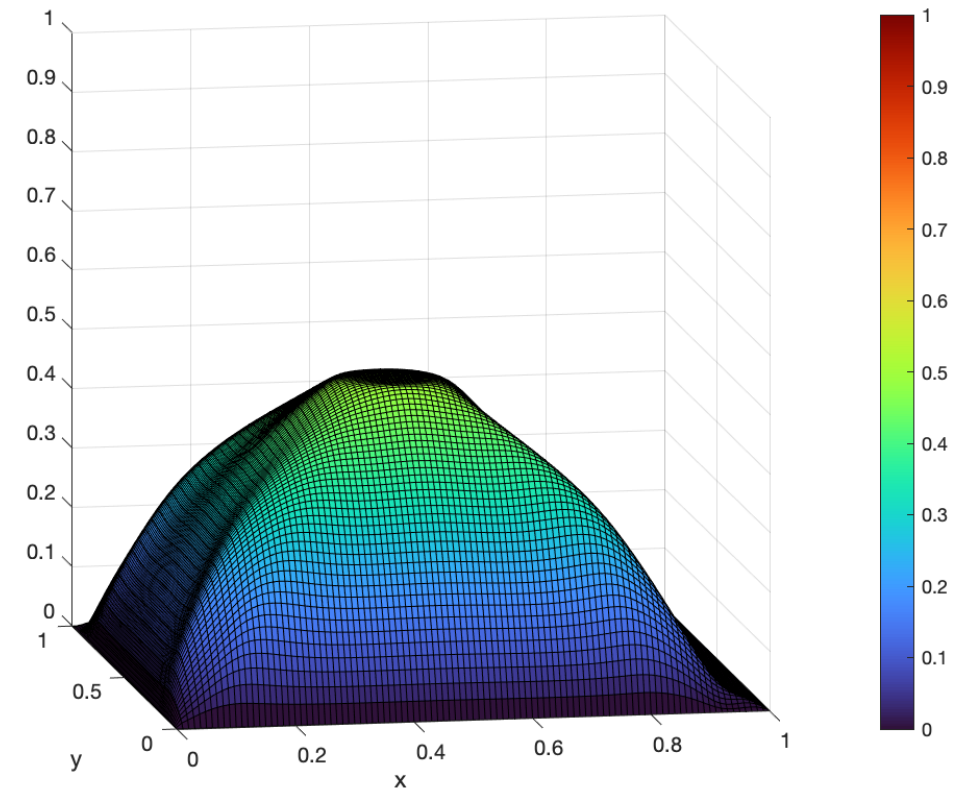}}
\caption{Test 2Dc: evolution of $\rho$ (first line) and $l$ (second and third lines) at time a),d),g) $t=20$, b),e),h) $t=48$, c),f),i) $t=100$ . White arrows describe the velocity field. 
\label{fig:test4}}
\end{figure}

\section{Conclusions}
\label{sect:conclusions}
In this paper, we have proposed a Povzner-Boltzmann-type kinetic description of swarming dynamics with transient leadership and continuous leader-follower transitions relying on a simplified version of the microscopic dynamics described in~\cite{cristiani2021JMB}. The particles of the underlying multi-agent system feature a multi-variate microscopic state consisting of position, velocity and degree of leadership. The velocity and the degree of leadership change in consequence of binary interactions according to two independent processes, which are non-local in space. The resulting dynamics are largely non-conservative, the only conserved quantity being the mass. Non-locality and non-conservativeness pose two major theoretical issues in deriving a macroscopic description of the system through a hydrodynamic limit of the kinetic equation. To circumvent them, taking inspiration from the kinetic theory of nearly elastic granular gases~\cite{fornasier2010PhD,toscani2004CHAPTER}, we consider a regime of high interaction rates and small non-conservative contributions to the interactions. The combination of these two factors produces an aggregate description which still keeps track of the non-locality of the interactions, differently from a classical hyperbolic scaling of space and time in the kinetic equation, which would localise inevitably the mean effect of the interactions. As a consequence, 
we obtain a non-local macroscopic model in the form of a system of evolution equations for relevant statistical quantities formulated on the original spatial and temporal scales of the particle interactions. This means that, in principle, our analytical description is focused on the agents of the swarm also when it comes to aggregate dynamics, differently from the classical case of gas molecules, which instead disappear conceptually when the gas is described at the fluid scale in the hydrodynamic regime. This technique can be fruitfully adapted and extended to upscale consistently also different and possibly richer particle dynamics.
As future perspective, we think that the leader/follower transition would be worth of further investigation in order to find the dynamics which better match the biological mechanisms, currently largely unexplored in both animal groups and human crowds.

\appendix

\section{Strong form of the kinetic operators}
\label{app:strong}
The strong form of the collisional operator $Q^v(f,f)$ may be obtained by means of a change of variables in~\eqref{def:Qv_weak}, taking advantage of the inverse binary collision:
\begin{align*}
    \begin{aligned}[b]
        \pr{v} &= -\frac{1}{1-(1-\lambda)\beta}\left[v+\alpha\frac{x-x_\ast}{\abs{x-x_\ast}^2}+(1-\lambda)\bigl(\beta\pr{v}_\ast+\gamma(x_\ast-x)\bigr)\right] \\
        \pr{v}_\ast &= -\frac{1}{1-(1-\lambda_\ast)\beta}\left[v_\ast+\alpha\frac{x_\ast-x}{\abs{x_\ast-x}^2}+(1-\lambda_\ast)\bigl(\beta\pr{v}+\gamma(x-x_\ast)\bigr)\right],
    \end{aligned}
\end{align*}
where $\pr{v}$, $\pr{v}_\ast$ denote the pre-interaction velocities giving rise to the post-interaction velocities $v$, $v_\ast$. These are nothing but the inverse rules of~\eqref{eq:micro_v}-\eqref{eq:micro_v_symm}. Then, after some algebraic manipulations, one gets
\begin{equation}
    Q^v(f,f)(x,v,\lambda,t)=\int_0^1\int_{\R^{2d}}B(\abs{x_\ast-x})\left(\frac{f(x,\pr{v},\lambda,t)f(x_\ast,\pr{v}_\ast,\lambda_\ast,t)}{\abs{1-\bigl(2-(\lambda+\lambda_\ast)\bigr)\beta}}
        -ff_\ast\right)\,dx_\ast\,dv_\ast\,d\lambda_\ast.
    \label{def:Qv}
\end{equation}

Analogously, the strong form of the collisional operator $Q^\lambda(f,f)$ may be obtained by means of a change of variables in~\eqref{def:Ql_weak}, taking advantage of the inverse rule
$$ \pr{\lambda}=1-\lambda, \qquad \pr{\lambda}_\ast=1-\lambda_\ast, $$
where $\pr{\lambda}$, $\pr{\lambda_\ast}$ are the pre-interaction degrees of leadership giving rise to the post-interaction ones $\lambda$, $\lambda_\ast$. Also in this case, these are the inverse rules of~\eqref{eq:micro_l}-\eqref{eq:micro_l_symm}. The result is
\begin{equation}
    Q^\lambda(f,f)(x,v,\lambda,t)=\int_0^1\int_{\R^{2d}}B(\abs{x_\ast-x})\bigl(f(x,v,\pr{\lambda},t)f(x_\ast,v_\ast,\pr{\lambda}_\ast,t)
        -ff_\ast\bigr)\,dx_\ast\,dv_\ast\,d\lambda_\ast.
    \label{def:Ql}
\end{equation}

To obtain the strong form of the operators $\bar{Q}^v(f,f)$ and $\cR^v(f,f)$, cf.~\eqref{eq:barQv_weak},~\eqref{eq:Rv_weak}, we proceed similarly. The inverse collision to~\eqref{eq:micro_v_cons} reads
$$ \pr{\bar{v}}=\frac{1-p}{1-2p}v_\ast-\frac{p}{1-2p}v, \qquad
    \pr{\bar{v}_\ast}=\frac{1-p}{1-2p}v-\frac{p}{1-2p}v_\ast $$
with $p:=\left(1-\frac{\lambda+\lambda_\ast}{2}\right)\beta_0\delta$, where $\pr{\bar{v}}$, $\pr{\bar{v}}_\ast$ denote the pre-interaction velocities giving rise to the post-interaction velocities $v$, $v_\ast$. Therefore, changing the variables in~\eqref{eq:barQv_weak} produces
\begin{equation}
    \bar{Q}^v(f,f)(x,v,\lambda,t)=\int_0^1\int_{\R^{2d}}B(\abs{x_\ast-x})
        \left(\frac{f(x,\pr{\bar{v}},\lambda,t)f(x_\ast,\pr{\bar{v}_\ast},\lambda_\ast,t)}{\abs{1-\bigl(2-(\lambda+\lambda_\ast)\bigr)\beta_0\delta}}
            -ff_\ast\right)dx_\ast\,dv_\ast\,d\lambda_\ast.
    \label{eq:barQv}
\end{equation}
The same change of variables in~\eqref{eq:Rv_weak} yields instead
\begin{align*}
    & \ave{\cR^v(f,f),\,\varphi} \\
    &= \int_{[0,\,1]^2}\int_{\R^{3d}}\Biggl[\frac{B(\abs{x_\ast-x})}{\abs{1-\bigl(2-(\lambda+\lambda_\ast)\bigr)\beta_0\delta}}
        \left(\beta_0\frac{\lambda_\ast-\lambda}{2\bigl(1-(2-(\lambda+\lambda_\ast))\beta_0\delta\bigr)}(v-v_\ast)\right. \\
    & \qquad \left.+\alpha_0\frac{x-x_\ast}{\abs{x-x_\ast}^2}+\gamma_0(1-\lambda)(x_\ast-x)\right)\cdot\nabla_v\varphi(v,\lambda)
        f(x,\pr{\bar{v}},\lambda,t)f(x_\ast,\pr{\bar{v}_\ast},\lambda_\ast,t)\Biggr]\,dx_\ast\,dv\,dv_\ast\,d\lambda\,d\lambda_\ast \\
    &=-\int_0^1\int_{\R^d}\varphi(v,\lambda)\ddiv_v\left[\int_0^1\int_{\R^{2d}}\frac{B(\abs{x_\ast-x})}{\abs{1-\bigl(2-(\lambda+\lambda_\ast)\bigr)\beta_0\delta}}
        \left(\beta_0\frac{\lambda_\ast-\lambda}{2\bigl(1-(2-(\lambda+\lambda_\ast))\beta_0\delta\bigr)}(v-v_\ast)\right.\right. \\
    & \qquad \left.\left.+\alpha_0\frac{x-x_\ast}{\abs{x-x_\ast}^2}+\gamma_0(1-\lambda)(x_\ast-x)\right)
        f(x,\pr{\bar{v}},\lambda,t)f(x_\ast,\pr{\bar{v}_\ast},\lambda_\ast,t)\,dx_\ast\,dv_\ast\,d\lambda_\ast\right]dv\,d\lambda,
\end{align*}
where the last passage follows from the divergence theorem, assuming further for convenience that $\varphi$ is compactly supported in $\R^d\times [0,\,1]$. Therefore:
\begin{multline*}
    \cR^v(f,f)(x,v,\lambda,t)= \\
    -\ddiv_v\Biggl[\int_0^1\int_{\R^{2d}}\frac{B(\abs{x_\ast-x})}{\abs{1-\bigl(2-(\lambda+\lambda_\ast)\bigr)\beta_0\delta}}\Biggl(\beta_0\frac{\lambda_\ast-\lambda}{2\bigl(1-(2-(\lambda+\lambda_\ast))\beta_0\delta\bigr)}(v-v_\ast) \\
    +\alpha_0\frac{x-x_\ast}{\abs{x-x_\ast}^2}+\gamma_0(1-\lambda)(x_\ast-x)\Biggr)f(x,\pr{\bar{v}},\lambda,t)f(x_\ast,\pr{\bar{v}_\ast},\lambda_\ast,t)\,dx_\ast\,dv_\ast\,d\lambda_\ast\Biggr].
\end{multline*}

The strong forms of the operators $\bar{Q}^\lambda(f,f)$ and $\cR^\lambda(f,f)$, cf.~\eqref{eq:barQl_weak},~\eqref{eq:Rl_weak}, can be obtained analogously starting from the inverse rules to~\eqref{eq:micro_lambda_cons}. In particular, these are
$$ \pr{\bar{\lambda}}=\frac{1-\nu}{1-2\nu}\lambda-\frac{\nu}{1-2\nu}\lambda_\ast, \qquad
    \pr{\bar{\lambda}_\ast}=\frac{1-\nu}{1-2\nu}\lambda_\ast-\frac{\nu}{1-2\nu}\lambda, $$
so that
$$ \bar{Q}^\lambda(f,f)(x,v,\lambda,t)=\int_0^1\int_{\R^{2d}}B(\abs{x_\ast-x})\left(\frac{f(x,v,\pr{\bar{\lambda}},t)f(x_\ast,v_\ast,\pr{\bar{\lambda}_\ast},t)}{\abs{1-2\nu}}
    -ff_\ast\right)dx_\ast\,dv_\ast\,d\lambda_\ast $$
and
\begin{multline}
    \cR^\lambda(f,f)(x,v,\lambda,t)= \\
    -\partial_\lambda\Biggl[\int_0^1\int_{\R^{2d}}\frac{B(\abs{x_\ast-x})}{\abs{1-2\nu}}\left(1-\frac{\nu}{1-2\nu}\lambda+\frac{3\nu-2}{1-2\nu}\lambda_\ast\right) \\  
    \times f(x,v,\pr{\bar{\lambda}},t)f(x_\ast,v_\ast,\pr{\bar{\lambda}_\ast},t)\,dx_\ast\,dv_\ast\,d\lambda_\ast\Biggr].
    \label{eq:Rl}
\end{multline}

\paragraph*{Conflict of Interest}
The authors declare that no possible conflict of interests arises from the content of this paper.

\paragraph*{Data Availability}
The few data related to the numerical computations will be made available on reasonable request.

\paragraph*{Acknowledgement}
E.C.\ would like to thank the Italian Ministry of University and Research (MUR) to support this research with funds coming from PRIN Project 2022 PNRR (No.\ 2022XJ9SX, entitled ``Heterogeneity on the road - Modeling, analysis, control'') and from PRIN Project 2022 (No.\ 2022238YY5, entitled ``Optimal control problems: analysis, approximation and applications'').

This study was carried out within the Spoke 7 of the MOST -- Sustainable Mobility National Research Center and received funding from the European Union Next-Generation EU (PIANO NAZIONALE DI RIPRESA E RESILIENZA (PNRR) – MISSIONE 4 COMPONENTE 2, INVESTIMENTO 1.4 – D.D. 1033 17/06/2022, CN00000023). This manuscript reflects only the authors' views and opinions. Neither the European Union nor the European Commission can be considered responsible for them.

E.C.\ and M.M. are funded by INdAM - GNCS Project, CUP  E53C23001670001, entitled ``Numerical modeling and high-performance computing approaches for multiscale models of Complex Systems''.

E.C.\ and M.M.\ are members of GNCS-INdAM research group. 

N.L.\ and A.T.\ are members of GNFM-INdAM research group.

\bibliographystyle{plain}
\bibliography{CeLnMmTa-swarming}

\begin{thebibliography}{10}

\bibitem{Peurichard2021}
P.~Aceves-Sanchez, P.~Degond, E.~E. Keaveny, A.~Manhart, S.~Merino-Aceituno,
  and D.~Peurichard.
\newblock Large-scale dynamics of self-propelled particles moving through
  obstacles: Model derivation and pattern formation.
\newblock {\em Bull. Math. Biol.}, 82(10), 2020.

\bibitem{albi2022AMO}
G.~Albi, S.~Almi, M.~Morandotti, and F.~Solombrino.
\newblock Mean-field selective optimal control via transient leadership.
\newblock {\em Appl. Math. Optim.}, 85(2):22, 2022.

\bibitem{albi2014SIAP}
G.~Albi, D.~Balagu\'e, J.~A. Carrillo, and J.~von Brecht.
\newblock Stability analysis of flock and mill rings for second order models in
  swarming.
\newblock {\em SIAM J. Appl. Math.}, 74(3):794--818, 2014.

\bibitem{albi2016SIAP}
G.~Albi, M.~Bongini, E.~Cristiani, and D.~Kalise.
\newblock Invisible control of self-organizing agents leaving unknown
  environments.
\newblock {\em SIAM J. Appl. Math.}, 76(4):1683--1710, 2016.

\bibitem{albi2019M3AS}
G.~Albi, M.~Bongini, F.~Rossi, and F.~Solombrino.
\newblock Leader formation with mean-field birth and death models.
\newblock {\em Math. Models Methods Appl. Sci.}, 29(04):633--679, 2019.

\bibitem{albi2023papercugino}
G.~Albi and F.~Ferrarese.
\newblock Kinetic description of swarming dynamics with topological interaction
  and emergent leaders.
\newblock arXiv:2307.12044, 2023.

\bibitem{albi2013AML}
G.~Albi and L.~Pareschi.
\newblock Modeling of self-organized systems interacting with a few
  individuals: from microscopic to macroscopic dynamics.
\newblock {\em Appl. Math. Lett.}, 26(4):397--401, 2013.

\bibitem{attanasi2014}
A.~Attanasi, A.~Cavagna, L.~Del~Castello, I.~Giardina, T.~S. Grigera,
  A.~Jeli\'c, S.~Melillo, L.~Parisi, O.~Pohl, E.~Shen, and M.~Viale.
\newblock Information transfer and behavioural inertia in starling flocks.
\newblock {\em Nat. Phys.}, 10:691--696, 2014.

\bibitem{attanasi2015}
A.~Attanasi, A.~Cavagna, L.~Del~Castello, I.~Giardina, A.~Jelic, S.~Melillo,
  L.~Parisi, O.~Pohl, E.~Shen, and M.~Viale.
\newblock Emergence of collective changes in travel direction of starling
  flocks from individual birds' fluctuations.
\newblock {\em J. R. Soc. Interface}, 12:20150319, 2015.

\bibitem{ballerini2008b}
M.~Ballerini, N.~Cabibbo, R.~Candelier, A.~Cavagna, E.~Cisbani, I.~Giardina,
  V.~Lecomte, A.~Orlandi, G.~Parisi, A.~Procaccini, M.~Viale, and
  V.~Zdravkovic.
\newblock Empirical investigation of starling flocks: {A} benchmark study in
  collective animal behaviour.
\newblock {\em Anim. Behav.}, 76(1):201--215, 2008.

\bibitem{ballerini2008a}
M.~Ballerini, N.~Cabibbo, R.~Candelier, A.~Cavagna, E.~Cisbani, I.~Giardina,
  V.~Lecomte, A.~Orlandi, G.~Parisi, A.~Procaccini, M.~Viale, and
  V.~Zdravkovic.
\newblock Interaction ruling animal collective behavior depends on topological
  rather than metric distance: {E}vidence from a field study.
\newblock {\em P. Natl. Acad. Sci. USA}, 105(4):1232--1237, 2008.

\bibitem{caglioti2020}
D.~Benedetto, P.~Butt\`a, and E.~Caglioti.
\newblock Some aspects of the inertial spin model for flocks and related
  kinetic equations.
\newblock {\em Math. Models Methods Appl. Sci.}, 30(10):1987--2022, 2020.

\bibitem{bernardi2021}
S.~Bernardi, R.~Eftimie, and K.~J. Painter.
\newblock Leadership through influence: what mechanisms allow leaders to steer
  a swarm?
\newblock {\em Bull. Math. Biol.}, 83:69, 2021.

\bibitem{BolleyM3AS2011}
F.~Bolley, J.~A. Ca\~nizo, and J.~A. Carrillo.
\newblock Stochastic mean-field limit: non-{L}ipschitz forces and swarming.
\newblock {\em Math. Models Methods Appl. Sci.}, 21(11):2179--2210, 2011.

\bibitem{bongini2014NHM}
M.~Bongini and M.~Fornasier.
\newblock Sparse stabilization of dynamical systems driven by attraction and
  avoidance forces.
\newblock {\em Netw. Heterog. Media}, 9(1):1--31, 2014.

\bibitem{borzi2015M3AS}
A.~Borzi and S.~Wongkaew.
\newblock Modeling and control through leadership of a refined flocking system.
\newblock {\em Math. Models Methods Appl. Sci.}, 25(02):255--282, 2015.

\bibitem{MihaM3AS2020}
M.~Bostan and J.~A. Carrillo.
\newblock Fluid models with phase transition for kinetic equations in swarming.
\newblock {\em Math. Models Methods Appl. Sci.}, 30(10):2023--2065, 2020.

\bibitem{butail2019detecting}
S.~Butail and M.~Porfiri.
\newblock Detecting switching leadership in collective motion.
\newblock {\em Chaos}, 29(1):011102, 2019.

\bibitem{CarrilloM3AS2011}
J.~A. Ca\~{n}izo, J.~A. Carrillo, and J.~Rosado.
\newblock A well-posedness theory in measures for some kinetic models of
  collective motion.
\newblock {\em Math. Models Methods Appl. Sci.}, 21(03):515--539, 2011.

\bibitem{Carrillo2014}
J.~A. Carrillo, Y.-P. Choi, and M.~Hauray.
\newblock The derivation of swarming models: mean-field limit and {W}asserstein
  distances.
\newblock In A.~Muntean and F.~Toschi, editors, {\em Collective Dynamics from
  Bacteria to Crowds: An Excursion Through Modeling, Analysis and Simulation},
  pages 1--46. Springer Vienna, Vienna, 2014.

\bibitem{CarrilloJEMS2019}
J.~A. Carrillo, Y.-P. Choi, M.~Hauray, and S.~Salem.
\newblock Mean-field limit for collective behavior models with sharp
  sensitivity regions.
\newblock {\em J. Eur. Math. Soc.}, 21:121–161, 2019.

\bibitem{Carrillo2017}
J.~A. Carrillo, Y.-P. Choi, and S.~P. Perez.
\newblock A review on attractive--repulsive hydrodynamics for consensus in
  collective behavior.
\newblock In N.~Bellomo, P.~Degond, and E.~Tadmor, editors, {\em Active
  Particles, Volume 1: Advances in Theory, Models, and Applications}, pages
  259--298. Springer International Publishing, Cham, 2017.

\bibitem{CarrilloKRM2009}
J.~A. Carrillo, M.~R. D’Orsogna, and V.~Panferov.
\newblock Double milling in self-propelled swarms from kinetic theory.
\newblock {\em Kinet. Relat. Models}, 2(2):363--378, 2009.

\bibitem{carrillo2010SIMA}
J.~A. Carrillo, M.~Fornasier, J.~Rosado, and G.~Toscani.
\newblock Asymptotic flocking dynamics for the kinetic {C}ucker-{S}male model.
\newblock {\em SIAM J. Math. Anal.}, 42(1):218--236, 2010.

\bibitem{Carrillo2010}
J.~A. Carrillo, M.~Fornasier, G.~Toscani, and F.~Vecil.
\newblock Particle, kinetic, and hydrodynamic models of swarming.
\newblock In G.~Naldi, L.~Pareschi, and G.~Toscani, editors, {\em Mathematical
  Modeling of Collective Behavior in Socio-Economic and Life Sciences}, pages
  297--336. Birkh{\"a}user Boston, Boston, 2010.

\bibitem{CarrilloM3AS}
J.~A. Carrillo, A.~Wr\'{o}blewska-Kami\'{n}ska, and E.~Zatorska.
\newblock On long-time asymptotics for viscous hydrodynamic models of
  collective behavior with damping and nonlocal interactions.
\newblock {\em Math. Models Methods Appl. Sci.}, 29(01):31--63, 2019.

\bibitem{cavagna2010}
A.~Cavagna, A.~Cimarelli, I.~Giardina, G.~Parisi, R.~Santagati, F.~Stefanini,
  and R.~Tavarone.
\newblock From empirical data to inter-individual interactions: unveiling the
  rules of collective animal behavior.
\newblock {\em Math. Models Methods Appl. Sci.}, 20:1491--1510, 2010.

\bibitem{cavagna2022N}
A.~Cavagna, A.~Culla, X.~Feng, I.~Giardina, T.~S. Grigera, W.~Kion-Crosby,
  S.~Melillo, G.~Pisegna, L.~Postiglione, and P.~Villegas.
\newblock Marginal speed confinement resolves the conflict between correlation
  and control in collective behaviour.
\newblock {\em Nat. Commun.}, 13(1):2315, 2022.

\bibitem{cavagna2015}
A.~Cavagna, L.~Del~Castello, I.~Giardina, T.~Grigera, A.~Jelic, S.~Melillo,
  T.~Mora, L.~Parisi, E.~Silvestri, M.~Viale, and A.~M. Walczak.
\newblock Flocking and turning: {A} new model for self-organized collective
  motion.
\newblock {\em J. Stat. Phys.}, 158:601--627, 2015.

\bibitem{cavagna2013}
A.~Cavagna, S.~M. Duarte~Queir\'os, I.~Giardina, F.~Stefanini, and M.~Viale.
\newblock Diffusion of individual birds in starling flocks.
\newblock {\em Proc. R. Soc. B.}, 280:20122484, 2013.

\bibitem{chen2016switching}
D.~Chen, T.~Vicsek, X.~Liu, T.~Zhou, and H.-T. Zhang.
\newblock Switching hierarchical leadership mechanism in homing flight of
  pigeon flocks.
\newblock {\em Europhys. Lett. EPL}, 114(6):60008, 2016.

\bibitem{couzin2005}
I.~D. Couzin, J.~Krause, N.~R. Franks, and S.~A. Levin.
\newblock Effective leadership and decision-making in animal groups on the
  move.
\newblock {\em Nature}, 433:513--516, 2005.

\bibitem{couzin2002}
I.~D. Couzin, J.~Krause, R.~James, G.~D. Ruxton, and N.~R. Franks.
\newblock Collective memory and spatial sorting in animal groups.
\newblock {\em J. Theor. Biol.}, 218(1):1--11, 2002.

\bibitem{cristiani2023CMS}
E.~Cristiani, A.~De~Santo, and M.~Menci.
\newblock A generalized mean-field game model for the dynamics of pedestrians
  with limited predictive abilities.
\newblock {\em Commun. Math. Sci.}, 21(1):65--82, 2023.

\bibitem{cristiani2021JMB}
E.~Cristiani, M.~Menci, M.~Papi, and L.~Brafman.
\newblock An all-leader agent-based model for turning and flocking birds.
\newblock {\em J. Math. Biol.}, 83:45, 2021.

\bibitem{cristiani2011MMS}
E.~Cristiani, B.~Piccoli, and A.~Tosin.
\newblock Multiscale modeling of granular flows with application to crowd
  dynamics.
\newblock {\em Multiscale Model. Simul.}, 9(1):155--182, 2011.

\bibitem{cristiani2014book}
E.~Cristiani, B.~Piccoli, and A.~Tosin.
\newblock {\em Multiscale modeling of pedestrian dynamics}, volume~12.
\newblock Springer, 2014.

\bibitem{CS2007}
F.~Cucker and S.~Smale.
\newblock Emergent behavior in flocks.
\newblock {\em IEEE T. Automat. Contr.}, 52:852--862, 2007.

\bibitem{degond2007CRM}
P.~Degond and S.~Motsch.
\newblock Macroscopic limit of self-driven particles with orientation
  interaction.
\newblock {\em C. R. Math. Acad. Sci. Paris}, 345(10):555--560, 2007.

\bibitem{degond2008M3AS}
P.~Degond and S.~Motsch.
\newblock Continuum limit of self-driven particles with orientation
  interaction.
\newblock {\em Math. Models Methods Appl. Sci.}, 18(suppl.):1193--1215, 2008.

\bibitem{eftimie2018}
R.~Eftimie.
\newblock {\em Hyperbolic and kinetic models for self-organised biological
  aggregations}, volume 2232 of {\em Lecture Notes in Mathematics}, chapter
  Multi-dimensional transport equations, pages 153--193.
\newblock Springer International Publishing, 2018.

\bibitem{falcone2013book}
M.~Falcone and R.~Ferretti.
\newblock {\em Semi-Lagrangian approximation schemes for linear and
  Hamilton-Jacobi equations}.
\newblock SIAM, 2013.

\bibitem{fornasier2010PhD}
M.~Fornasier, J.~Haskovec, and G.~Toscani.
\newblock Fluid dynamic description of flocking via {P}ovzner--{B}oltzmann
  equation.
\newblock {\em Phys. D}, 240:21--31, 2011.

\bibitem{leader1}
J.~Herbert-Read.
\newblock Collective behaviour: Leadership and learning in flocks.
\newblock {\em Curr. Biol.}, 25(23):R1127--R1129, 2015.

\bibitem{kolokolnikov2011PRE}
T.~Kolokolnikov, H.~Sun, D.~Uminsky, and A.~L. Bertozzi.
\newblock Stability of ring patterns arising from two-dimensional particle
  interactions.
\newblock {\em Phys. Rev. E}, 84:015203, 2011.

\bibitem{lacho1990ARMA}
M.~Lachowicz and M.~Pulvirenti.
\newblock A stochastic system of particles modelling the {E}uler equations.
\newblock {\em Arch. Rat. Mech. Anal.}, 109:81–93, 1990.

\bibitem{ling2019}
H.~Ling, G.~E. Mclvor, J.~Westley, K.~van~der Vaart, J.~Yin, R.~T. Vaughan,
  A.~Thornton, and N.~T. Ouellette.
\newblock Collective turns in jackdaw flocks: kinematics and information
  transfer.
\newblock {\em J. R. Soc. Interface}, 16:20190450, 2019.

\bibitem{loffredo2023PhB}
E.~Loffredo, D.~Venturelli, and I.~Giardina.
\newblock Collective response to local perturbations: how to evade threats
  without losing coherence.
\newblock {\em Phys. Biol.}, in press.

\bibitem{loy2021KRM}
N.~Loy and A.~Tosin.
\newblock {B}oltzmann-type equations for multi-agent systems with label
  switching.
\newblock {\em Kinet. Relat. Models}, 14(5):867--894, 2021.

\bibitem{markou2021pp}
I.~Markou.
\newblock Invariance of velocity angles and flocking in the {I}nertial {S}pin
  model.
\newblock Preprint, arXiv:2110.14388, 2021.

\bibitem{motsch2011JSP}
S.~Motsch and E.~Tadmor.
\newblock A new model for self-organized dynamics and its flocking behavior.
\newblock {\em J. Stat. Phys.}, 144:923, 2011.

\bibitem{mwaffo2018detecting}
V.~Mwaffo, J.~Keshavan, T.~L. Hedrick, and S.~Humbert.
\newblock Detecting intermittent switching leadership in coupled dynamical
  systems.
\newblock {\em Sci. Rep.}, 8:10338, 2018.

\bibitem{papadopoulou2022RSOS}
M.~Papadopoulou, H.~Hildenbrandt, D.~W.~E. Sankey, S.~J. Portugal, and C.~K.
  Hemelrijk.
\newblock Emergence of splits and collective turns in pigeon flocks under
  predation.
\newblock {\em R. Soc. Open Sci.}, 9(2):211898, 2022.

\bibitem{pareschi2013book}
L.~Pareschi and G.~Toscani.
\newblock {\em Interacting Multiagent Systems: Kinetic Equations and {M}onte
  {C}arlo Methods}.
\newblock Oxford University Press, Oxford, 2013.

\bibitem{leader2}
B.~Pettit, Z.~Ákos, T.~Vicsek, and D.~Biro.
\newblock Speed determines leadership and leadership determines learning during
  pigeon flocking.
\newblock {\em Curr. Biol.}, 25(23):3132--3137, 2015.

\bibitem{piccoliARMA2011}
B.~Piccoli and A.~Tosin.
\newblock Time-evolving measures and macroscopic modeling of pedestrian flow.
\newblock {\em Arch. Ration. Mech. Anal.}, 199:707--738, 2011.

\bibitem{pomeroy1992}
H.~Pomeroy and F.~Heppner.
\newblock Structure of turning in airborne {R}ock {D}ove (\emph{Columba Livia})
  flocks.
\newblock {\em The Auk}, 109(2):256--267, 1992.

\bibitem{povzner1962AMSTS}
A.~Y. Povzner.
\newblock The {B}oltzmann equation in kinetic theory of gases.
\newblock {\em Amer. Math. Soc. Transl. Ser. 2}, 47:193--216, 1962.

\bibitem{procaccini2011}
A.~Procaccini, A.~Orlandi, A.~Cavagna, I.~Giardina, F.~Zoratto, D.~Santucci,
  F.~Chiarotti, C.~K. Hemelrijk, E.~Alleva, G.~Parisi, and C.~Carere.
\newblock Propagating waves in starling, \textit{{S}turnus vulgaris}, flocks
  under predation.
\newblock {\em Anim. Behav.}, 82(4):759--765, 2011.

\bibitem{reynolds2022JRSI}
A.~M. Reynolds, G.~E. McIvor, A.~Thornton, P.~Yang, and N.~T. Ouellette.
\newblock Stochastic modelling of bird flocks: accounting for the cohesiveness
  of collective motion.
\newblock {\em J. R. Soc. Interface}, 19(189):20210745, 2022.

\bibitem{sumpter2006}
D.~J.~T. Sumpter.
\newblock The principles of collective animal behaviour.
\newblock {\em Phil. Trans. R. Soc. B}, 361:5--22, 2006.

\bibitem{toscani2004CHAPTER}
G.~Toscani.
\newblock Kinetic and hydrodynamic models of nearly elastic granular flows.
\newblock In A.~J\"{u}ngel, R.~Manasevich, P.~A. Markowich, and H.~Shahgholian,
  editors, {\em Nonlinear Differential Equation Models}, Wien, 2004.
  Springer-Verlag.

\bibitem{vicsek1995}
T.~Vicsek, A.~Czir\'ok, E.~Ben-Jacob, I.~Cohen, and O.~Shochet.
\newblock Novel type of phase transition in a system of self-driven particles.
\newblock {\em Phys. Rev. Lett.}, 75:1226--1229, 1995.

\bibitem{vicsek2012}
T.~Vicsek and A.~Zafeiris.
\newblock Collective motion.
\newblock {\em Phys. Rep.}, 517(3):71--140, 2012.

\end{thebibliography}

\end{document}